\documentclass[journal]{IEEEtran}
\usepackage{graphicx}
\usepackage{amsfonts} 
\usepackage{amsmath,paralist}
\usepackage{arydshln}
\usepackage{mleftright}
\usepackage{braket}
\usepackage{multirow}
\usepackage{blkarray}
\usepackage{enumitem}
\usepackage{color}
\usepackage{amsthm}
\usepackage{nccmath}
\usepackage{cite}
\usepackage{fltpoint}
\usepackage{enumitem}
\usepackage{soul}
\usepackage{caption}

\usepackage{cancel}

\usepackage{algcompatible}
\usepackage{algorithm}

\makeatletter
\newcommand{\algmargin}{\the\ALG@thistlm}
\algnewcommand{\PARSTATE}[1]
{\STATE\parbox[t]{\dimexpr\linewidth-\algmargin}
{#1\vspace{0.1cm}}}

\newcommand{\titlemath}[1]{\texorpdfstring{$#1$}{#1}}

\newcommand{\rotilde}{\rotatebox[origin=c]{270}{$\sim$}}

\usepackage{tikz}
\usetikzlibrary{decorations.pathreplacing}
\usepackage{physics}

\newtheorem{theorem}{Theorem}
\newtheorem{corollary}{Corollary}
\newtheorem{lemma}{Lemma}
\newtheorem{remark}{Remark}
\newtheorem{definition}{Definition}

\newcommand{\F}{\mathbb{F}}
\newcommand{\bl}[1]{{\color{blue}#1}}

\usepackage[hyphens]{url}
\usepackage[hidelinks,colorlinks=true,urlcolor= blue,linkcolor = blue,citecolor = red]{hyperref}

\begin{document}

\title{Theory of Communication Efficient \\Quantum Secret Sharing}
\author{
\IEEEauthorblockN{\ \\Kaushik Senthoor and Pradeep Kiran Sarvepalli\\}
\IEEEauthorblockA{Department of Electrical Engineering
\\Indian Institute of Technology Madras
\\Chennai 600 036, India}
}

\maketitle

\begin{abstract}
A $((k,n))$ quantum threshold secret sharing (QTS) scheme is a quantum cryptographic protocol for sharing a quantum secret among $n$ parties such that the secret can be recovered by any $k$ or more parties while $k-1$ or fewer parties have no information about the secret. Despite extensive research on these schemes, there has been very little study on optimizing the quantum communication cost during recovery. Recently, we initiated the study of communication efficient quantum threshold secret sharing (CE-QTS) schemes. These schemes reduce the communication complexity in QTS schemes by accessing $d>k$ parties for recovery; here $d$ is fixed ahead of encoding the secret. In contrast to the standard QTS schemes which require $k$ qudits for recovering each qudit in the secret, these schemes have a lower communication cost of $\frac{d}{d-k+1}$. In this paper, we further develop the theory of communication efficient quantum threshold schemes.
Here, we propose universal CE-QTS schemes which reduce the communication cost for all $d>k$ simultaneously.
We provide a framework based on ramp quantum secret sharing to construct CE-QTS and universal CE-QTS schemes. We give another construction for universal CE-QTS schemes based on Staircase codes. We derived a lower bound on communication complexity and show that our constructions are optimal. Finally, an information theoretic model is developed to analyse CE-QTS schemes and the lower bound on communication complexity is proved again using this model.
\end{abstract}

\begin{IEEEkeywords}
quantum secret sharing, communication complexity, quantum cryptography, threshold secret sharing schemes, Staircase codes.
\end{IEEEkeywords}

\section{Introduction}

%\subsection{Motivation.}

\IEEEPARstart{Q}{uantum} secret sharing schemes are  protocols that enable the secure distribution of  a secret among mutually collaborating parties so that only certain collections of parties can recover the
secret. 
\noindent
Quantum secret sharing schemes were first proposed by Hillery {\em et al.} for classical secrets \cite{hillery99}.
Subsequently, Cleve {\em et al.} proposed quantum secret sharing schemes for quantum secrets \cite{cleve99}.
Since these pioneering works, there has been extensive progress in this field, and it continues to be actively researched \cite{karlsson99,smith99,gottesman00,bandyopadhyay00,nascimento01,imai03,ps10,ben12,markham08,senthoor19,tyc02, qin20}. 
Quantum secret sharing has also been experimentally demonstrated by many groups \cite{tittel01,hao11,bogdanski08,bell14,schmid05,gaertner07,lance04, wei13, pinnell20}.
The progress has been rapid with demonstrations over distances as large as 50 km\cite{wei13}.
Furthermore, non-binary protocols over 11-dimensional qudits have also been demonstrated \cite{pinnell20}. 

Quantum secret sharing can be done under various settings: with classical data as the secret or an arbitrary quantum state as the secret, with parties having classical and quantum data (hybrid) or only quantum data, with or without pre-existing quantum entanglement shared among the parties, to name a few. 
Here, we consider the setting where the secret is an arbitrary quantum state, with all the parties having only quantum data and no pre-existing quantum entanglement.
In this paper, we are interested in optimizing the resources needed for quantum secret sharing. 
Specifically, we study the communication efficient threshold quantum secret sharing (CE-QTS) schemes and propose the improved model of universal CE-QTS schemes. 

The most popular quantum secret sharing scheme is the quantum threshold secret sharing scheme (QTS). 
In this scheme, out of the total $n$ parties, a minimum of $k$ parties are required to recover the secret. Also, here we look at only perfect QTS schemes, where any set of less than $k$ parties should not have any information on the secret. 
It is often denoted as a $((k,n))$ scheme. 
The state given to each party is called the share of the party.
After the secret has been shared, the parties who plan to recover the secret combine their
shares and reconstruct the secret. 
Alternatively, the parties involved in the recovery could communicate all or part of their share to a third party designated as the combiner.
This is the secret recovery model we focus in this paper.
The amount of quantum communication to the combiner for recovering the secret is called the communication complexity.
For sharing a secret of size $m$ qudits under this setting, a standard $((k,n))$ scheme (for example, \cite{cleve99}) requires $m n$ qudits to be shared for share distribution ($m$ qudits for each party) and 
at least $mk$ qudits for recovery. 
A slightly different model where one of the collaborating parties itself can act as the combiner and the remaining parties communicate their shares (in part or in full) to this party to recover the secret is also possible.
The definition of communication cost will be slightly different in this model.
However, constructions and bounds discussed in this paper can be adapted to such a model as well.

\subsection{Previous work}
The analogous problem of reducing communication complexity has been studied for classical secret sharing schemes \cite{wang08,bitar16,bitar18,huang16,huang17,penas18} but not as much in the quantum setting. Ref.~\cite{nascimento01} and \cite{ben12} aim to reduce the quantum communication during secret distribution to the parties but do not look at reducing the quantum communication cost during secret recovery. Only recently, \cite{senthoor19} showed that the quantum communication cost during secret recovery can 
be reduced by using a subset of $d$ parties whose cardinality is more than the threshold $k$ required to recover the secret. This scheme is called $((k,n,d))$ communication efficient quantum secret sharing (CE-QTS) scheme.
These gains can be significant and for a $((k,n=2k-1,d))$ threshold scheme, it was shown that the 
gains in communication complexity of recovery per secret qudit can be as large as $O(k)$. 
For sharing a secret of $m$ qudits, this scheme requires $mn$ qudits to be shared for secret distribution, $mk$ qudits for secret recovery when accessing $k$ parties and $dm/(d-k+1)$ qudits when accessing $d$ parties.
However, the improvement in communication cost only works for a fixed value of $d$ in the range of $k<d\leq n$. 
The value of $d$ is decided prior to encoding of the secret and cannot be changed.

\subsection{ Contributions}
In this paper, we develop the theory of communication efficient quantum secret sharing schemes. 
Specifically, we address the problem of designing quantum threshold schemes that are universal in the sense that any subset of parties of an arbitrary size greater than $k$ would provide further gains in communication cost during recovery. 
This is the first such class of universal communication efficient quantum threshold secret sharing schemes where the number of parties contacted for secret recovery can be varied from 
$k$ to $n$.

First, we give a framework for constructing CE-QTS schemes from a combination of ramp QSS schemes and threshold schemes.
We also propose a construction of CE-QTS schemes for both fixed $d$ and universal $d$ with this framework using the ramp secret sharing schemes proposed in \cite{ogawa05}.
This framework can also be used to derive other constructions for CE-QTS schemes by using different ramp QSS schemes.

Second, we propose a class of universal CE-QTS schemes based on the Staircase codes.
These schemes are inspired by the classical communication efficient secret sharing schemes of \cite{bitar16,bitar18}.
The constructions for these classical schemes are also related to codes for distributed storage aimed at reducing communication cost\cite{rashmi11}.

The constructions for universal CE-QTS schemes proposed in this paper, when an arbitrary $d\geq k$ number of parties are contacted, achieve the same communication complexity as that of fixed $d$.
So there is no penalty in communication complexity with the increased flexibility to change $d$. The universal CE-QTS constructions provide the same storage cost and communication cost (normalized to secret size) as the CE-QTS constructions. But the universal CE-QTS constructions need to have larger secret sizes to provide communication efficiency for various values of $d$. For a short summary of our constructions, refer Table \ref{tab:contributions}. 

Third, we derive lower bounds on the communication complexity of CE-QTS schemes (both fixed $d$ and universal).
We also propose an information theoretic model of CE-QTS schemes and prove that our constructions are optimal with respect to both share size and communication cost. 
The information theoretic model is used to give an alternative proof for the bound on communication cost.

Some preliminary results of this paper are discussed in the upcoming conference publication \cite{senthoor20}.
\begin{table*}[t]
\begin{center}
\begin{tabular}{|l|c|c|c|l|}
\hline & Number of parties & Secret size, $m$ & Communication & Dimension of
\\ & accessed by combiner, $d$ & & cost, CC$_n(d)/m$ & qudits, $q$ (prime)
\\\hline QTS\cite{gottesman00} & $d=k$, fixed & 1 & $k$ & $\geq 2k-1$
\\\hline CE-QTS (Staircase codes) \cite{senthoor19} & $k\leq d\leq n$, fixed & $d-k+1$ & $\frac{d}{d-k+1}$ & $>2k-1$
\\\hline CE-QTS (Concatenation, Corollary \ref{co:ramp-fixed-ceqts-constn}) & $k\leq d\leq n$, fixed & $d-k+1$ & $\frac{d}{d-k+1}$ & $>d+k-1$
\\\hline Universal CE-QTS (Staircase codes, Theorem \ref{th:staircase-constn-univ-ceqts}) & $k\leq d\leq n$, variable & lcm$\{1,2,\hdots,k\}$ & $\frac{d}{d-k+1}$ & $\geq 2(2k-1)$
\\\hline Universal CE-QTS (Concatenation, Corollary \ref{co:ramp-univ-ceqts-constn}) & $k\leq d\leq n$, variable & lcm$\{1,2,\hdots,n-k+1\}$ & $\frac{d}{d-k+1}$ & $>n+k-1$
\\\hline
\end{tabular}
\end{center}
\captionsetup{justification=justified}
\caption{Parameters of various $((k,n))$ QTS constructions. Here $2\leq k\leq n\leq 2k-1$.
For all these constructions, the individual share size is $m$ and CC$_n(k)/m=k$.
\label{tab:contributions}}
\end{table*}

\subsection{Organization}
We begin with a brief review of quantum secret sharing schemes in Section~\ref{sec:bg}.
Then we give a concrete illustration of the universal communication efficient quantum secret sharing schemes in Section~\ref{s:example}.
In Section \ref{s:framework}, we propose the Concatenation framework for constructing % $((k,n,d))$ 
CE-QTS schemes from ramp and threshold QSS schemes.
We also extend this framework to construct universal CE-QTS schemes.
In Section~\ref{s:iv}, we give a construction of universal CE-QTS schemes based on Staircase codes. 
We derive lower bounds on the communication complexity of CE-QTS schemes in Section \ref{s:v}.
In Section \ref{s:vi}, we propose an information theoretic model for studying CE-QTS schemes.
Finally, we conclude with a brief sketch of further directions of research.

\section{Background} \label{sec:bg}
\subsection{Notation}
Let $q$ be a prime and $\mathbb{F}_q$ denote a finite field with $q$ elements. 
We take the standard basis of $\mathbb{C}^q$ to be  $\{\ket{x}\mid x\in \F_q \}$.
We denote $\ket{x_1x_2\cdots x_\ell}$  by $\ket{\underline{x}}$ where $\underline{x}$ is the vector with the entries $(x_1,x_2,\hdots,x_\ell)$.
The standard basis for $\mathbb{C}^{q^n}$ is taken to be $\{\ket{\underline{x}}\mid \underline{x}\in \F_q^n \}$.
For any invertible matrix $K\in\F_q^{\ell\times\ell}$, we define the unitary operation $U_K$ 
\begin{eqnarray*}
U_K\ket{\underline{x}}  = \ket{K\underline{x}} =\ket{\underline{y}},%\label{eq:UK-defn} %eq:mtx_unitary
\end{eqnarray*}
where $\underline{y}= (y_1,\ldots, y_n)$ and $y_i = \sum_{j}K_{ij}x_j$.
We define the two qudit unitary operator $L_\alpha$ as 
\begin{eqnarray*}
L_\alpha\ket{i}_c\ket{j}_t  = \ket{i}_c\ket{j+\alpha i}_t, %\label{eq:Lt-defn}
\end{eqnarray*}
where  $i, j\in\F_q$ and $\alpha\in\F_q$ is a constant. The subscript $c$ and $t$ indicate that they are control and target qudits respectively. 
This operator generalizes the CNOT gate. 

We  use the notation $[n]:=\{1,2,\ldots, n \}$ and $[i,j]:=\{i, i+1,\ldots, j \}$.
Let $V$ be a $m\times n$ matrix  and  $A \subseteq [m]$, $B\subseteq [n]$.
We denote by $V_A$, the submatrix of 
$V$ formed by taking the rows indexed by entries in $A$.
Similarly, we can form a submatrix of $V$ by taking the columns of $V$. This is indicated as $V^B$. 
We can also form a submatrix $V_A^B$ of $V$ which takes rows indexed by $A$ and columns indexed by $B$.
For a matrix $V\in \F_q^{m\times n}$, the notation $\ket{V}$ indicates the state $\ket{v_{11}v_{21}\hdots v_{m1}}$$\ket{v_{12}v_{22}\hdots v_{m2}}$$\hdots$$\ket{v_{1n}v_{2n}\hdots v_{mn}}$ where $v_{ij}$ is the element of $V$ in $i$th row and $j$th column.
Let $A\in \F_q^{m\times n} $ matrix and $K$ is an invertible $m\times m $ matrix, then we can transform the state 
$\ket{A}$ to $\ket{KA}$ by the unitary operation $U_K^{\otimes n}$.
We refer to this operation as applying $K$ on $\ket{A}$ to obtain $\ket{KA}$.

\subsection{Quantum secret sharing (QSS)}
A quantum secret sharing scheme is a protocol to encode the secret in arbitrary quantum state and share it among $n$ parties such that certain subsets of parties, called authorized sets, can recover the secret (recoverability) and certain subsets of parties, called unauthorized sets, do not have any information about the secret (secrecy). The access structure $\Gamma$ of a QSS scheme is defined as
\begin{eqnarray*}
\Gamma=\{X\subseteq[n]:X\text{ is an authorized set}\}.
\end{eqnarray*}

A QSS scheme is called perfect quantum secret sharing scheme if any subset of the $n$ parties is either an authorized set or an unauthorized set and non-perfect otherwise. For non-perfect schemes, some subsets of the $n$ parties are allowed to have partial information about the secret. These sets are called intermediate sets.

A concrete realization of a quantum secret sharing scheme is specified by giving an encoding for the basis states of the secret. 
An encoding has to satisfy the properties of recoverability and secrecy to realize a QSS scheme.
\begin{definition}
A quantum secret sharing scheme for an access structure $\Gamma$ is the encoding and distribution of the secret in an arbitrary quantum state among $n$ parties such that
\begin{itemize}
\item (Recoverability) any authorized set $A\in\Gamma$ can recover the secret \textit{i.e.} there exists some recovery operation which can decode the secret from the shares in $A$,
\item (Secrecy) any unauthorized set $B\notin\Gamma$ has no information about the secret.
\end{itemize}
\label{de:qss}
\end{definition}
In a pure state QSS scheme, the encoding is such that the combined state of all shares is a pure state whenever the secret is in pure state. 
Otherwise, the scheme is called mixed state scheme.
\begin{lemma}[Mixed state schemes from pure state schemes]
\text{\cite[Theorem 3]{gottesman00}}
\label{lm:mixed-to-pure}
Any mixed state QSS scheme can be described as a pure state QSS scheme with one share discarded.
\end{lemma}
The no-cloning theorem implies that the complement of an authorized set is unauthorized set. 
In pure state schemes the converse also holds as given in the following result. 
\begin{lemma}[Authorized sets in pure state schemes]
\cite[Corollary 2]{gottesman00}
\label{lm:pu-auth}
In a pure state quantum secret sharing scheme, complement of any unauthorized set is an authorized set.
\end{lemma}
We use the following notation for parameters of QSS schemes: 
$q$ is the fixed dimension of all the qudits in the scheme, $m$ gives the size of the secret in qudits and $w_i$ gives the size of the $i$th share in qudits.

\subsection{Quantum threshold secret sharing (QTS)}
An important class of perfect quantum secret sharing schemes are the quantum threshold secret sharing schemes.
In threshold schemes, a set of parties is either authorized or unauthorized based on the number of parties in the set.
\begin{definition}[Quantum threshold scheme]
A $((k,n))$ quantum threshold secret sharing scheme for $1<k\leq n\leq 2k-1$ is a QSS scheme with $n$ parties where any $k$ or more parties can recover the secret, but $k-1$ or fewer parties have no information on the secret.
\end{definition}
If $n>2k$, then there exist two non-overlapping authorized sets which can give two copies of the secret thus violating no-cloning theorem.

Cleve {\em et al.} \cite{cleve99} have given a construction for $((k,n))$ QTS schemes as follows.
Consider the case of $n=2k-1$. Take $m=1$ and a prime $q\geq 2k-1$. 
The encoding for a basis state of the secret $s\in \mathbb{F}_q$ is given by the 
following superposition.
\begin{gather}
\ket{s}\ \mapsto\sum_{\underline{r}\in\mathbb{F}_q^{k-1}}\ket{v_1(\underline{r},s)}\ket{v_2(\underline{r},s)}\hdots\ket{v_n(\underline{r},s)}
\label{eq:qts-enc}
\end{gather}
Here $\underline{r}=(r_1,r_2,\hdots,r_{k-1})\in\mathbb{F}_q^{k-1}$ and $v_i(\underline{r},s)\in\mathbb{F}_q$ is the evaluation of the polynomial
\begin{eqnarray*}
v_i(\underline{r},s)=r_1+r_2x_i+\hdots+r_{k-1}x_i^{k-2}+sx_i^{k-1}.
\end{eqnarray*}
where $x_1,x_2,\hdots,x_n$ are distinct constants from $\mathbb{F}_q$.
Each of the $n$ parties is given one qudit from the encoded state.

For example, the encoding for a $((k=2,n=3))$ QTS scheme will be as follows where each qudit has dimension three.
\begin{eqnarray*}
\ket{s}\ \mapsto\sum_{r\in\mathbb{F}_3}&\ket{r}\ket{r+s}\ket{r+2s}
\end{eqnarray*}
To obtain a $((k,n))$ QTS scheme for $n<2k-1$, simply discard $2k-1-n$ shares after encoding the secret in the above scheme. 
\begin{lemma}
\text{\cite{cleve99}}
The encoding in \eqref{eq:qts-enc} provides a $q$-ary $((k,n))$ quantum threshold secret sharing scheme for $n\leq 2k-1$ with the following parameters.
\begin{gather*}
q\geq 2k-1\text{ (prime)}\\
m=1\\
w_1=w_2=\cdots=w_n=1
\end{gather*}
\label{lm:cleve-qts}
\end{lemma}
This scheme can be used to encode a secret of $m>1$ qudits by individually encoding each qudit in the secret.

\subsection{Storage and communication complexity}
The storage cost of a secret sharing scheme is directly related to the sizes of the shares. 
In this context the following result has been shown about the size of a share. 
\begin{lemma}[Share size, \cite{gottesman00}] \label{lm:qts-opt}
The size of each share in a threshold QSS scheme should be at least as large as the size of the secret.
\end{lemma}
Clearly, the QTS scheme in Lemma \ref{lm:cleve-qts} has optimal storage cost.
Apart from storage cost which depends on how the secret is encoded and distributed among the parties, it is also important to see how much quantum communication is needed during the secret recovery. 
There are two prominent approaches to reconstructing the secret. 
In the first approach, the parties from an authorized set could collaborate among themselves by means of nonlocal operations to recover the secret.
In the second approach, they can communicate all or part of their shares to a third party called the combiner.
In this paper, we focus on the latter method of secret reconstruction.

\begin{definition}[Communication cost for an authorized set]
The communication cost for an authorized set in a QSS scheme is the number of qudits sent to the combiner by the parties in that set for recovering the secret.
\end{definition}
For the same encoding of the secret, it is possible to have different recovery operations for a given authorized set, thus giving multiple values for the communication cost.
However the above definition for communication cost is defined for a particular recovery operation defined by the QSS scheme for an authorized set.
\begin{definition}[Communication cost for $d$ in QTS]
The communication cost for threshold $d\geq k$ in a $((k,n))$ quantum threshold secret sharing scheme is the maximum communication cost over all the authorized sets of size $d$. This will be denoted as $\text{CC}_n(d)$.
\end{definition}
Thus, for the QTS scheme defined in Lemma \ref{lm:cleve-qts}, the communication cost for secret recovery is CC$_n(k)=k$.

\subsection{Fixed \titlemath{d} communication efficient QTS (CE-QTS)}
Assume that the combiner in a QTS scheme has access to more than $k$ parties in the scheme. Then, the $((k,n))$ QTS scheme will still have the same communication cost of $k$ qudits.
However, by allowing each party in a $((k,n))$ QTS scheme to send only a part of its share to the combiner, it is possible to reduce this communication cost further.
\begin{definition}[CE-QTS]
A $((k,n))$ threshold secret sharing scheme is said to be communication efficient, if for some $d$ such that $k<d\leq n$,
\begin{equation}
\text{CC}_n(d)<\text{CC}_n(k)
\label{eq:ce-ineq}
\end{equation}
Such schemes are denoted as $((k, n, d))$ CE-QTS schemes.
\label{de:ce-qts}
\end{definition}
Here, $d$ is a fixed integer satisfying $k<d\leq n$.
{The strict inequality \eqref{eq:ce-ineq} in this definition is necessary because any $((k,n))$ scheme can allow recovery from $d$ parties by communicating some $k$ shares from these $d$ parties thus achieving CC$_n(d)=\text{CC}_n(k)$.}

This definition of $((k,n,d))$ CE-QTS schemes requires $d>k$. 
However, as we will see later on the bound on communication cost in Section \ref{s:v},
when only $d=k$ parties are accessed, the standard QTS schemes of \cite{cleve99} are optimal with respect to our bound on the communication complexity.
In addition, our constructions for CE-QTS in Sections \ref{s:framework} and \ref{s:iv} reduce to the standard QTS schemes.

A construction for $((k,n,d))$ CE-QTS schemes based on Staircase codes is given in \cite{senthoor19}.
For $n=2k-1$, this CE-QTS scheme is constructed as follows. The encoding for a basis state of the secret $\underline{s}=(s_1,s_2,\hdots,s_m)\in\mathbb{F}_q$ is given by the following superposition
\begin{eqnarray}
\ket{s_1 s_2\hdots s_m}\ \mapsto\sum_{\underline{r}\in\mathbb{F}_q^{m(k-1)}}
\bigotimes_{i=1}^{2k-1}\ket{c_{i1} c_{i2}\ldots c_{im}}
\label{eq:fixed-d-ce-qts-enc}
\end{eqnarray}
where $\underline{r}=(r_1,r_2,\hdots,r_{m(k-1)})\in\mathbb{F}_q^{m(k-1)}$ and $c_{ij}$ is the $(i,j)$th entry of the matrix
\begin{equation}
C=VY.\nonumber
\end{equation}
Here, $V$ is a Vandermonde matrix defined as
\begin{eqnarray}
V=\left[\begin{array}{cccc} 1 & x_1 & \ldots & x_1^{d-1}\\
1 & x_2 & \ldots & x_2^{d-1}\\
\vdots & \vdots& \ddots & \vdots\\
1 & x_n & \ldots & x_n^{d-1}
\end{array} \right].\nonumber
\end{eqnarray}
where $x_1, x_2,\hdots, x_n$ are distinct non-zero constants from $\F_q$. The matrix $Y$ is given by
{\small
\begin{eqnarray}
Y=\left[
    \begin{array}{c:c}
        \begin{matrix} s_1 \\ s_2 \\ \vdots \end{matrix} & \text{\huge 0}_{(m-1)\times(m-1)} \\
        \cdashline{2-2}[4pt/4pt]
        s_m & \begin{matrix} \hspace{-0.25in}r_{k-m+1} & r_{k-m+2} & \ \hdots & \ \ \ \ \ r_{k-1} \end{matrix} \\
        \cdashline{1-2}[4pt/4pt]
        \begin{matrix} r_1\\ r_2\\ \vdots\\ r_{k-1} \end{matrix} & \begin{matrix} \ r_k            & r_{2(k-1)+1} &  \hdots & r_{(m-1)(k-1)+1} \\
                                                                                                                      \ r_{k+1}      & r_{2(k-1)+2} &  \hdots & r_{(m-1)(k-1)+2} \\
                                                                                                                      \ \vdots        & \vdots          &  \ddots & \vdots \\
                                                                                                                      \ r_{2(k-1)} &  r_{3(k-1)}   &  \hdots & r_{m(k-1)}\end{matrix}
\end{array}
\right].
%\label{eq:msgMatrix}
\nonumber
\end{eqnarray}
}
After encoding, the first set of $m$ qudits are given to the first party, the second set of $m$ qudits given to the second party and so on till the $n$th party.
When the combiner accesses $k$ parties, each of these $k$ parties sends all its $m=d-k+1$ qudits. When the combiner accesses $d$ parties, each of these $d$ parties sends only its first qudit.

\begin{lemma}
\text{\cite{senthoor19}}
The encoding in \eqref{eq:fixed-d-ce-qts-enc} provides a $q$-ary $((k,n,d))$ communication efficient quantum threshold secret sharing scheme with the following parameters
\begin{gather*}
q>2k-1\text{ (prime)}
\\m=d-k+1
\\w_1=w_2=\hdots=w_n=d-k+1
\\\text{CC}_n(k)=k(d-k+1)
\\\text{CC}_n(d)=d.
\end{gather*}
\label{lm:senthoor-ce-qts}
\end{lemma}

To obtain a $((k,n,d))$ CE-QTS scheme for $n<2k-1$, simply discard $2k-1-n$ shares after encoding the secret in the above scheme. By Lemma \ref{lm:qts-opt}, this scheme has an optimal storage cost. It is also proved in \cite{senthoor19} that this scheme gives an optimal communication cost when the combiner accesses $d$ parties, for the specific case of $n=2k-1$. In this paper, we prove that optimality of this scheme holds for $n<2k-1$ as well.

For example, for $k=3, d=5$, this construction gives a $((3,5,5))$ CE-QTS scheme with the parameters
\begin{subequations}
\label{eq:eg-senthoor-ce-qts}
\begin{gather}
q=7\\
m=3\\
w_1=w_2=\hdots=w_5=3\\
\text{CC}_n(3)=9,\ \text{CC}_n(5)=5.
\end{gather}
\end{subequations}
The matrices $V$ and $Y$ in this scheme are given by
\begin{equation*}
V=
\begin{bmatrix}
1&1&1&1&1\\1&2&4&1&2\\1&3&2&6&4\\1&4&2&1&4\\1&5&4&6&2
\end{bmatrix}
\text{and\ }
Y=
\left[
\begin{tabular}{ccc}
$s_1$&0&0\\$s_2$&0&0\\$s_3$&$r_1$&$r_2$\\$r_1$&$r_3$&$r_5$\\$r_2$&$r_4$&$r_6$
\end{tabular}
\right].
\end{equation*}
The encoding for the scheme is given by the following mapping
\begin{eqnarray}
\label{eq:enc_qudits_3_5}
\ket{\underline{s}}\mapsto\sum_{\underline{r}\in\F_7^6}\ket{c_{11}c_{12}c_{13}}&&\!\!\ket{c_{21}c_{22}c_{23}}\ket{c_{31}c_{32}c_{33}}
\\[-0.5cm]&&\ \ \ \ \ket{c_{41}c_{42}c_{43}}\ket{c_{51}c_{52}c_{53}}\nonumber
\end{eqnarray}
where $\underline{s}=(s_1,s_2,s_3)$ indicates a basis state of the quantum secret, $\underline{r}=(r_1,r_2,\hdots,r_6)$
and 
$c_{ij} $ is the $(i,j)$th entry of the matrix
\begin{equation}
C=VY.\nonumber
\end{equation}
The encoded state in \eqref{eq:enc_qudits_3_5} can also be written as,
\begin{align*}
\sum_{\underline{r}\in\mathbb{F}_7^6}
\begin{array}{l}
\ket{v_1(\underline{s},r_1,r_2)}\ket{v_1(0,0,r_1,r_3,r_4)}
\ket{v_1(0,0,r_2,r_5,r_6)}
\\\ket{v_2(\underline{s},r_1,r_2)}\ket{v_2(0,0,r_1,r_3,r_4)}
\ket{v_2(0,0,r_2,r_5,r_6)}
\\\ket{v_3(\underline{s},r_1,r_2)}\ket{v_3(0,0,r_1,r_3,r_4)}
\ket{v_3(0,0,r_2,r_5,r_6)}
\\\ket{v_4(\underline{s},r_1,r_2)}\ket{v_4(0,0,r_1,r_3,r_4)}
\ket{v_4(0,0,r_2,r_5,r_6)}
\\\ket{v_5(\underline{s},r_1,r_2)}\ket{v_5(0,0,r_1,r_3,r_4)}
\ket{v_5(0,0,r_2,r_5,r_6)}.
\end{array}
\end{align*}
$v_i()$ indicates the polynomial evaluation given by
\begin{eqnarray*}
v_i(f_1,f_2,f_3,f_4,f_5)&=&f_1+f_2.x_i+f_3.x_i^2+f_4.x_i^3+f_5.x_i^4
\end{eqnarray*}
where the expression $v_i(\underline{s},r_1,r_2)$ denotes $v_i(s_1,s_2,s_3,r_1,r_2)$. Here we have taken $x_i=i$ for $1\leq i\leq 5$.
From this encoded state, the first party gets the three qudits from the first row, the second user gets the three qudits from the second row and so on till the fifth party.

When combiner requests $k=3$ parties, each party sends its complete share. 
When $d=5$, the combiner downloads the first qudit of each share from all the five parties. The secret recovery for this scheme is explained in detail in Appendix \ref{ap:ceqts-full-eg}. 

\subsection{Ramp quantum secret sharing (RQSS)}
The QTS scheme defined earlier is a perfect QSS scheme \text{i.e.} any set of parties is either authorized or unauthorized. But it is also possible to design a non-perfect threshold scheme such that a set of parties may be neither authorized nor unauthorized. A generalization of the threshold schemes leads to the ramp quantum secret sharing.

\begin{definition}[Ramp secret sharing schemes]
A $((t,n;z))$ ramp quantum secret sharing scheme for $1\leq z<t\leq n\leq t+z$ is a QSS scheme with $n$ parties where any $t$ or more parties can recover the secret, but $z$ or fewer parties have no information on the secret.
\end{definition}
{Note that the notation for RQSS schemes should not be confused with that of CE-QTS schemes.}

When $z=t-1$, then the ramp scheme is identical to a $((t,n))$ perfect threshold scheme.
For $z<t-1$, there are intermediate sets (of size $z+1$ to $t-1$) which are not able to reconstruct the secret but can have partial information about the secret.

Ogawa {\em et al.} \cite{ogawa05} provided a construction for $((t,n;z))$ ramp QSS schemes for $n\leq t+z$ as follows. Consider the case of $n=t+z$. Take $m=t-z$ and a prime $q>t+z$. The encoding for the basis state of the secret $\underline{s}=(s_1,s_2,\hdots,s_m)\in\mathbb{F}_q^m$ is given by the superposition
\begin{eqnarray}
\ket{s_1 s_2\hdots s_m}\mapsto\sum_{\underline{r}}\ket{u_1(\underline{s},\underline{r}),u_2(\underline{s},\underline{r}),\hdots,u_n(\underline{s},\underline{r})}.\ \ \ 
\label{eq:rqss-enc}
\end{eqnarray}
Here $\underline{r}=(r_1,r_2,\hdots,r_z)\in\mathbb{F}_q^z$ and $u_i(\underline{s},\underline{r})$ is the polynomial evaluation
\begin{eqnarray}
u_i(\underline{s},\underline{r})=s_1&+&s_2 x_i+\hdots+s_m x_i^{m-1}\nonumber
\\&&+r_1 x_i^m+r_2 x_i^{m+1}+\hdots+r_z x_i^{t-1}\ \ \ \nonumber
\end{eqnarray}
where $x_1,x_2,\hdots,x_n$ are distinct non-zero constants from $\mathbb{F}_q$.
\begin{remark}
A $((t,n;z))$ ramp QSS scheme can be obtained from a $((t,n+\ell;z))$ scheme by simply dropping some $\ell$ shares.
\label{re:ramp-by-dropping}
\end{remark}
Thus, this construction gives $((t,n;z))$ ramp schemes for any $n\leq t+z$.
For example, an encoding for a $((t=3,n=4;z=1))$ ramp QSS scheme will be as follows where each qudit has dimension 5.
\begin{eqnarray}
\ket{s_1 s_2}\ \mapsto\sum_{r_1\in\mathbb{F}_5}&\ket{s_1+s_2+r_1}\ket{s_1+2s_2+4r_1}\ \ \ \ &\nonumber
\\[-0.4cm]&\ \ \ \ \ket{s_1+3s_2+4r_1}\ket{s_1+4s_2+r_1}&\nonumber
\end{eqnarray}
Each party is given one of the qudits from the encoded state.
\begin{lemma}
\text{\cite{ogawa05}}
The encoding in \eqref{eq:rqss-enc} provides a $q$-ary $((t,n;z))$ ramp quantum secret sharing scheme for $z<t, n\leq t+z$ with the following parameters
\label{lm:ogawa-ramp}
\begin{gather*}
q>t+z\text{ (prime)}\\
m=t-z\\
w_1=w_2=\hdots=w_n=1.
\end{gather*}
\end{lemma}
This scheme can be used to encode a secret of $m=\ell(t-z)$ qudits by individually encoding every set of $t-z$ qudits in the secret.
For $t=k, z=k-1$, this scheme is very similar to the $((k,n))$ QTS scheme in Lemma \ref{lm:cleve-qts}.
\begin{lemma}
\label{lm:rqss-opt}
\text{\cite[Corollary 2]{ogawa05}}
The share size averaged over all parties in a $((t,n;z))$ ramp QSS scheme should be at least as large as $\frac{1}{t-z}$ times the size of the secret.
\end{lemma}
Note that the bound on storage cost in ramp QSS is in terms of average share size rather than individual share size. Clearly, the RQSS scheme from Lemma \ref{lm:ogawa-ramp} achieves this bound.

\subsection{Quantum information theory}
We briefly recall some of the terms of quantum information theory and introduce the notation used in the paper.
For further reading, we refer the reader to \cite{nielsen00}.

The von Neumann entropy of a quantum system $A$ with density matrix $\rho_A$ is given by 
\begin{equation}
\mathsf{S}(A)=-\tr(\rho_A\ \text{log}\ \rho_A)=-\sum_{i=1}^{M_A}\lambda_i\ \text{log}\ \lambda_i.\nonumber
\end{equation}
Here $\{\lambda_i\}$ are the eigenvalues of $\rho_A$ acting on a Hilbert space $\mathcal{H}_A$ of dimension $M_A$. The maximum value for $\mathsf{S}(A)$ is given by
\begin{equation}
\mathsf{S}(A)\leq\log M_A.
\label{eq:max-entropy}
\end{equation}

Consider the bipartite quantum system $AB$ whose density matrix $\rho_{AB}$ over the Hilbert space $\mathcal{H}_A\otimes\mathcal{H}_B$. 
Joint quantum entropy of $AB$ is defined as
\begin{equation}
\mathsf{S}(AB)=-\tr(\rho_{AB}\ \text{log}\ \rho_{AB}).\nonumber
\end{equation}
It satisfies two important properties.
\begin{eqnarray}
\mathsf{S}(AB)\leq \mathsf{S}(A)+\mathsf{S}(B)
\label{eq:sub-addi}
\\\mathsf{S}(AB)\geq|\mathsf{S}(A)-\mathsf{S}(B)|
\label{eq:araki-lieb}
\end{eqnarray}
The property \eqref{eq:sub-addi} is called subadditivity and \eqref{eq:araki-lieb} is called the Araki-Lieb inequality.

Mutual information between two quantum systems $A$ and $B$ is defined as
\begin{equation}
I(A:B)=\mathsf{S}(A)+\mathsf{S}(B)-\mathsf{S}(AB).\nonumber
\end{equation}

Consider a quantum system $\mathcal{Q}$ defined over the Hilbert space $\mathcal{H}_\mathcal{Q}$ of dimension $N$. Then the density matrix corresponding to $\mathcal{Q}$ can be defined as
\begin{equation}
\rho_\mathcal{Q}=\sum_{i=0}^{N-1}\lambda_i\ketbra{\phi_i}{\phi_i}
\nonumber
\end{equation}
where $\{\lambda_i\}$ gives the probability distribution in a measurement over some basis of orthonormal states $\{\ket{\phi_0},\ket{\phi_1},\hdots,\ket{\phi_{N-1}}\}$. Let $\mathcal{R}$ be the reference system such that the combined system $\mathcal{R}Q$ is in the pure state
\begin{equation}
\ket{\Phi_{\mathcal{R}Q}}=\sum_{i=0}^{N-1}\sqrt{\lambda_i}\ket{\phi_i}_\mathcal{R}\ket{\phi_i}_Q.
\nonumber
\end{equation}
Let $\mathcal{W}$ be a quantum operation which takes the state $Q$ to the new state $Q'$. Then the necessary and sufficient condition for the existence of a quantum operation $\mathcal{Y}$ which can recover the state $Q$ from $Q'$ is given by the condition
\begin{equation}
I(\mathcal{R}:Q')=I(\mathcal{R}:Q).
\label{eq:qdpi}
\end{equation}
This result is due to the quantum data processing inequality given in the following lemma.
\begin{lemma}[Quantum data processing inequality \cite{schumacher96}]
Consider an arbitrary quantum state $Q$ with a reference system $\mathcal{R}$ such that $Q\mathcal{R}$ is in pure state. If $\mathcal{W}$ is a quantum operation which takes state $Q$ to $Q'$, then
\begin{eqnarray}
\mathsf{S}(Q)\geq\mathsf{S}(Q')-\mathsf{S}(\mathcal{R}Q')
\nonumber
\end{eqnarray}
with equality achieved if and only if the original state $Q$ can be completely recovered from $Q'$.
\label{lm:qdpi}
\end{lemma}
\section{Universal CE-QTS: A First Look}\label{s:example}
In this section, we take the first steps for a formal treatment of universal communication efficient quantum threshold schemes. 
After defining them, we illustrate the gains in communication complexity for a suitably designed quantum threshold scheme. 
Later sections in this paper provide constructions for such universal communication efficient quantum secret sharing schemes.
\begin{definition}[Universal CE-QTS]
A $((k,n))$ threshold secret sharing scheme is said to be universal communication efficient, if for any  $d_i$ and $d_\ell$ such that $k\leq d_i<d_\ell\leq n$, $\text{CC}(d_\ell)<\text{CC}(d_i)$. Such schemes are denoted as $((k,n,*))$ universal CE-QTS schemes
\end{definition}
In other words, in universal CE-QTS schemes, $\text{CC}_n(n)<\text{CC}_n(n-1)<\hdots<\text{CC}_n(k+1)<\text{CC}_n(k)$.
{Similar to Definition \ref{de:ce-qts}, this definition also requires strict reduction in communication cost CC$_n(d)$ for increasing values of $d$.}
\subsection{An example for universal CE-QTS}
Consider the example of $((k=3,n=5,*))$ universal CE-QTS scheme with the following parameters.
\begin{subequations}
\begin{gather}
q=11
\\m=3
\\w_1=w_2=\hdots=w_5=3
\\\text{CC}_5(3)=9,\ \text{CC}_5(4)=8,\ \text{CC}_5(5)=5.
\end{gather}
\end{subequations}
The encoding for the scheme is given by the following mapping
\begin{eqnarray}
\label{eq:enc_qudits_3_5_s}
\ket{\underline{s}}\mapsto\sum_{\underline{r}\in\F_{11}^6}\ket{c_{11}c_{12}c_{13}}&&\!\!\ket{c_{21}c_{22}c_{23}}\ket{c_{31}c_{32}c_{33}}
\\[-0.5cm]&&\ \ \ \ \ket{c_{41}c_{42}c_{43}}\ket{c_{51}c_{52}c_{53}}\nonumber
\end{eqnarray}
where $\underline{s}=(s_1,s_2,s_3)\in\mathbb{F}_{11}^3$ indicates a basis state of the quantum secret, $\underline{r}=(r_1,r_2,\hdots,r_6)\in\mathbb{F}_{11}^6$
and 
$c_{ij} $ is the $(i,j)$th entry of the matrix
\begin{equation*}
C=VY.
\end{equation*}
Here the matrices $V$ and $Y$ are defined as follows. 
\begin{eqnarray*}
V=
\begin{bmatrix}
9&3&4&6&1\\2&9&3&4&6\\8&2&9&3&4\\7&8&2&9&3\\5&7&8&2&9
\end{bmatrix}
\text{\ \ and\ \ \ }
Y=
\left[
\begin{tabular}{ccc}
$s_1$&0&0\\$s_2$&$r_1$&0\\$s_3$&$r_2$&$r_3$\\$r_1$&$r_3$&$r_5$\\$r_2$&$r_4$&$r_6$
\end{tabular}
\right].
\end{eqnarray*}
The encoded state in \eqref{eq:enc_qudits_3_5_s} can also be written as,
\begin{align*}
\sum_{\underline{r}\in\mathbb{F}_{11}^6}
\begin{array}{l}
\ket{v_1(\underline{s},r_1,r_2)}\ket{v_1(0,r_1,r_2,r_3,r_4)}
\ket{v_1(0,0,r_3,r_5,r_6)}
\\\ket{v_2(\underline{s},r_1,r_2)}\ket{v_2(0,r_1,r_2,r_3,r_4)}
\ket{v_2(0,0,r_3,r_5,r_6)}
\\\ket{v_3(\underline{s},r_1,r_2)}\ket{v_3(0,r_1,r_2,r_3,r_4)}
\ket{v_3(0,0,r_3,r_5,r_6)}
\\\ket{v_4(\underline{s},r_1,r_2)}\ket{v_4(0,r_1,r_2,r_3,r_4)}
\ket{v_4(0,0,r_3,r_5,r_6)}
\\\ket{v_5(\underline{s},r_1,r_2)}\ket{v_5(0,r_1,r_2,r_3,r_4)}
\ket{v_5(0,0,r_3,r_5,r_6)}.
\end{array}
\end{align*}
Here $v_i()$ indicates the expression
\begin{eqnarray*}
v_i(f_1,f_2,f_3,f_4,f_5)=v_{i1}f_1+v_{i2}f_2+v_{i3}f_3+v_{i4}f_4+v_{i5}f_5
\end{eqnarray*}
where $v_{ij}=[V]_{ij}$ and the expression $v_i(\underline{s},r_1,r_2)$ denotes $v_i(s_1,s_2,s_3,r_1,r_2)$.
The matrix $V$ is a Cauchy matrix.
From this encoded state, the first party gets the three qudits from the first row, the second user gets the three qudits from the second row and so on till the fifth party.

When combiner requests $d=5$ parties, they send the first qudit from each of their shares. 
When $d=4$, the combiner downloads the first two qudits of each share of the four parties contacted. 
When $d=3$, the combiner downloads all three qudits of the share of the three parties contacted.
(For clarity, the qudits accessible to the combiner have been highlighted in blue in the description below.)

Consider the case when $d=5$ \textit{i.e.} the first qudits from all five parties are accessed.
\begin{align*}
\sum_{\underline{r}\in\mathbb{F}_{11}^6}
\begin{array}{l}
\bl{\ket{v_1(\underline{s},r_1,r_2)}}\ket{v_1(0,r_1,r_2,r_3,r_4)}
\ket{v_1(0,0,r_3,r_5,r_6)}
\\\bl{\ket{v_2(\underline{s},r_1,r_2)}}\ket{v_2(0,r_1,r_2,r_3,r_4)}
\ket{v_2(0,0,r_3,r_5,r_6)}
\\\bl{\ket{v_3(\underline{s},r_1,r_2)}}\ket{v_3(0,r_1,r_2,r_3,r_4)}
\ket{v_3(0,0,r_3,r_5,r_6)}
\\\bl{\ket{v_4(\underline{s},r_1,r_2)}}\ket{v_4(0,r_1,r_2,r_3,r_4)}
\ket{v_4(0,0,r_3,r_5,r_6)}
\\\bl{\ket{v_5(\underline{s},r_1,r_2)}}\ket{v_5(0,r_1,r_2,r_3,r_4)}
\ket{v_5(0,0,r_3,r_5,r_6)}
\end{array}
\end{align*}
Applying the operation $U_{V^{-1}}$ on these five qudits, we obtain
\begin{align*}
\bl{\ket{\underline{s}}}\sum_{\underline{r}\in\mathbb{F}_{11}^6}
\begin{array}{l}
\ket{v_1(0,r_1,r_2,r_3,r_4)}
\ket{v_1(0,0,r_3,r_5,r_6)}
\\\ket{v_2(0,r_1,r_2,r_3,r_4)}
\ket{v_2(0,0,r_3,r_5,r_6)}
\\\ket{v_3(0,r_1,r_2,r_3,r_4)}
\ket{v_3(0,0,r_3,r_5,r_6)}
\\\bl{\ket{r_1}}\ket{v_4(0,r_1,r_2,r_3,r_4)}
\ket{v_4(0,0,r_3,r_5,r_6)}
\\\bl{\ket{r_2}}\ket{v_5(0,r_1,r_2,r_3,r_4)}
\ket{v_5(0,0,r_3,r_5,r_6)}
\end{array}
\end{align*}
Here, the three qudits containing the basis state of the secret are not entangled with any of the other qudits. Thus, any arbitrary superposition of the basis states can be recovered with the above step.

Consider the case when $d=4$. Assume that the first four parties are accessed. The first two qudits from the four parties are sent to the combiner.
\begin{align*}
\sum_{\underline{r}\in\mathbb{F}_{11}^6}
\begin{array}{l}
\bl{\ket{v_1(\underline{s},r_1,r_2)}\ket{v_1(0,r_1,r_2,r_3,r_4)}}
\ket{v_1(0,0,r_3,r_5,r_6)}
\\\bl{\ket{v_2(\underline{s},r_1,r_2)}\ket{v_2(0,r_1,r_2,r_3,r_4)}}
\ket{v_2(0,0,r_3,r_5,r_6)}
\\\bl{\ket{v_3(\underline{s},r_1,r_2)}\ket{v_3(0,r_1,r_2,r_3,r_4)}}
\ket{v_3(0,0,r_3,r_5,r_6)}
\\\bl{\ket{v_4(\underline{s},r_1,r_2)}\ket{v_4(0,r_1,r_2,r_3,r_4)}}
\ket{v_4(0,0,r_3,r_5,r_6)}
\\\ket{v_5(\underline{s},r_1,r_2)}\ket{v_5(0,r_1,r_2,r_3,r_4)}
\ket{v_5(0,0,r_3,r_5,r_6)}
\end{array}
\end{align*}
Applying the operation $U_{K_1}$ on the set of four second qudits, where $K_1$ is the inverse of $V_{[4]}^{[2,5]}$, we obtain
\begin{eqnarray}
\sum_{\underline{r}\in\mathbb{F}_{11}^6}
\begin{array}{l}
\bl{\ket{v_1(\underline{s},r_1,r_2)}\ket{r_1}}
\ket{v_1(0,0,r_3,r_5,r_6)}
\\\bl{\ket{v_2(\underline{s},r_1,r_2)}\ket{r_2}}
\ket{v_2(0,0,r_3,r_5,r_6)}
\\\bl{\ket{v_3(\underline{s},r_1,r_2)}\ket{r_3}}
\ket{v_3(0,0,r_3,r_5,r_6)}
\\\bl{\ket{v_4(\underline{s},r_1,r_2)}\ket{r_4}}
\ket{v_4(0,0,r_3,r_5,r_6)}
\\\ket{v_5(\underline{s},r_1,r_2)}\ket{v_5(0,r_1,r_2,r_3,r_4)}
\ket{v_5(0,0,r_3,r_5,r_6)}.
\end{array}
%\label{eq:d4-recovery-1}
\nonumber
\end{eqnarray}
Then, on applying the operators $L_{10}\ket{r_2}\ket{v_1(\underline{s},r_1,r_2)}$, $L_5\ket{r_2}\ket{v_2(\underline{s},r_1,r_2)}$, $L_7\ket{r_2}\ket{v_3(\underline{s},r_1,r_2)}$ and
$L_8\ket{r_2}$ $\ket{v_4(\underline{s},r_1,r_2)}$,
we obtain
\begin{eqnarray*}
\sum_{\underline{r}\in\mathbb{F}_{11}^6}
\begin{array}{l}
\bl{\ket{v_1(\underline{s},r_1,0)}\ket{r_1}}
\ket{v_1(0,0,r_3,r_5,r_6)}
\\\bl{\ket{v_2(\underline{s},r_1,0)}\ket{r_2}}
\ket{v_2(0,0,r_3,r_5,r_6)}
\\\bl{\ket{v_3(\underline{s},r_1,0)}\ket{r_3}}
\ket{v_3(0,0,r_3,r_5,r_6)}
\\\bl{\ket{v_4(\underline{s},r_1,0)}\ket{r_4}}
\ket{v_4(0,0,r_3,r_5,r_6)}
\\\ket{v_5(\underline{s},r_1,r_2)}\ket{v_5(0,r_1,r_2,r_3,r_4)}
\ket{v_5(0,0,r_3,r_5,r_6)}.
\end{array}
\end{eqnarray*}
Applying the operation $U_{K_2}$ on the set of four first qudits, where $K_2$ is the inverse of $V_{[4]}^{[4]}$, we obtain the following state.
\begin{eqnarray*}
\bl{\ket{\underline{s}}}
\sum_{\underline{r}\in\mathbb{F}_{11}^6}
\hspace{-0.1cm}
\begin{array}{l}
\bl{\ket{r_1}}
\ket{v_1(0,0,r_3,r_5,r_6)}
\\\bl{\ket{r_2}}
\ket{v_2(0,0,r_3,r_5,r_6)}
\\\bl{\ket{r_3}}
\ket{v_3(0,0,r_3,r_5,r_6)}
\\\bl{\ket{r_1}\ket{r_4}}
\ket{v_4(0,0,r_3,r_5,r_6)}
\\\ket{v_5(\underline{s},r_1,r_2)}\ket{v_5(0,r_1,r_2,r_3,r_4)}
\ket{v_5(0,0,r_3,r_5,r_6)}
\end{array}\nonumber
\end{eqnarray*}
\vspace{-0.3cm}
\begin{equation}
\label{eq:entangled_secret}
\vspace{-0.1cm}
\end{equation}
We disentangle the basis state $\ket{\underline{s}}$ from the rest of qudits by applying the operator $U_{K_3}$ on $\ket{r_1}\ket{r_2}\ket{r_3}\ket{r_4}$ to get $\ket{r_1}\ket{r_2}\ket{r_3}\ket{v_5(0,r_1,r_2,r_3,r_4)}$ and then applying $U_{K_4}$ on $\ket{s_1}\ket{s_2}\ket{s_3}\ket{r_1}\ket{r_2}$ to get $\ket{s_1}\ket{s_2}\ket{s_3}\ket{r_1}\ket{v_5(\underline{s},r_1,r_2)}$.
\begin{equation}
K_3=\left[
\begin{tabular}{c}
1 0 0 0\\
0 1 0 0\\
0 0 1 0\\\hline
$V_{\{5\}}^{[2,5]}$
\end{tabular}
\right]
\text{\ and\ \ }
K_4=\left[
\begin{tabular}{c}
1 0 0 0 0\\
0 1 0 0 0\\
0 0 1 0 0\\
0 0 0 1 0\\\hline
$V_{\{5\}}$
\end{tabular}
\right]\nonumber
\end{equation}
Now, we obtain
\begin{eqnarray*}
\hspace{-0.1cm}
\bl{\ket{\underline{s}}}
\sum_{\underline{r}\in\mathbb{F}_{11}^6}
\hspace{-0.1cm}
\begin{array}{l}
\bl{\ket{r_1}}
\ket{v_1(0,0,r_3,r_5,r_6)}
\\\bl{\ket{v_5(\underline{s},r_1,r_2)}}
\ket{v_2(0,0,r_3,r_5,r_6)}
\\\bl{\ket{r_3}}
\ket{v_3(0,0,r_3,r_5,r_6)}
\\\bl{\ket{r_1}\ket{v_5(0,r_1,r_2,r_3,r_4)}}
\ket{v_4(0,0,r_3,r_5,r_6)}
\\\ket{v_5(\underline{s},r_1,r_2)}\ket{v_5(0,r_1,r_2,r_3,r_4)}
\ket{v_5(0,0,r_3,r_5,r_6)}
\end{array}
\end{eqnarray*}
\vspace{-0.8\baselineskip}
\begin{eqnarray}
&&\hspace{-0.6cm}=\bl{\ket{\underline{s}}}
\sum_{\substack{(r_1,r_2,r_3,r_4',\\r_5,r_6)\in\mathbb{F}_{11}^6}}
\begin{array}{l}
\bl{\ket{r_1}}
\ket{v_1(0,0,r_3,r_5,r_6)}
\\\bl{\ket{v_5(\underline{s},r_1,r_2)}}
\ket{v_2(0,0,r_3,r_5,r_6)}
\\\bl{\ket{r_3}}
\ket{v_3(0,0,r_3,r_5,r_6)}
\\\bl{\ket{r_1}\ket{r_4'}}
\ket{v_4(0,0,r_3,r_5,r_6)}
\\\ket{v_5(\underline{s},r_1,r_2)}\ket{r_4'}
\ket{v_5(0,0,r_3,r_5,r_6)}
\end{array}
\label{eq:disentangled_secret_before}
\\&&\hspace{-0.6cm}=\bl{\ket{\underline{s}}}
\sum_{\substack{(r_1,r_2',r_3,r_4',\\r_5,r_6)\in\mathbb{F}_{11}^6}}
\begin{array}{l}
\bl{\ket{r_1}}
\ket{v_1(0,0,r_3,r_5,r_6)}
\\\bl{\ket{r_2'}}
\ket{v_2(0,0,r_3,r_5,r_6)}
\\\bl{\ket{r_3}}
\ket{v_3(0,0,r_3,r_5,r_6)}
\\\bl{\ket{r_1}\ket{r_4'}}
\ket{v_4(0,0,r_3,r_5,r_6)}
\\\ket{r_2'}\ket{r_4'}
\ket{v_5(0,0,r_3,r_5,r_6)}.
\end{array}
\label{eq:disentangled_secret}
\end{eqnarray}
The variable change in \eqref{eq:disentangled_secret_before} is possible because the qudits
$\sum_{r_4\in\mathbb{F}_{11}}\bl{\ket{v_5(0,r_1,r_2,r_3,r_4)}}\ket{v_5(0,r_1,r_2,r_3,r_4)}$
give the uniform superposition
$\sum_{r_4'\in\mathbb{F}_{11}}\bl{\ket{r_4'}}\ket{r_4'}$
independent of $r_1,r_2,r_3,r_5,r_6$.
The variable change from $r_2$ to $r_2'$ can also be obtained similarly.

Now, the secret is disentangled with the rest of the qudits. Thus, any arbitrary superposition of the basis states can be recovered with above steps for $d=4$.

In the case when $d=3$, each of the three contacted parties sends all three qudits in its share.
\begin{align*}
\sum_{\underline{r}\in\mathbb{F}_{11}^6}
\begin{array}{l}
\bl{\ket{v_1(\underline{s},r_1,r_2)}\ket{v_1(0,r_1,r_2,r_3,r_4)}
\ket{v_1(0,0,r_3,r_5,r_6)}}
\\\bl{\ket{v_2(\underline{s},r_1,r_2)}\ket{v_2(0,r_1,r_2,r_3,r_4)}
\ket{v_2(0,0,r_3,r_5,r_6)}}
\\\bl{\ket{v_3(\underline{s},r_1,r_2)}\ket{v_3(0,r_1,r_2,r_3,r_4)}
\ket{v_3(0,0,r_3,r_5,r_6)}}
\\\ket{v_4(\underline{s},r_1,r_2)}\ket{v_4(0,r_1,r_2,r_3,r_4)}
\ket{v_4(0,0,r_3,r_5,r_6)}
\\\ket{v_5(\underline{s},r_1,r_2)}\ket{v_5(0,r_1,r_2,r_3,r_4)}
\ket{v_5(0,0,r_3,r_5,r_6)}
\end{array}
\end{align*}
The secret recovery for $d=3$ also uses operations similar to those in the case of $d=4$. For sake of completeness, the secret recovery for $d=3$ in this scheme has been explained in Appendix \ref{ap:univ-ceqts-d-3}.

In all the three cases, the first step was to recover the basis state $\ket{\underline{s}}=\ket{s_1s_2s_3}$. The recovery is complete at this point if the secret is any one of the basis states (identical to a classical secret). But the quantum secret can be in an arbitrary superposition of basis states. To recover this quantum secret, the three qudits containing information on the secret needs to be disentangled from the rest of the qudits. For example, the first three qudits in \eqref{eq:entangled_secret}, though they have information on the basis states, are still entangled with the other qudits while these qudits are disentangled from the other qudits in \eqref{eq:disentangled_secret}.

\subsection{Comparison with fixed \titlemath{d} CE-QTS}
In contrast with the above scheme, for the standard $((3,5))$ QSS scheme due to Cleve {\em et al.} 3 qudits need to be communicated for recovery of 1 qudit of secret whenever the combiner accesses three or more parties. The $((3,5,5))$ CE-QTS scheme from \cite{senthoor19} described in \eqref{eq:eg-senthoor-ce-qts} gives a better communication cost of 5/3 qudits per 1 qudit of secret when the combiner accesses 5 parties. But this scheme does not provide the flexibility of also contacting four parties communication efficiently. The scheme provided above can solve that problem. It provides communication efficiency at both $d=5$ and $d=4$.

At $d=4$, the above scheme gives communication cost of 8 qudits to recover secret of 3 qudits \textit{i.e.} 8/3 qudits per one qudit of secret. However this is not the optimal communication cost for $d=4$. Because, for $d=4$, the communication cost in a $((3,5,4))$ fixed $d$ CE-QTS scheme from \cite{senthoor19} gives 2 qudits per one qudit of secret. The constructions proposed in the coming sections can give a $((3,5,*))$ universal CE-QTS scheme with the same communication efficiency as the fixed $d$ CE-QTS schemes of \cite{senthoor19} at both $d=4$ and $d=5$.

%\section{Constructions based on cascaded ramp QSS}
\section{Concatenation Framework for Constructing Communication Efficient QTS Schemes}\label{s:framework}
In this section, we develop  a framework for constructing communication efficient quantum secret sharing schemes. 
We propose a general framework which can be used to derive many classes of CE-QTS schemes. 
Ramp secret sharing schemes and threshold schemes are the central ingredients of the proposed constructions.
First, we give a systematic method to construct CE-QTS schemes where the combiner can contact $d$ parties, and reconstruct the secret. 
Here  $d$ is determined prior to secret distribution.
Then, we provide a systematic method to construct CE-QTS schemes where the combiner can contact any $d$ parties to reconstruct the secret. 
Here, $d$ can be determined after secret distribution arbitrarily by the combiner.

\subsection{Fixed \titlemath{d} CE-QTS from ramp QSS}
\label{ss:iii_a}
Suppose we have a $((k,n,d))$ CE-QTS scheme.
Consider any authorized set of $d>k$ parties.
Since this is an authorized set, we can reconstruct the secret. 
In a communication efficient scheme, these $d$ parties do not communicate their entire shares
to the combiner. 
They only communicate a portion of their share. 
For gaining the intuition, let us assume that the portion communicated by a party when a set of $d$ parties are contacted by the combiner is independent of the choice of the remaining $d-1$ parties. 
Since this is a $((k,n,d))$ scheme, any $k-1$ or fewer {\em portions} \textit{i.e.} partial shares cannot reveal any information about the secret. 
However, $k$ or more portions may reveal partial information about the secret, while 
$d$ out of all the $n$ portions can completely recover the secret. 
Therefore, the set of portions communicated by all the $n$ parties to the combiner can be modelled as a $((d,n;k-1))$ ramp QSS scheme.

Now let us see if we can build a $((k,n))$ QTS scheme out of this $((d,n;k-1))$ ramp QSS scheme.
If $k$ of these $n$ parties attempt to reconstruct the secret with just their shares from the ramp scheme, then their $k$ shares may not be enough for the reconstruction of the secret.
The combiner will need shares from $d-k$ more parties of the ramp scheme for the additional information required to recover the secret for sure.
So we extend the ramp scheme to a $((d,n+d-k;k-1))$ scheme by allowing for $d-k$ more new shares in the previous ramp scheme.
These additional $d-k$ shares of the ramp scheme are distributed to the $n$ parties after encoding by a
$((k,n))$ threshold scheme so that even if only $k$ parties are contacted by the combiner these $d-k$ extra shares necessary for secret recovery can be recovered.
The full scheme is illustrated in Fig.~\ref{fig:fixed_d_rqss}
and formally proved in Theorem~\ref{th:ramp-fixed-ceqts}.
\begin{figure}[ht]
\begin{center}
\hspace{-0.5cm}
\begin{tikzpicture}[scale=0.7, every node/.style={scale=0.78}]
%Ramp QSS box
\draw (-0.6,-0.05) -- (1.5,-0.05);
\draw (1.5,-0.05) -- (1.5,5.5);
\draw (1.5,5.5) -- (-0.6,5.5);
\draw (-0.6,5.5) -- (-0.6,-0.05);
\node at (0.45,4) {\small $((t',$$n';$$z'))$};
\node at (0.45,3.5) {\small ramp QSS};
\node at (0.45,2.5) {\small $n'$$=$$n$$+$$d$$-$$k$};
\node at (0.45,2) {\small $t'$$=$$d$};
\node at (0.45,1.5) {\small $z'$$=$$k$$-$$1$};
%input qudits
\draw[->] (-1,2.75) -- (-0.6,2.75);
\node at (-1.35,2.75) {$\ket{\phi}$};
%output qudits
\draw[->] (1.5,5.25) -- (6.15,5.25);
\draw[->] (1.5,4.75) -- (6.15,4.75);
\draw[->] (1.5,4.25) -- (6.15,4.25);
\node at (1.8,3.75) {$\vdots$};
\draw[->] (1.5,3) -- (6.15,3);
\draw [decorate,decoration={brace,amplitude=6}](2,5.55) -- (2,2.7);
\node at (2,5.75) {\small $n$};
\draw[->] (1.5,1.8) -- (3.25,1.8);
\draw[->] (1.5,1.3) -- (3.25,1.3);
\node at (1.8,0.925) {$\vdots$};
\draw[->] (1.5,0.3) -- (3.25,0.3);
\draw [decorate,decoration={brace,amplitude=6}](2,2.1) -- (2,0);
\node at (2.1,2.3) {\small $d-k$};
%layer separator 1
\node at (1,6.4) {\small Layer 1 encoding};
\draw [dashed] (2.75,-0.7) -- (2.75,6.7);
\node at (4.15,6.4) {\small Layer 2 encoding};
%Perfect QSS box
\begin{scope}[xshift=-0.5cm]
\draw (3.75,-0.2) -- (5.1,-0.2);
\draw (5.1,-0.2) -- (5.1,2.3);
\draw (5.1,2.3) -- (3.75,2.3);
\draw (3.75,2.3) -- (3.75,-0.2);
\node at (4.425,1.3) {\small $((k$,$n))$};
\node at (4.425,0.9) {\small QTS};
\draw (5.1,2.05) -- (9.75,2.05);
\node at (5.35,1.675) {$\vdots$};
\draw (5.1,1.05) -- (10.5,1.05);
\draw (5.1,0.55) -- (11,0.55);
\draw (5.1,0.05) -- (11.5,0.05);
\draw [decorate,decoration={brace,amplitude=6}](5.5,2.35) -- (5.5,-0.25);
\node at (5.6,-0.5) {\small $n$};
\end{scope}
\begin{scope}[xshift=3.5cm]
\draw (5.75,2.05) -- (5.75,2.9);
\draw (6.5,1.05) -- (6.5,4.15);
\draw (7,0.55) -- (7,4.65);
\draw (7.5,0.05) -- (7.5,5.15);
\draw[->] (5.75,2.9) -- (5.5,2.9);
\draw[->] (6.5,4.15) -- (5.5,4.15);
\draw[->] (7,4.65) -- (5.5,4.65);
\draw[->] (7.5,5.15) -- (5.5,5.15);
\end{scope}
%layer separator 2
\draw [dashed] (5.6,-0.7) -- (5.6,6.7);
% ceqts grid
\begin{scope}[xshift=1cm,yshift=0.25cm]
\draw (5.15,5.25) -- (5.15,2.5);
\draw (5.75,5.25) -- (5.75,2.5);
\draw (8,5.25) -- (8,2.5);
\draw (8,5.25) -- (5.15,5.25);
\node at (5.45,5) {$A_1$};
\node at (6.875,5) {$B_1$};
\node at (8.3,5.15) {$S_1$};
\draw (5.15,4.75) -- (8,4.75);
\node at (5.45,4.5) {$A_2$};
\node at (6.875,4.5) {$B_2$};
\node at (8.3,4.65) {$S_2$};
\draw (5.15,4.25) -- (8,4.25);
\node at (5.45,4) {$A_3$};
\node at (6.875,4) {$B_3$};
\node at (8.3,4.15) {$S_3$};
\draw (5.15,3.75) -- (8,3.75);
\node at (5.45,3.5) {$\vdots$};
\node at (6.875,3.5) {$\vdots$};
\node at (8.3,3.5) {$\vdots$};
\draw (5.15,3) -- (8,3);
\node at (5.45,2.75) {$A_n$};
\node at (6.875,2.75) {$B_n$};
\node at (8.3,2.9) {$S_n$};
\draw (5.15,2.5) -- (8,2.5);
\end{scope}
\end{tikzpicture}
\captionsetup{justification=justified}
\caption{Concatenation framework for constructing $((k,n,d))$ CE-QTS scheme with $((d,n+d-k;k-1))$ ramp QSS scheme.}
\label{fig:fixed_d_rqss}
\end{center}
\end{figure}
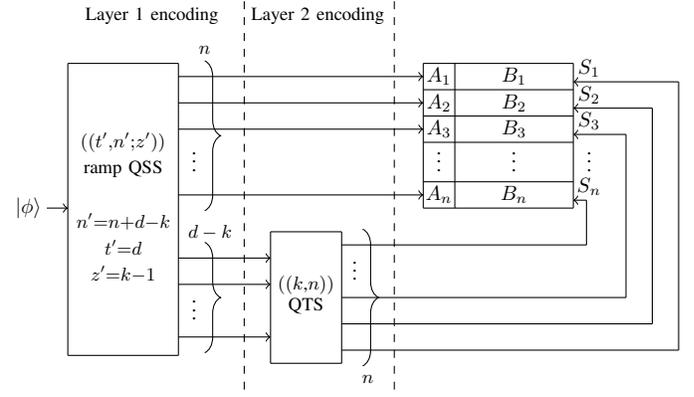
\begin{algorithm}[ht]
\caption{Encoding for a $((k,n,d))$ CE-QTS scheme using a $((d,n+d-k;k-1))$ ramp QSS and a $((k,n))$ QTS.}
\begin{algorithmic}[1]
\REQUIRE {Secret $\ket{\phi}$}
\ENSURE {Shares of the $n$ parties, $S_j$ for $1\leq j\leq n$}
\STATE Encode the secret $\ket{\phi}$ using the $((d,n+d-k;k-1))$ ramp QSS scheme.
Denote the $j$th share generated by the ramp scheme as $A_j$ for $1\leq j\leq n+d-k$ such that the last $d-k$ shares have the largest share sizes.
\STATE Encode the quantum state in $(A_{n+1},A_{n+2},\hdots,A_{n+d-k})$ using a $((k, n))$ quantum threshold scheme.
Denote the $j$th share of this QTS scheme as $B_j$ for $1\leq j\leq n$.
\STATE Distribute $S_j=(A_j,B_j)$ to the $j$th party for $1\leq j\leq n$.
\end{algorithmic}
\label{alg:fixed_d_rqss_enc}
\end{algorithm}
\begin{algorithm}[ht]
\caption{Secret recovery for the $((k,n,d))$ CE-QTS scheme from the encoding in Algorithm \ref{alg:fixed_d_rqss_enc}}
\begin{algorithmic}[1]
\REQUIRE {Shares of $k$ parties or layer 1 from any $d$ parties}
\ENSURE {Secret $\ket{\phi}$}
\IF{combiner has access to only $k$ shares}
 \PARSTATE{Download full shares from the $k$ parties.}
 \PARSTATE{Use layer 2 from the $k$ parties to recover the input to the $((k,n))$ QTS scheme \textit{i.e.} $(A_{n+1},A_{n+2},\hdots,A_{n+d-k})$.}
 \PARSTATE{Use $(A_{n+1},A_{n+2},\hdots,A_{n+d-k})$ and layer 1 from the $k$ parties to get $d$ shares of the ramp QSS scheme and recover the secret $\ket{\phi}$.}
\ELSIF{combiner has access to $d$ shares}
 \PARSTATE{Download layer 1 from the $d$ parties.}
 \PARSTATE{Use layer 1 from the $d$ parties to get $d$ shares of the ramp QSS scheme and recover the secret $\ket{\phi}$.}
\ENDIF
\end{algorithmic}
\label{alg:fixed_d_rqss_rec}
\end{algorithm}
\begin{theorem}[Concatenation framework for fixed $d$ CE-QTS]
\label{th:ramp-fixed-ceqts}
A $((k,n,d))$ CE-QTS scheme exists, if a $((d,n+d-k;k-1))$ ramp QSS scheme and a $((k,n))$ QTS scheme exist. 
The encoding for this scheme is given in Algorithm~\ref{alg:fixed_d_rqss_enc} and the recovery in 
Algorithm~\ref{alg:fixed_d_rqss_rec}.
\end{theorem}
\begin{proof}
The proof is by giving an explicit construction of a $((k, n, d ))$ CE-QTS scheme
from the given ramp QSS and $((k, n))$ threshold schemes. The encoding for the $((k,n,d))$ CE-QTS scheme is as given in Algorithm~\ref{alg:fixed_d_rqss_enc}.
Each share $S_j$ consists of two portions $(A_j, B_j)$.
We say that $A_j$ forms the first layer of the share $S_j$ and $B_j$ the second layer.
Here, for any $L\subseteq[n]$, $S_L$ denotes $\{S_j\}_{j\in L}$ and $|S_j|$ gives the number of qudits in the share $S_j$.
Similar notations are used for $\{A_j\}$ and $\{B_j\}$ as well.
\begin{compactenum}[(i)]
\item \textit{Recoverability}: 
The secret recovery for the $((k,n,d))$ CE-QTS scheme is as given in 
Algorithm~\ref{alg:fixed_d_rqss_rec}. 
While the combiner accesses any set of $d$ parties, it just needs layer 1 of these parties to recover the secret from the underlying ramp scheme. But while accessing only $k$ parties, the combiner needs $d-k$ more shares of the ramp scheme to recover the secret. 
These $d-k$ extra shares are recovered from the $((k, n ))$ scheme with qudits from second layer.
\item \textit{Secrecy}: Consider any set $L\subseteq [n]$ such that $|L|=k-1$. 
By Lemma \ref{lm:mixed-to-pure}, let $E_1$ be the purifying state for the ramp QSS scheme such that the shares $A_1,A_2,\hdots,A_{n+d-k},E_1$ give a pure state scheme encoding $\ket{\phi}$.
Similarly, let $E_2$ be the purifying state for the perfect QSS scheme such that the shares $B_1,B_2,\hdots,B_n,E_2$ give a pure state scheme encoding $(A_{n+1},A_{n+2},\hdots,A_{n+d-k})$.
Overall, $S_{[n]}\cup\{E_1,E_2\}$ gives a pure state scheme encoding $\ket{\phi}$.
If it can be proved that $S_{[n]\backslash L}\cup\{E_1,E_2\}$ can recover the secret, then by no-cloning theorem, $S_L$ has no information on the secret which proves the secrecy property of the CE-QTS scheme of Algorithm~\ref{alg:fixed_d_rqss_enc}.

Assume that Alice has the shares $S_{[n]\backslash L}\cup\{E_1,E_2\}$.
Clearly, $B_L$ is an unauthorized set in the QTS scheme.
By Lemma \ref{lm:pu-auth}, $B_{[n]\backslash L}\cup\{E_2\}$ is an authorized set for recovering $(A_{n+1},A_{n+2},\hdots,A_{n+d-k})$.
Thus, Alice recovers $(A_{n+1},A_{n+2},\hdots,A_{n+d-k})$ from the QTS scheme.
Now, Alice has the shares $A_{[n+d-k]\backslash L}\cup\{E_2\}$.
$A_L$ is an unauthorized set in the ramp QSS scheme.
By Lemma \ref{lm:pu-auth}, $A_{[n-d+k]\backslash L}\cup\{E_1\}$ is an authorized set in the ramp QSS scheme. Hence, Alice recovers the secret $\ket{\phi}$ from the ramp QSS scheme.
\item \textit{Communication efficiency}: Consider the set of $d$ parties given by $D\subseteq[n]$ which has maximum communication cost among all sets of $d$ parties. By definition, the communication cost of this set of $d$ parties equals CC$_n(d)$. Pick a $K\subset D$ such that $|K|=k$.
\begin{eqnarray}
\text{CC}_n(k)&=&\sum_{j\in K}|S_j|=\sum_{j\in K}(|A_j|+|B_j|)\nonumber
\\&\geq&\sum_{j\in K}|A_j|+\sum_{j\in K}\sum_{\ell=n+1}^{n+d-k}|A_\ell|
\label{eq:layer2-size}
\\&=&\sum_{j\in K}|A_j|+k\sum_{\ell=n+1}^{n+d-k}|A_\ell|
\nonumber
\\&>&\sum_{j\in K}|A_j|+\sum_{\ell=n+1}^{n+d-k}|A_\ell|
\label{eq:k-exceeds-1}
\\&\geq&\sum_{j\in K}|A_j|+\sum_{j\in D\backslash K}|A_j|
\label{eq:largestd_k}
\\&=&\sum_{j\in D}|A_j|=\text{CC}_n(d)\nonumber
\end{eqnarray}
where $J=D\backslash K$.
The bound on \eqref{eq:layer2-size} is due to Lemma~\ref{lm:qts-opt} which implies that each share $B_i$ of the QTS scheme is at least as large as the input state given by $(A_{n+1},A_{n+2},\hdots,A_{n+d-k})$.
The strict inequality in \eqref{eq:k-exceeds-1} is because $k>1$.
The inequality \eqref{eq:largestd_k} is due to the fact that the shares $A_{n+1},A_{n+2},\hdots,A_{n+d-k}$ have the largest sizes among all the $n+d-k$ shares of the ramp scheme. 
Finally, we arrive at the conclusion that the communication complexity decreases with the size of the authorized set. 
\end{compactenum}
This concludes the proof of the theorem.
\end{proof}

Theorem~\ref{th:ramp-fixed-ceqts} can be used with various ramp QSS and threshold schemes. 
Note that Theorem~\ref{th:ramp-fixed-ceqts} does not require the alphabet $q$ to be a prime.
The communication complexity of the resulting schemes clearly depends on the underlying ramp QSS scheme and QTS scheme. 
Here, we propose a construction for CE-QTS scheme using the ramp QSS scheme proposed by Ogawa {\em et al.}\cite{ogawa05} and the QTS scheme from Cleve {\em et al.}\cite{cleve99}.

\begin{corollary}[Concatenated construction for fixed $d$ CE-QTS]
\label{co:ramp-fixed-ceqts-constn}
A $q$-ary $((k,n,d))$ communication efficient QTS scheme can be constructed using the encoding in Algorithm \ref{alg:fixed_d_rqss_enc} with the following 
parameters.
\begin{gather*}
q>d+k-1\text{ (prime)}\\
m=d-k+1\\
w_1=w_2=\hdots=w_n=d-k+1\\
\text{CC}_n(k)=k(d-k+1)\\
\text{CC}_n(d)=d.
\end{gather*}
\end{corollary}
\begin{proof}
Consider the Concatenation framework from Theorem~\ref{th:ramp-fixed-ceqts}.
Use the ramp scheme from \cite{ogawa05} given in Lemma \ref{lm:ogawa-ramp} and the QTS scheme from \cite{cleve99} given in Lemma \ref{lm:cleve-qts} for the underlying schemes.

By Lemma \ref{lm:ogawa-ramp}, the dimension of the qudits has to be a prime $q$ such that 
$q>d-k+1$.
This also satisfies the constraint on the dimension for the QTS scheme from Lemma \ref{lm:cleve-qts}.
The size of the secret in the ramp scheme is $m=d-k+1$ qudits.

Each share of the ramp QSS is of size one qudit. Thus the first layer of each share in the CE-QTS has one qudit.
The input state for the $((k,n))$ QTS will have $d-k$ qudits.
By Lemma \ref{lm:cleve-qts}, the size of each share of the QTS scheme is also $d-k$. 
Hence, the second layer of each share in CE-QTS has $d-k$ qudits.
In total, each share in the CE-QTS scheme has
%\begin{equation}
$w_j=d-k+1$
%\end{equation}
qudits for $1\leq j\leq n$.

When the combiner attempts to recover from just $k$ parties, each of them transmits the entire share of $d-k+1$ qudits.
Thus
%\begin{equation}
$\text{CC}_n(k)=k(d-k+1)$.
%\end{equation}
When the combiner contacts any $d$ parties, each of them sends a qudit from the first layer, giving
%\begin{equation}
$\text{CC}_n(d)=d$.
%\end{equation}
\end{proof}
In the CE-QTS scheme as described in Corollary \ref{co:ramp-fixed-ceqts-constn}, note that the dimension of each of the $d-k+1$ qudits in the secret has to be more than $d+k-1$.
Compare this with the CE-QTS scheme from \cite{senthoor19} which can give a smaller dimension of $q>2k-1$.
(Refer Table \ref{tab:contributions}.)
However, using other ramp schemes in this framework could lead to CE-QTS schemes with qudits of dimension less than or equal to $d-k+1$.

\subsection{Universal CE-QTS from ramp QSS}\label{ss:iii_b}
Consider an $((n,n;k-1))$ ramp QSS scheme (marked black in Fig. \ref{fig:var_rqss}). Now, if a combiner has access to only $n-1$ out of the $n$ parties, the combiner will not be able to recover the secret unless he receives one more share from this scheme. If these $n-1$ parties can send the combiner some more qudits containing information about an extra share, then the combiner can recover the secret with this extra share.

This flexibility can be achieved by instead taking an $((n+1,n;k-1))$ ramp QSS scheme where the first $n$ shares are given to $n$ parties and the $(n+1)$th share is encoded and distributed among the $n$ parties through an $((n-1,n;k-1))$ scheme (which is indicated with blue in Fig. \ref{fig:var_rqss}). Then, whenever the combiner has access to only $n-1$ parties, he will first decode the $((n-1,n;k-1))$ scheme to recover the extra share and then use the $n-1$ shares from the $((n,n+1;k-1))$ ramp scheme along with this extra share to recover the secret.
\begin{figure}[ht]
\begin{center}
\hspace{-0.5cm}
\begin{tikzpicture}[scale=0.7, every node/.style={scale=0.78}]
%Ramp QSS box
\draw (-0.6,-0.05) -- (1.5,-0.05);
\draw (1.5,-0.05) -- (1.5,5.5);
\draw (1.5,5.5) -- (-0.6,5.5);
\draw (-0.6,5.5) -- (-0.6,-0.05);
\node at (0.45,4) {\small $((t,$$n';$$z))$};
\node at (0.45,3.5) {\small ramp QSS};
\node at (0.45,2.5) {\small $n'$$=$$n$\bl{$+$$1$}};
\node at (0.45,2) {\small $t$$=$$n$};
\node at (0.45,1.5) {\small $z$$=$$k$$-$$1$};
%input qudits
\draw[->] (-1,2.75) -- (-0.6,2.75);
\node at (-1.35,2.75) {$\ket{\phi}$};
%output qudits
\draw[->] (1.5,5.25) -- (6.15,5.25);
\draw[->] (1.5,4.55) -- (6.15,4.55);
\draw[->] (1.5,3.85) -- (6.15,3.85);
\node at (1.75,3.125) {$\vdots$};
\draw[->] (1.5,2.4) -- (6.15,2.4);
\draw [decorate,decoration={brace,amplitude=6}](1.85,5.55) -- (1.85,2.1);
\node at (1.85,5.75) {\small $n$};
\bl{\draw[->] (1.5,0.3) -- (2.85,0.3);}
%layer separator 1
\node at (0.7,6.4) {\small Layer 1 encoding};
\draw [dashed] (2.45,-1.4) -- (2.45,6.7);
\node at (4.15,6.4) {\small Layer 2 encoding};
%Ramp QSS 2 box
\bl{
\begin{scope}[xshift=-0.5cm,yshift=-0.7cm]
\draw (3.35,-0.5) -- (5.25,-0.5);
\draw (5.25,-0.5) -- (5.25,2.5);
\draw (5.25,2.5) -- (3.35,2.5);
\draw (3.35,2.5) -- (3.35,-0.5);
\node at (4.3,2) {\small $((t$,$n'$;$z))$};
\node at (4.3,1.5) {\small ramp QSS};
\node at (4.3,1) {\small $n'$$=$$n$};
\node at (4.3,0.5) {\small $t$$=$$n$$-$$1$};
\node at (4.3,0) {\small $z$$=$$k$$-$$1$};
\draw (5.25,2.05) -- (9.75,2.05);
\node at (5.5,1.675) {$\vdots$};
\draw (5.25,1.05) -- (10.5,1.05);
\draw (5.25,0.55) -- (11,0.55);
\draw (5.25,0.05) -- (11.5,0.05);
\draw [decorate,decoration={brace,amplitude=6}](5.65,2.35) -- (5.65,-0.25);
\node at (5.75,-0.5) {\small $n$};
\end{scope}
\begin{scope}[xshift=3.5cm]
\draw (5.75,1.35) -- (5.75,2.4);
\draw (6.5,0.35) -- (6.5,3.85);
\draw (7,-0.15) -- (7,4.55);
\draw (7.5,-0.65) -- (7.5,5.25);
\draw[->] (5.75,2.4) -- (5.5,2.4);
\draw[->] (6.5,3.85) -- (5.5,3.85);
\draw[->] (7,4.55) -- (5.5,4.55);
\draw[->] (7.5,5.25) -- (5.5,5.25);
\end{scope}
}%bl
%layer separator 2
\draw [dashed] (5.7,-1.4) -- (5.7,6.7);
% Variable RQSS grid
\begin{scope}[xshift=1cm,yshift=0.35cm]
\draw (5.15,5.25) -- (5.15,1.7);
\draw (6,5.25) -- (6,1.7);
\draw (5.15,5.25) -- (6,5.25);
\node at (5.6,4.9) {$S_1^{(1)}$};
\node at (8.3,5.15) {$S_1$};
\draw (5.15,4.55) -- (6,4.55);
\node at (5.6,4.2) {$S_2^{(1)}$};
\node at (8.3,4.45) {$S_2$};
\draw (5.15,3.85) -- (6,3.85);
\node at (5.6,3.5) {$S_3^{(1)}$};
\node at (8.3,3.75) {$S_3$};
\draw (5.15,3.15) -- (6,3.15);
\node at (5.6,2.9) {$\vdots$};
\node at (8.3,3.15) {$\vdots$};
\draw (5.15,2.4) -- (6,2.4);
\node at (5.6,2.05) {$S_n^{(1)}$};
\node at (8.3,2.3) {$S_n$};
\draw (5.15,1.7) -- (6,1.7);
\bl{
\draw (8,5.25) -- (8,1.7);
\draw (6,5.25) -- (8,5.25);
\node at (7,4.9) {$S_1^{(2)}$};
\draw (6,4.55) -- (8,4.55);
\node at (7,4.2) {$S_2^{(2)}$};
\draw (6,3.85) -- (8,3.85);
\node at (7,3.5) {$S_3^{(2)}$};
\draw (6,3.15) -- (8,3.15);
\node at (7,2.9) {$\vdots$};
\draw (6,2.4) -- (8,2.4);
\node at (7.1,2.05) {$S_n^{(2)}$};
\draw (5.15,1.7) -- (8,1.7);
}
\end{scope}
\end{tikzpicture}
\captionsetup{justification=justified}
\caption{Concatenation of two ramp quantum secret sharing schemes to  construct a $((t,n;k-1))$ ramp QSS scheme with flexible $t\in\{n-1,n\}$.}
\label{fig:var_rqss}
\end{center}
\end{figure}
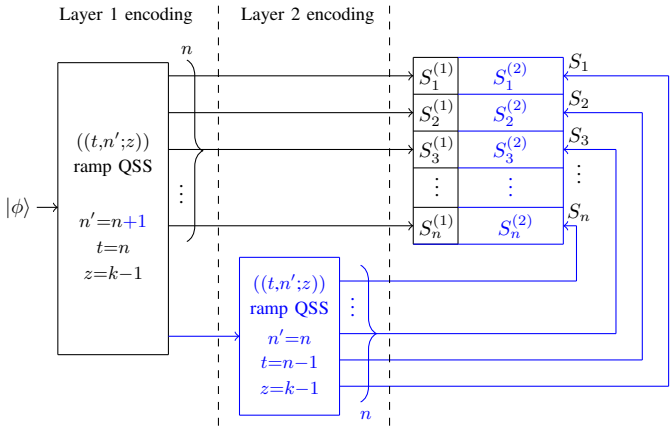

Thus, by concatenating an $((n-1,n;k-1))$ ramp scheme which encodes the extra share from an $((n,n+1;k-1))$ ramp scheme, a $((t,n;k-1))$ ramp QSS with a flexible threshold $t\in\{n-1,n\}$ can be designed. Similarly, a $((k,n,*))$ universal CE-QTS scheme is a QTS scheme in which the secret recovery can happen efficiently for all thresholds $d\in\{k,k+1,\hdots,n\}$. The main idea in our following framework for constructing universal CE-QTS schemes is that this generalization of $d$ can be achieved by concatenating $n-k+1$ ramp schemes with increasing threshold $t$ successively.

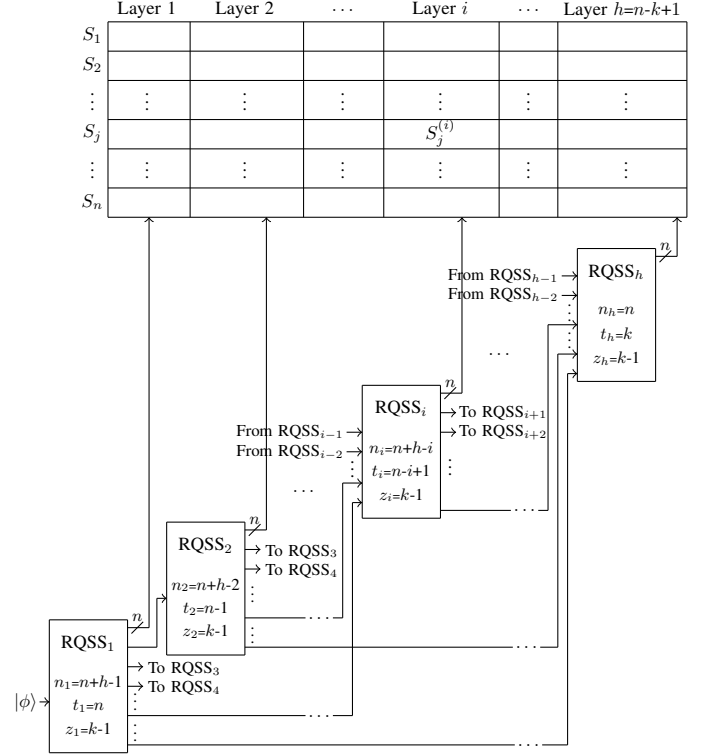
\begin{figure}[ht]
\begin{center}
\hspace{-0.2cm}
\begin{tikzpicture}[scale=0.52, every node/.style={scale=0.7}]
\newcommand\BoxLayer[5]
{
\tikzset{shift={(#2,#3)}}
%layer box
\draw (0,-0.2) -- (2,-0.2);
\draw (2,-0.2) -- (2,3.2);
\draw (2,3.2) -- (0,3.2);
\draw (0,3.2) -- (0,-0.2);
\node at (1,2.6) {RQSS$_{#1}$};
%\node at (1,2) {\footnotesize $n_#1$=};
\node at (1,1.6) {\small $n_#1${\scriptsize=}#5};
\node at (1,1) {\small $t_#1${\scriptsize=}#4};
\node at (1,0.4) {\small $z_#1${\scriptsize=}$k$-1};
%layer output to encoding
\draw (2,3) -- (2.55,3);
\draw (2.25-0.15,3-0.15) -- (2.25+0.15,3+0.15);
\node at (2.25,3.25) {\small $n$};
\tikzset{shift={(-#2,-#3)}}
}
%layer 1 box
\draw[->] (-0.25,1.1) -- (0,1.1);
\node at (-0.6,1.1) {$\ket{\phi}$};
\BoxLayer{1}{0}{0}{$n$}{$n$+$h$-1}
\draw (2,2.5) -- (2.75,2.5); %to layer 2
\draw[->] (2,2) -- (2.4,2); %to layer 3
\node at (3.45,1.95) {\small To RQSS$_3$};
\draw[->] (2,1.5) -- (2.4,1.5); %to layer 4
\node at (3.45,1.45) {\small To RQSS$_4$};
\node at (2.2,1.25) {\small $\vdots$};
\draw (2,0.75) -- (6.55,0.75); %to layer i
\node at (2.2,0.45) {\small $\vdots$};
\draw (2,0) -- (11.9,0); %to layer k
%layer 2 box
\BoxLayer{2}{3}{2.5}{$n$-1}{$n$+$h$-2}
\draw (2.75,2.5) -- (2.75,3.75); %from layer 1
\draw[->] (2.75,3.75) -- (3,3.75);
\draw[->] (5,5) -- (5.4,5); %to layer 3
\node at (6.45,4.95) {\small To RQSS$_3$};
\draw[->] (5,4.5) -- (5.4,4.5); %to layer 4
\node at (6.45,4.45) {\small To RQSS$_4$};
\node at (5.2,4) {\small $\vdots$};
\draw (5,3.25) -- (6.55,3.25); %to layer i
\node at (5.2,2.95) {\small $\vdots$};
\draw (5,2.5) -- (11.85,2.5); %to layer k
\node at (6.5,6.5) {$\hdots$};
%layer i box
\BoxLayer{i}{8}{6}{$n$-$i$+1}{$n$+$h$-$i$}
\node at (6.9,0.75) {$\hdots$};
\draw (7.25,0.75) -- (7.75,0.75); %from layer 1
\draw (7.75,0.75) -- (7.75,6.2);
\draw[->] (7.75,6.2) -- (8,6.2);
\node at (6.9,3.25) {$\hdots$};
\draw (7.25,3.25) -- (7.5,3.25); %from layer 2
\draw (7.5,3.25) -- (7.5,6.7);
\draw[->] (7.5,6.7) -- (8,6.7);
\node at (7.75,7.2) {$\vdots$};
\draw[->] (7.6,7.5) -- (8,7.5); %from layer (i-2)
\node at (6.15,7.5) {\small From RQSS$_{i-2}$};
\draw[->] (7.6,8) -- (8,8); %from layer (i-1)
\node at (6.15,8) {\small From RQSS$_{i-1}$};
\draw[->] (10,8) -- (10.4,8); %to layer (i+1)
\node at (11.6,8.5) {\small To RQSS$_{i+1}$};
\draw[->] (10,8.5) -- (10.4,8.5); %to layer (i+2)
\node at (11.6,8) {\small To RQSS$_{i+2}$};
\node at (10.25,7.25) {$\vdots$};
\draw (10,6) -- (11.8,6);
\node at (11.5,10) {$\hdots$};
%layer h=n-k+1 box
\BoxLayer{h}{13.5}{9.5}{$k$}{$n$}
\tikzset{shift={(4.5,3.5)}}
\node at (7.75,-3.5) {$\hdots$};  %from layer 1
\draw (8.1,-3.5) -- (8.75,-3.5);
\draw (8.75,-3.5) -- (8.75,6);
\draw[->] (8.75,6) -- (9,6);
\node at (7.7,-1) {$\hdots$}; %from layer 2
\draw (8.05,-1) -- (8.5,-1);
\draw (8.5,-1) -- (8.5,6.5);
\draw[->] (8.5,6.5) -- (9,6.5);
\node at (8.8,7) {\small $\vdots$};
\node at (7.65,2.5) {$\hdots$}; %from layer i
\draw (8,2.5) -- (8.25,2.5);
\draw (8.25,2.5) -- (8.25,7.25);
\draw[->] (8.25,7.25) -- (9,7.25);
\node at (8.8,7.7) {\small $\vdots$};
\draw[->] (8.6,8) -- (9,8); %from layer (k--2)
\node at (7.1,8) {\small From RQSS$_{h-2}$};
\draw[->] (8.6,8.5) -- (9,8.5); %from layer (k-1)
\node at (7.1,8.5) {\small From RQSS$_{h-1}$};
\tikzset{shift={(-4.5,-3.5)}}
%layer output to encoding
\draw[->] (2.55,3) -- (2.55,13.5);
\draw[->] (5.55,5.5) -- (5.55,13.5);
\draw[->] (10.55,9) -- (10.55,13.5);
\draw[->] (16.05,12.5) -- (16.05,13.5);
%Encoding box
\draw (1.5,13.5) -- (1.5,18.5);
\node at (2.5,18.8) {Layer 1};
\node at (2.5,16.6) {$\vdots$};
\node at (2.5,14.85) {$\vdots$};
\draw (3.6,13.5) -- (3.6,18.5);
\node at (5,18.8) {Layer 2};
\node at (5,16.6) {$\vdots$};
\node at (5,14.85) {$\vdots$};
\draw (6.5,13.5) -- (6.5,18.5);
\node at (7.5,18.8) {$\hdots$};
\node at (7.5,16.6) {$\vdots$};
\node at (7.5,14.85) {$\vdots$};
\draw (8.55,13.5) -- (8.55,18.5);
\node at (10,18.8) {Layer $i$};
\node at (10,16.6) {$\vdots$};
\node at (10,15.65) {$S_j^{(i)}$};
\node at (10,14.85) {$\vdots$};
\draw (11.5,13.5) -- (11.5,18.5);
\node at (12.25,18.8) {$\hdots$};
\node at (12.25,16.6) {$\vdots$};
\node at (12.25,14.85) {$\vdots$};
\draw (13,13.5) -- (13,18.5);
\node at (14.65,18.8) {Layer $h${\small =}$n$-$k$+$1$};
\node at (14.75,16.6) {$\vdots$};
\node at (14.75,14.85) {$\vdots$};
\draw (16.3,13.5) -- (16.3,18.5);
\draw (1.5,18.5) -- (16.3,18.5);
\node at (1.1,18.15) {$S_1$};
\draw (1.5,17.75) -- (16.3,17.75);
\node at (1.1,17.4) {$S_2$};
\draw (1.5,17) -- (16.3,17);
\node at (1.1,16.6) {$\vdots$};
\draw (1.5,16) -- (16.3,16);
\node at (1.1,15.65) {$S_j$};
\draw (1.5,15.25) -- (16.3,15.25);
\node at (1.1,14.85) {$\vdots$};
\draw (1.5,14.25) -- (16.3,14.25);
\node at (1.1,13.9) {$S_n$};
\draw (1.5,13.5) -- (16.3,13.5);
\end{tikzpicture}
\vspace{0cm}
\captionsetup{justification=justified}
\caption{Concatenation framework for constructing $((k,n,*))$ universal CE-QTS scheme by concatenating multiple ramp QSS schemes. Here $d_i=n+1-i$ for $1\leq i\leq n-k+1$ and the $((t_i=d_i,n_i=n+d_i-k;z_i=k-1))$ ramp QSS scheme is denoted by RQSS$_i$.}
\label{fig:univ_d_rqss}
\end{center}
\end{figure}
\begin{algorithm}[ht]
\caption{Encoding for a $((k,n,*))$ universal CE-QTS scheme }
\begin{algorithmic}[1]
\REQUIRE {Secret $\ket{\phi}$}
\ENSURE {Shares of the $n$ parties, $S_j$ for $1\leq j\leq n$}
\STATE Encode the secret $\ket{\phi}$ using the RQSS$_1$ scheme.
\FOR{$i=1$ to $n-k+1$}
 \PARSTATE{Distribute the smallest $n$ shares from the RQSS$_i$ scheme $(S_1^{(i)},S_2^{(i)},\hdots,S_n^{(i)})$ to the $n$ parties. This is called the $i$th layer of the encoding.}
 \IF{$d_i>k$}
    \PARSTATE{For all $1\leq\ell\leq d_i-k$, the share $S_{n+\ell}^{(i)}$ goes as part of input to the RQSS$_{i+\ell}$ scheme.}
    \PARSTATE{The combined state of all the qudits passed from the previous $i$ ramp schemes to the RQSS$_{i+1}$ scheme is encoded using the RQSS$_{i+1}$ scheme.}
 \ENDIF
 \ENDFOR
\end{algorithmic}
\label{alg:univ_d_rqss_enc}
\end{algorithm}

\begin{algorithm}[ht]
\caption{Secret recovery for the $((k,n,*))$ universal CE-QTS scheme in Algorithm \ref{alg:univ_d_rqss_enc}.}
\begin{algorithmic}[1]
\REQUIRE {The first $i$ layers of qudits from any $d_i$ parties for any $1\leq i\leq n-k+1$}
\ENSURE {Secret $\ket{\phi}$}
\STATE Use the $i$th layer of the $d_i$ parties to recover the input state of the RQSS$_i$ scheme.
\FOR{$\ell=i-1$ to $1$ step -1}
 \PARSTATE{Consider the RQSS$_\ell$ scheme. The $\ell$th layer of the $d_i$ parties will give $d_i$ shares of this ramp scheme.}
 \PARSTATE{Collect $d_\ell-d_i=i-\ell$ more shares of this ramp scheme one each from the input states recovered from the layers $\ell+1$ to $i$.}
 \PARSTATE{Use all these $d_\ell=d_i+i-\ell$ shares to recover the input state of the RQSS$_\ell$ scheme.}
 \ENDFOR
\STATE The input state of the RQSS$_1$ scheme gives the secret $\ket{\phi}$.
\end{algorithmic}
\label{alg:univ_d_rqss_rec}
\end{algorithm}

\begin{theorem}%[Universal $d$ CE-QTS from ramp QSS]
\label{lm:ramp-univ-ceqts}
If $q$-ary $((d_i,n+d_i-k;k-1))$ ramp QSS schemes exist for $1\leq i\leq n-k+1$ where $d_i=n+1-i$, then a $q$-ary $((k,n,*))$ universal communication efficient QTS exists.
The encoding for this scheme is given in Algorithm~\ref{alg:univ_d_rqss_enc} and the recovery in 
Algorithm~\ref{alg:univ_d_rqss_rec}.
\end{theorem}
\begin{proof}
We prove this result by giving an explicit construction for a CE-QTS scheme from the given ramp QSS schemes.
The encoding of the $((k,n,*))$ universal CE-QTS is as given in Algorithm \ref{alg:univ_d_rqss_enc} while the recovery is given in 
Algorithm~\ref{alg:univ_d_rqss_rec}.
Then we prove the communication efficiency of the proposed showing that communication complexity strictly reduces with size of the authorized set. 
The $((d_i,n+d_i-k;k-1))$ ramp QSS scheme is referred to as RQSS$_{i}$ here.
Here, for any $L\subseteq[n]$, $S_L$ denotes $\{S_j\}_{j\in L}$ and $|S_j|$ gives the number of qudits in the share $S_j$.
Similar notations are used for $\{S_j^{(i)}\}$ as well.

\begin{compactenum}[(i)]
\item \textit{Recoverability}: 
The secret recovery for the $((k,n,*))$ universal CE-QTS scheme is as given in 
Algorithm~\ref{alg:univ_d_rqss_rec}.
Whenever the combiner accesses $d_i$ parties, each of those parties send the first $i$ layers to the combiner. Once this is done, the combiner has $d_i$ shares in the RQSS$_i$ scheme. Hence, RQSS$_i$ can be decoded and its input qudits recovered. However, for decoding RQSS$_\ell$ schemes for $1\leq\ell\leq i-1$, the combiner still needs $d_\ell-d_i=i-\ell$ shares. For each RQSS$_\ell$, these deficit shares can be provided by the input qudits recovered from the schemes RQSS$_{\ell+1}$, RQSS$_{\ell+2},\hdots,$ RQSS$_i$, one share from each of these $i-\ell$ schemes. This iterative decoding of RQSS$_\ell$ will finally give the secret $\ket{\phi}$ after decoding RQSS$_1$. 
\item \textit{Secrecy}: Consider the set $J\subset[n]$ such that $|J|=k-1$. By Lemma \ref{lm:mixed-to-pure}, let $E_i$ be the purifying state for the RQSS$_i$ scheme for all $1\leq i\leq n-k+1$. Assume Alice has the set of shares $\{S_{[n]\backslash J},E_1,E_2,\hdots,E_{n-k+1}\}$. For RQSS$_{n-k+1}$, now Alice has the purifying state and every share except some $k-1$ shares. This set of $k-1$ shares in RQSS$_{n-k+1}$ has no information on its qudits. Therefore, by Lemma \ref{lm:pu-auth}, Alice has an authorized set for RQSS$_{n-k+1}$, from which she recovers its input qudits. These qudits will now give one extra share to each of the schemes RQSS$_{n-k}$ till RQSS$_1$. With this extra share, RQSS$_{n-k}$ will have an authorized set and from which Alice recovers its input qudits and retrieves one extra share to each of the schemes RQSS$_{n-k-1}$ till RQSS$_{1}$. By this iterative recovery process, finally Alice can recover the secret $\ket{\phi}$ from RQSS$_1$. Thus, the secret can be recovered from the set of shares $\{S_{[n]\backslash J},E_1,E_2,\hdots,E_{n-k+1}\}$. Hence, by no-cloning theorem, $S_J$ has no information on the secret \textit{i.e.} any $k-1$ or less parties in this scheme has no information on the secret.
\item \textit{Communication efficiency}: We now prove that for any $d_i$ such that $k<d_i\leq n$, the communication cost in our scheme is less than that of $d_i-1$.
By definition, CC$_n(d_i)$ is the maximum among the communication costs of all authorized sets of size $d_i$. Let $D\subseteq[n]$ be the authorized set which has this maximum communication cost CC$_n(d_i)$. Let $p\in D$ be one of these $d_i$ parties. 
Clearly, CC$_n(d_i-1)$ should be greater than or equal to the communication cost of the authorized set given by $D\backslash\{p\}$.
\begin{eqnarray}
\text{CC}_n(d_i-1)&\geq&\sum_{j\in D\backslash\{p\}}\sum_{\ell=1}^{i+1}|S^{(\ell)}_j|\nonumber
\\&=&\sum_{j\in D\backslash\{p\}}\sum_{\ell=1}^{i}|S^{(\ell)}_j|+\sum_{j\in D\backslash\{p\}}|S^{(i+1)}_j|\nonumber
\\\label{eq:temp6}
\end{eqnarray}
The $d_i-1$ shares in $\{S_j^{(i+1)}\}_{j\in D\backslash\{p\}}$ are from the $((d_i-1,n+d_i-1-k;k-1))$ RQSS$_i$ ramp scheme. Recall from Remark \ref{re:ramp-by-dropping} that after discarding the remaining $n-k$ shares from RQSS$_i$ scheme, this set of shares alone will give a $((d_i-1,d_i-1;k-1))$ ramp scheme which encodes the same state as RQSS$_i$ scheme. By Lemma \ref{lm:rqss-opt}, the average share size of this ramp scheme is at least $\frac{1}{d_i-k}$ times the total input size \textit{i.e.}
\begin{equation*}
\frac{1}{d_i-1}\sum_{j\in D\backslash\{p\}}|S^{(i+1)}_j|\geq\frac{1}{d_i-k}\sum_{j=1}^{i}|S^{(j)}_{n+i+1-j}|.
\end{equation*}
Applying this bound in \eqref{eq:temp6}, we obtain
\begin{flalign}
&\text{CC}_n(d_i-1)&\nonumber
\\&\ \ \ \geq\sum_{j\in D\backslash\{p\}}\sum_{\ell=1}^{i}|S^{(\ell)}_j|+\frac{d_i-1}{d_i-k}\sum_{\ell=1}^{i}|S^{(\ell)}_{n+i-\ell+1}|&\nonumber
\\&\ \ \ >\sum_{j\in D\backslash\{p\}}\sum_{\ell=1}^{i}|S^{(\ell)}_j|+\sum_{\ell=1}^{i}|S^{(\ell)}_{n+i-\ell+1}|&
\label{eq:temp7}
\\&\ \ \ \geq\sum_{j\in D\backslash\{p\}}\sum_{\ell=1}^{i}|S^{(\ell)}_j|+\sum_{\ell=1}^{i}|S^{(\ell)}_p|&
\label{eq:largestdi_k}
\\&\ \ \ =\sum_{j\in D}\sum_{\ell=1}^{i}|S^{(\ell)}_j| =\text{CC}_n(d_i)\nonumber
\end{flalign}
The strict inequality in \eqref{eq:temp7} is because $k>1$.
The inequality \eqref{eq:largestdi_k} is due to the fact that the shares $S_{n+1}^{(i)},S_{n+2}^{(i)},\hdots,S_{n+d_i-k}^{(i)}$ have the largest sizes among the $n+d_i-k$ shares of the RQSS$_i$ scheme. 
Therefore, we can conclude that $\text{CC}_n(d_i-1)>\text{CC}_n(d_i)$ proving that the proposed scheme is communication efficient.
\end{compactenum}
This concludes the proof of the theorem.
\end{proof}
With the above framework, the following construction for a universal CE-QTS can be provided by using the ramp QSS scheme by Ogawa et al\cite{ogawa05}.
\begin{corollary}[Concatenated construction for universal CE-QTS]%[Universal $d$ CE-QTS from ramp QSS]
\label{co:ramp-univ-ceqts-constn}
A $q$-ary $((k,n,*))$ universal communication efficient QTS scheme can be constructed using the encoding in Algorithm \ref{alg:univ_d_rqss_enc} with the following parameters.
\begin{gather*}
q>n+k-1\text{ (prime)}\\
m=\textup{lcm}\{1,2,\hdots,n-k+1\} \\
w_1=w_2=\hdots=w_n=m\\
\textup{CC}_n(d)=\frac{dm}{d-k+1} \text{ for }d\in\{k,k+1,\hdots,n\}
\end{gather*}
\end{corollary}
\begin{proof}
Consider the universal CE-QTS scheme from Algorithm \ref{alg:univ_d_rqss_enc} and use the schemes from \cite{ogawa05} given in Lemma \ref{lm:ogawa-ramp} for the underlying ramp schemes.
Clearly the dimension of each qudit $q$ should be above $t_i+z_i=d_i+k-1=n+k-i$ for all $1\leq i\leq n-k+1$.
Therefore, $q>n+k-1$.

Let $e_i$ be the number of qudits in the input state of the ramp QSS scheme RQSS$_i$ corresponding to the $i$th layer. The secret is the input to the scheme RQSS$_1$. Clearly,
%\begin{equation}
$e_1=m.$
%\end{equation}
For $i>1$, the input state of the ramp QSS scheme RQSS$_i$ has one share each from the ramp QSS schemes RQSS$_1$ to RQSS$_{i-1}$.
\begin{equation}
e_i=\sum_{\ell=1}^{i-1}|S_{n+i-\ell}^{(\ell)}|
\label{eq:recursion-1}
\end{equation}
Recall that, in the $((t_i,n_i;z_i))$ ramp schemes given in 
Lemma~\ref{lm:ogawa-ramp}, the size of each share is $\frac{1}{t_i-z_i}$ times the secret size \textit{i.e.} for any $1\leq j\leq n+d_i-k$,
\begin{equation}
|S_j^{(i)}|=\frac{e_i}{d_i-k+1}.
\label{eq:recursion-2}
\end{equation}
Solving the recursion from \eqref{eq:recursion-1} and \eqref{eq:recursion-2} with the initial condition $e_1=m$, we obtain, for $2\leq i\leq n-k+1$,
\begin{eqnarray}
e_i&=&\frac{m}{d_{i-1}-k+1}.
%\label{eq:rqss-i-input-size}
\nonumber
\end{eqnarray}

Note that for each 
$1\leq i\leq n-k+1$, implementing the scheme RQSS$_i$ requires $e_i$ to be divisible by $t_i-z_i=d_i-k+1$. This can be achieved by taking
%\begin{equation}
$m=\text{lcm}\{1,2,\hdots,n-k+1\}$.
%\end{equation}

From \eqref{eq:recursion-2}, the size of the $j$th share from RQSS$_i$ is
\begin{gather*}
|S_j^{(1)}|=\frac{m}{(d_1-k+1)}
\\|S_j^{(i)}|=\frac{m}{(d_i-k+1)(d_{i-1}-k+1)}
%\label{eq:rqss-i-share-size}
\end{gather*}
for $2\leq i\leq n-k+1$.
The total communication cost during secret recovery from a set of any $d_i$ parties given by $D$ can be calculated as
\begin{equation}
\text{CC}_n(d_i)=\sum_{j\in D}\sum_{\ell=1}^{i}|S_j^{(\ell)}|
=\frac{d_i m}{d_i-k+1}
\nonumber
\end{equation}
Also, for $1\leq j\leq n$, the size of the $j$th share is given by
\begin{eqnarray}
w_j=\sum_{i=1}^{n-k+1}|S_j^{(i)}|
=\sum_{i=1}^{n-k+1}\frac{e_i}{d_i-k+1}
=m.
\nonumber
\end{eqnarray}
\end{proof}
The  above corollary gives a construction based on the concatenation framework for a universal CE-QTS scheme. In the next section, we give another construction for universal CE-QTS schemes.

\section{Universal CE-QTS schemes based on Staircase codes}\label{s:iv}
In this section, we propose an alternate construction of universal CE-QTS based on classical communication efficient secret schemes constructed using Staircase codes\cite{bitar18}.
While constructing QSS schemes based on classical secret sharing schemes, there are some important differences.
For QSS schemes, the secret recovery should recover not just the basis states but also any arbitrary superposition of the basis states.
Hence the qudits containing the secret have to be disentangled from the remaining qudits, thus making the secret recovery in QSS schemes more involved.
\subsection{Encoding}\label{ss:iv_a}
\noindent
Communication efficient quantum secret sharing schemes for particular values of $k$ and $n=2k-1$ can be designed to work for all possible values of $d$ in the range $k$ through $n$ where $k\leq d\leq n$. We introduce the following terms before discussing the scheme. For $1\leq i\leq k$,
\begin{subequations}
\label{eq:ceqts-params}
\begin{gather}
d_i=n+1-i=2k-i\\
%m_i=d_i-k+1\\
%m=\text{lcm}\{m_1,m_2,\hdots, m_k\}\\
m=\textup{lcm}\{k,k-1,\hdots,1\}\\
a_i=m/(d_i-k+1)\\
b_i=a_{i} -\ a_{i-1} \text{ for }i>1,\ b_1=a_1  
\end{gather}
\end{subequations}
Here $m$ is the total number of secret qudits shared. 
The total number of qudits with each party is also given by $m$.
This is consistent with the fact that in a perfect threshold secret sharing scheme the size of the share must be at least as large as the secret \cite{gottesman00,imai03}.

Now $a_i$ gives the number of qudits communicated from each accessible share when $d_i$ parties are accessed to recover the secret. 
This means that $a_id_i$ qudits are communicated to the combiner when $d_i$ parties are contacted. 
Pick a prime
\begin{equation*}
q\geq 2(2k-1).
\end{equation*}
Consider the basis state of the secret $\underline{s}=(s_1,s_2,\ldots,s_m)\in\mathbb{F}_q^m$ and $\underline{r}=(r_1, r_2,\ldots,r_{m(k-1)})\in \mathbb{F}_q^{m(k-1)}$. 

Entries in $\underline{s}$ are rearranged into the matrix $S$ of size $k\times (m/k)$.\vspace{-0.25cm}
\begin{eqnarray}
S=\left[\begin{array}{cccc} 
s_1&s_{k+1}& \cdots & s_{m-k+1}\\
s_2&s_{k+2}& \cdots & s_{m-k+2}\\
\vdots&\vdots& \ddots & \vdots\\
s_k&s_{2k}& \cdots & s_{m}\\
\end{array} \right]\label{eq:secret-ceqts}
\end{eqnarray}

Entries in $\underline{r}$ are rearranged into $k$ matrices \textit{i.e.} $R_1$ of size $(k-1) \times b_1$, $R_2$  of size $(k-1)\times b_2$ and so on till $R_k$ of size $(k-1)\times b_k$.
\begin{eqnarray}
R_1=
\left[\begin{array}{cccc}
r_1& r_k& \cdots & r_{(a_1-1)(k-1)+1} \\
r_2& r_{k+1}& \cdots & r_{(a_1-1)(k-1)+2}\\
\vdots& \vdots & \ddots & \vdots\\
r_{k-1}& r_{2(k-1)}& \cdots & r_{a_1(k-1)}
\end{array} \right]\nonumber
\end{eqnarray}
For $2\leq i\leq k$, $R_i$ is given by 
\begin{eqnarray}
\!\left[\!\!\begin{array}{cccc}
r_{a_{i-1}(k-1)+1}& r_{(a_{i-1}+1)(k-1)+1}& \cdots & r_{(a_i-1)(k-1)+1} \\
r_{a_{i-1}(k-1)+2}& r_{(a_{i-1}+1)(k-1)+2}& \cdots & r_{(a_i-1)(k-1)+2}\\
\vdots& \vdots & \ddots & \vdots\\
r_{(a_{i-1}+1)(k-1)}& r_{(a_{i-1}+2)(k-1)}& \cdots & r_{a_i(k-1)}
\end{array}\!\! \right].\nonumber
\end{eqnarray}
The  matrix $C$, called code matrix,  is defined as follows.
\begin{eqnarray*}
C=VY%\label{eq:codeMatrix}
%C_{(2k-1)\times m} = V_{(2k-1)\times (2k-1)}  Y_{(2k-1) \times m}\\
\end{eqnarray*}
where $Y$ is given by
\begin{eqnarray*}
Y=
\left[
\begin{tabular}{c:c:c:c:c}
\multirow{4}{*}{$\ S\ $} & {\large \ 0\ } & \multirow{2}{*}{\large 0} & \multirow{4}{*}{$\ \ddots\ $} & \multirow{3}{*}{\large 0}\\ \cdashline{2-2}
&\multirow{3}{*}{$D_1$} & &\\ \cdashline{3-3}
& & \multirow{2}{*}{$D_2$} &\\ \cdashline{5-5}
& & & & $\ \ D_{k-1}\ \ $\\
\cdashline{1-5}
\multirow{2}{*}{$R_1$} & \multirow{2}{*}{$R_2$} & \multirow{2}{*}{$R_3$} & \multirow{2}{*}{$\hdots$} & \multirow{2}{*}{$R_k$}\\
& & & &\\
\end{tabular}
\right]
\end{eqnarray*}
and $V$ is an $n\times n$ Cauchy matrix given by
\begin{eqnarray}
\label{eq:def_cauchy_mtx}
[V]_{ij}=\frac{1}{x_i-y_j}
\end{eqnarray}
where $x_1, x_2,\hdots, x_n, y_1, y_2, \hdots, y_n$ are distinct constants from $\F_q$.
Here, $D_i$ of size $(k-i)\times b_{i+1}$ is constructed by rearranging the entries in $i$th row of the matrix $[R_1\ R_2\hdots\ R_i]$. Clearly, $D_i$ contains $a_i=(k-i)b_{i+1}$ entries.

The encoding for a universal QTS is given as follows:
\begin{eqnarray}
\ket{s_1 s_2\hdots s_m}\ \mapsto\sum_{\underline{r}\in\mathbb{F}_q^{m(k-1)}}
\ \bigotimes_{i=1}^{n}\ \ket{c_{i,1} c_{i,2}\hdots c_{i,m}} \label{eq:enc_qudits_univ_d}
\end{eqnarray}
where $c_{ij}$  is the entry in $C=VY$ from $i$th row and $j$th column.
After encoding, the $i$th set of $m$ qudits is given to the $i$th party.

For example, take $k=3$. The $((k=3,n=5,*))$ scheme will have the following parameters. 
\begin{gather*}
q=11\\
m=\text{lcm}\{1,2,3\}=6\\
w_1=w_2=w_3=w_4=w_5=6\\
d_1=5,d_2=4,d_3=3\\
a_1=2,a_2=3,a_3=6\\
b_1=2,b_2=1,b_3=3
\end{gather*}
Then $C$, the coding matrix for $k=3$ is given as.
\begin{eqnarray*}
\left[
\begin{tabular}{ccccc}
9&3&4&6&1
\\2&9&3&4&6
\\8&2&9&3&4
\\7&8&2&9&3
\\5&7&8&2&9
\end{tabular}
\right]
\left[
\begin{tabular}{cc:c:ccc}
$s_1$ & $s_4$ & 0 & 0 & 0 & 0\\
$s_2$ & $s_5$ & $r_1$ & 0 & 0 & 0\\
$s_3$ & $s_6$ & $r_3$ & $r_2$ & $r_4$ & $r_6$\\\hdashline
$r_1$ & $r_3$ & $r_5$ & $r_7$ & $r_9$ & $r_{11}$\\
$r_2$ & $r_4$ & $r_6$ & $r_8$ & $r_{10}$ & $r_{12}$
\end{tabular}
\right]
\end{eqnarray*}
Here $V$ is a Cauchy matrix as defined in \eqref{eq:def_cauchy_mtx} with $y_1=0$, $y_2=1$, $y_3=2$, $y_4=3$, $y_5=4$, $x_1=5$, $x_2=6$, $x_3=7$, $x_4=8$, $x_5=9$.
The encoding for this $((3,5,*))$ scheme is then given by \eqref{eq:enc_qudits_univ_d}. Note that each entry in matrix $C$, $c_{ij}$ is a function of $\underline{s}$ and $\underline{r}$. 
However, $D_i$ are functions of $\underline{r}$ alone.
For a detailed description of this scheme, refer to the appendix in \cite{senthoor20a}.

Our encoding matrix is somewhat similar to the matrix used in \cite{bitar18}. 
However, there are some minor structural differences. 
Since we are encoding quantum states in superposition, there is no need for 
generating random bits. 
Furthermore, due to the no-cloning theorem, the total number of parties cannot exceed 
$2k-1$.

\subsection{Reconstruction of the secret}\label{ss:iv_b}
The combiner can reconstruct the secret depending upon the choice of $d$. 
Once $d=d_i$ is chosen, the combiner contacts a set of any $d_i$ parties to reconstruct the secret. 
Each of the contacted party sends $a_i=\frac{m}{d_i-k+1}$ qudits to the combiner. 
In total, the combiner has $\frac{d_im}{d_i-k+1}=a_id_i$ qudits.

With respect to the $((3,5,*))$ example in the previous section, suppose that the third party is contacted for reconstruction. 
If the party  belongs to recovery set of size $d_1=5$, then $a_1=2$ qudits are communicated to the combiner. 
Similarly, if  $d_2=4$, then $a_2=3$  and if $d_3=3$, then $a_3=6$ qudits are sent. 

The secret reconstruction happens in two stages. First, the basis states of the secret are reconstructed through suitable unitary operations. 
The classical secret sharing schemes stop the reconstruction at this point. 
But, the qudits containing the basis states of the secret can be entangled with the remaining qudits. 
So, in the second stage, the secret is extracted into a set of  qudits that are disentangled with the remaining qudits. 
\begin{lemma}[Secret recovery]\label{lm:recovery}
For a $((k,2k-1,*))$ scheme with the encoding  
given in \eqref{eq:enc_qudits_univ_d}, we can recover the secret from any 
$d=2k-i$ shares where $1\leq i\leq k$ by downloading only the first $a_i=\frac{m}{d-k+1}$ qudits from each share where 
$m$ is as given in \eqref{eq:ceqts-params}.
\end{lemma}
\begin{proof} 
Each of the $d$ participants sends their first $a_i$ qudits to the combiner for reconstructing the secret. 
Let $D = \{j_1, j_2, \hdots, j_d\} \subseteq \{1,2,\hdots,2k-1\}$ be the set of $d$ shares chosen and $E=\{j_{d+1},j_{d+2},\hdots,j_{2k-1}\}$ be the complement of $D$. 
Then, \eqref{eq:enc_qudits_univ_d} can be rearranged as
\begin{eqnarray}
\sum_{\underline{r}\in\mathbb{F}_q^{m(k-1)}}
&&\textcolor{blue}{\ket{c_{j_1,1}c_{j_2,1}...c_{j_d,1}}
\ket{c_{j_1,2}c_{j_2,2}...c_{j_d,2}}}\nonumber
\\[-0.2in]&&\textcolor{blue}{\ \ \ \ \ \ \hdots\ket{c_{j_1,a}c_{j_2,a}...c_{j_d,a}}}\nonumber
\\&&\ \ \ket{c_{j_{d+1},1}c_{j_{d+2},1}...c_{j_n,1}}
\ket{c_{j_{d+1},2}c_{j_{d+2},2}...c_{j_n,2}}\nonumber
\\&&\ \ \ \ \ \ \ \ \hdots\ket{c_{j_{d+1},a}c_{j_{d+2},a}...c_{j_n,a}}\nonumber
\\&&\ \ \ \ \ket{c_{1,a+1}c_{2,a+1}...c_{n,a+1}}
\ket{c_{1,a+2}c_{2,a+2}...c_{n,a+2}}\nonumber
\\&&\ \ \ \ \ \ \ \ \ \ \ \ \hdots\ket{c_{1,m}c_{2,m}...c_{n,m}}
\label{eq:acc_qudits}
\end{eqnarray}
where we have highlighted (in blue) the qudits communicated to the combiner.
For the sake exposition we will first cover the case of $i=1$ \textit{i.e.} $d_i=2k-1$ where all the parties are contacted for their first 
$a_1$ qudits by the combiner. 
\\\\\textit{Case (i): $i=1$}%\vspace{0.2cm}
\\For $i=1$, $d=2k-1=n$. Now \eqref{eq:acc_qudits} can be rewritten as
\begin{eqnarray*}
\sum_{\underline{r}\in\mathbb{F}_q^{m(k-1)}}
\textcolor{blue}{\ket{V(S,R_1)}}&&\ket{V(0,D_1,R_2)}\ket{V(0,D_2,R_3)}\nonumber
\\[-0.5cm]&&\ \ \ \ \ \ \ \ \ \ \ \ \ \ \ \ \ \ \ \ \hdots\ket{V(0,D_{k-1},R_k)}
\end{eqnarray*}
where we slightly abused the notation. By $A(B_1,B_2,B_3)$ we actually refer to the 
matrix product $A\left[\begin{array}{ccc} B_1^t & B_2^t & B_{3}^t\end{array} \right]^t$.

Since $V$ is an $n\times n$ Cauchy matrix and, we can apply ${V}^{-1}$ to the state 
$\ket{V(S, R_1)}$, to obtain 
\begin{eqnarray*}
\textcolor{blue}{\ket{S}}\sum_{\underline{r}\in\mathbb{F}_q^{m(k-1)}}
\textcolor{blue}{\ket{R_1}}&&\ket{V(0,D_1,R_2)}\ket{V(0,D_2,R_3)}\nonumber
\\[-0.5cm]&&\ \ \ \ \ \ \ \ \ \ \ \ \ \ \ \ \ \ \ \ \ \ \ \hdots\ket{V(0,D_{k-1},R_k)}
\end{eqnarray*}
We can clearly see that the secret is disentangled with the rest of the qudits.
Therefore, we can recover arbitrary superpositions also. 

\noindent
\\\\\textit{Case (ii): $2\leq i\leq k$: } Under this case,
the state of the system is as follows. (This is the same as
\eqref{eq:acc_qudits}, only the qudits in possession of the combiner have been highlighted.)
\begin{eqnarray*}
\sum_{\underline{r}\in\mathbb{F}_q^{m(k-1)}}
&&\!\!\!\textcolor{blue}{\ket{V_D(S,R_1)}\ \ket{V_D(0,D_1,R_2)}\hdots\ket{V_D(0,D_{i-1},R_i)}}\nonumber
\\[-0.5cm]&&\ket{V_E(S,R_1)}\ \ket{V_E(0,D_1,R_2)}\hdots\ket{V_E(0,D_{i-1},R_i)}\nonumber
\\&&\ \ \ket{V(0,D_i,R_{i+1})}\hdots\ket{V(0,D_{k-1},R_k)}\nonumber
\end{eqnarray*}
We can simplify this state using the fact $V_D(0, D_j, R_{j+1}) = V_D^{[j+1,n]} (D_j, R_{j+1})$.
\begin{eqnarray*}
=\sum_{\underline{r}\in\mathbb{F}_q^{m(k-1)}}
&&\textcolor{blue}{\ket{V_D(S,R_1)}\ket{{V_D}^{[2,2k-1]}(D_1,R_2)}}\nonumber
\\[-0.5cm]&&\ \ \ \ \ \ \ \ \ \ \ \ \ \ \ \ \ \ \ \ \ \ \ \ \ \ \textcolor{blue}{\hdots\ket{{V_D}^{[i,2k-1]}(D_{i-1},R_i)}}\nonumber
\\&&\ \ \ \ \ket{V_E(S,R_1)}\ \ket{V_E(0,D_1,R_2)}\nonumber
\\&&\ \ \ \ \ \ \ \ \ \ \ \ \ \ \ \ \ \ \ \ \ \ \ \ \ \ \ \ \ \ \hdots\ket{V_E(0,D_{i-1},R_i)}\nonumber
\\&&\ \ \ \ \ \ \ \ \ket{V(0,D_i,R_{i+1})}\hdots\ket{V(0,D_{k-1},R_k)}\nonumber
\end{eqnarray*}
Since ${V_D}^{[i,2k-1]}$ is a $d\times d$ Cauchy matrix, the combiner can apply the inverse of ${V_D}^{[i,2k-1]}$ to 
$\ket{V_D^{[i,n]}(D_{i-1}, R_{i})}$ to transform the state as follows. 
\begin{eqnarray}
\sum_{\underline{r}\in\mathbb{F}_q^{m(k-1)}}\hspace{-0.5cm}
&&\textcolor{blue}{\ket{V_D(S,R_1)}\ \ket{{V_D}^{[2,2k-1]}(D_1,R_2)}}\nonumber
\\[-0.5cm]&&\ \ \ \ \ \ \ \ \textcolor{blue}{\hdots\ket{{V_D}^{[i-1,2k-1]}(D_{i-2},R_{i-1})}\ \ket{D_{i-1}}\ket{R_i}}\nonumber
\\&&\ \ \ \ \ket{V_E(S,R_1)}\ \ket{V_E(0,D_1,R_2)}\hdots\ket{V_E(0,D_{i-1},R_i)}\nonumber
\\&&\ \ \ \ \ \ \ \ \ket{V(0,D_i,R_{i+1})}\hdots\ket{V(0,D_{k-1},R_k)}\nonumber
\end{eqnarray}
Note that the matrix $D_{i-1}$ contains elements from the $(i-1)$th row of $R_{i-1}$.
Rearranging the qudits, we get
\begin{eqnarray*}
\sum_{\underline{r}\in\mathbb{F}_q^{m(k-1)}}\hspace{-0.5cm}
&&\textcolor{blue}{\ket{V_D(S,R_1)}\ket{{V_D}^{[2,2k-1]}(D_1,R_2)}}\nonumber
\\[-0.5cm]&&\ \ \ \ \ \ \ \ \ \ \ \ \ \ \ \ \ \ \ \ \ \ \ \ \ \textcolor{blue}{\hdots\ket{{V_D}^{[i-2,2k-1]}(D_{i-3},R_{i-2})}}\nonumber
\\&&\ \ \textcolor{blue}{\ket{W_{i-1}(D_{i-2},R_{i-1})}\ \ket{D_{i-1}\backslash\{R_{i-1}\}}\ket{R_i}}\nonumber
\\&&\ \ \ \ \ket{V_E(S,R_1)}\ \ket{V_E(0,D_1,R_2)}\hdots\ket{V_E(0,D_{i-1},R_i)}\nonumber
\\&&\ \ \ \ \ \ \ket{V(0,D_i,R_{i+1})}\hdots\ket{V(0,D_{k-1},R_k)}\nonumber
\end{eqnarray*}
where $D_\ell\backslash\{R_j,R_{j+1},\hdots,R_\ell\}$ indicates a vector with entries from $D_\ell$ which are not in the matrices $R_j,R_{j+1},\hdots,R_\ell$.

Here $W_\ell = [{{V_D}^{[\ell,2k-1]}}^t\ \underline{w}_{\ell,k+1}\ \underline{w}_{\ell,k+2}\hdots\underline{w}_{\ell,k+i-\ell}]^t$ for $1\leq\ell\leq i-1$ where $\underline{w}_{\ell,j}$ is a column vector of length $(2k-\ell)$ with one in the $j$th position and zeros elsewhere. $W_\ell$ is a $(2k-\ell)\times(2k-\ell)$ full-rank matrix. Clearly,
\begin{equation*}
W_\ell
\left[\begin{array}{c}
D_{\ell-1}\\R_\ell
\end{array}\right]
=\left[\begin{array}{c}
{V_D}^{[\ell,2k-1]}(D_{\ell-1},R_\ell)\\R_{\ell,[\ell,i-1]}
\end{array}\right]
\end{equation*}
Now applying $W_{i-1}^{-1}$ to the state $\ket{W_{i-1}(D_{i-2},R_{i-1})}$, we obtain
\begin{eqnarray*}
\sum_{\underline{r}\in\mathbb{F}_q^{m(k-1)}}\hspace{-0.5cm}
&&\textcolor{blue}{\ket{V_D(S,R_1)}\ket{{V_D}^{[2,2k-1]}(D_1,R_2)}}\nonumber
\\[-0.5cm]&&\ \ \ \ \ \ \ \ \ \ \ \ \ \ \ \ \ \ \ \ \ \ \ \ \ \textcolor{blue}{\hdots\ket{{V_D}^{[i-2,2k-1]}(D_{i-3},R_{i-2})}}\nonumber
\\&&\ \ \ \ \textcolor{blue}{\ket{D_{i-2}}\ket{R_{i-1}}\ \ket{D_{i-1}\backslash\{R_{i-1}\}}\ket{R_i}}\nonumber
\\&&\ \ \ \ \  \ket{V_E(S,R_1)}\ \ket{V_E(0,D_1,R_2)}\hdots\ket{V_E(0,D_{i-1},R_i)}\nonumber
\\&&\ \ \ \ \   \ket{V(0,D_i,R_{i+1})}\hdots\ket{V(0,D_{k-1},R_k)}\nonumber
\end{eqnarray*}
Rearranging the qudits, we obtain,
\begin{eqnarray*}
\sum_{\underline{r}\in\mathbb{F}_q^{m(k-1)}}\hspace{-0.5cm}
&&\textcolor{blue}{\ket{V_D(S,R_1)}\ket{{V_D}^{[2,2k-1]}(D_1,R_2)}}\nonumber
\\[-0.5cm]&&\ \ \ \ \ \ \ \ \ \ \ \ \ \ \ \ \ \ \ \ \ \ \ \ \ \textcolor{blue}{\hdots\ket{W_{i-2}(D_{i-3},R_{i-2})}}\nonumber
\\&&\ \ \textcolor{blue}{\ket{D_{i-2}\backslash R_{i-2}}\ket{R_{i-1}}\ \ket{D_{i-1}\backslash \{R_{i-1},R_{i-2}\}}\ket{R_i}}\nonumber
\\&&\ \ \ \ket{V_E(S,R_1)}\ \ket{V_E(0,D_1,R_2)}\hdots\ket{V_E(0,D_{i-1},R_i)}\nonumber
\\&&\ \ \ \ \ket{V(0,D_i,R_{i+1})}\hdots\ket{V(0,D_{k-1},R_k)}\nonumber
\end{eqnarray*}
Repeating this process for $(D_{i-3},R_{i-2})$ through $(S,R_1)$, by applying the inverses of $W_{i-2}, W_{i-3},\hdots W_{1}$ in successive steps to the suitable sets of qudits and rearranging, we obtain,
\begin{eqnarray*}
\textcolor{blue}{\ket{S}}\!\!\sum_{\substack{\underline{r}\in\\\mathbb{F}_q^{m(k-1)}}}\hspace{-0.5cm}
&&\textcolor{blue}\ \ \bl{\ket{R_1}\ket{R_2}\hdots \ket{R_i}}\nonumber
\\[-0.7cm]&&\ \ \ \ket{V_E(S,R_1)}\ket{V_E(0,D_1,R_2)}\hdots\ket{V_E(0,D_{i-1},R_i)}\nonumber
\\&&\ \ \ \ \ket{V(0,D_i,R_{i+1})}\hdots\ket{V(0,D_{k-1},R_k)}\nonumber
\end{eqnarray*}
The $m$ qudits corresponding to $\ket{S}$ is still entangled with other qudits in the system. 

Since $D_{i-1}$ is formed by entries from the $(i-1)$th row in $[R_1\ R_2\ \hdots\ R_{i-1}]$, we can rearrange the qudits to obtain
\begin{eqnarray*}
\textcolor{blue}{\ket{S}}\!\!\sum_{\substack{\underline{r}\in\\\mathbb{F}_q^{m(k-1)}}}\hspace{-0.5cm}
&&\ \bl{\ket{R_{1,J_{i-1}}}\ket{R_{2,J_{i-1}}}\hdots \ket{R_{i-1,J_{i-1}}}\ket{D_{i-1}}\ket{R_i}}\nonumber
\\[-0.7cm]&&\ \ \ \ \!\!\!\ket{V_E(S,R_1)}\ \ket{V_E(0,D_1,R_2)}\hdots\ket{V_E(0,D_{i-1},R_i)}\nonumber
\\&&\ \ \ \ \ \ \ \ \ket{V(0,D_i,R_{i+1})}\hdots\ket{V(0,D_{k-1},R_k)}\nonumber
\end{eqnarray*}
where $J_\ell=[k-1]\backslash\{\ell\}$ for $1\leq\ell\leq i-1$.\vspace{0.1cm}

Consider the $(2k-\ell)\times(2k-\ell)$ full-rank matrix
\begin{eqnarray*}
P_\ell=\left[
\renewcommand{\arraystretch}{1.5} % Default value: 1
\begin{tabular}{ccc}
$I_{k-\ell+1}$&
\multicolumn{2}{c}{\large 0}
\\\hdashline
\multicolumn{3}{c}{$V_E^{[\ell,2k-1]}$}
\\\hdashline
\multicolumn{2}{c}{\large \ 0\ }&
$I_{k-i}$
\end{tabular}
\right]
\end{eqnarray*}
where $1\leq\ell\leq i-1$. Apply $P_{i-1}$ on $\ket{D_{i-1}}\ket{R_i}$ to obtain
\begin{eqnarray*}
\textcolor{blue}{\ket{S}}\!\!\sum_{\substack{\underline{r}\in\\\mathbb{F}_q^{m(k-1)}}}\hspace{-0.5cm}
&&\ \bl{\ket{R_{1,J_{i-1}}}\ket{R_{2,J_{i-1}}}\hdots \ket{R_{i-1,J_{i-1}}}}\nonumber
\\[-0.7cm]&&\ \ \ \ \ \ \ \ \ \ \bl{\ket{D_{i-1}}\ket{V_E(0,D_{i-1},R_i)}\ket{R_{i,[i,k-1]}}}\nonumber
\\&&\ \ \ \ \!\!\!\ket{V_E(S,R_1)}\ \ket{V_E(0,D_1,R_2)}\hdots\ket{V_E(0,D_{i-1},R_i)}\nonumber
\\&&\ \ \ \ \ \ \ \ \ket{V(0,D_i,R_{i+1})}\hdots\ket{V(0,D_{k-1},R_k)}\nonumber
\end{eqnarray*}
Now, this can be rearranged to get
\begin{eqnarray*}
\textcolor{blue}{\ket{S}}\sum_{\substack{(R_1,R_2,\hdots R_{i-1},\\R_{i,[i,k-1]},\\R_{i+1}\hdots R_k)\nonumber
\\\in\mathbb{F}_q^{m(k-1)-(i-1)b_i}}}\hspace{-0.5cm}
&&\textcolor{blue}{\ket{R_1,R_2,\hdots ,R_{i-1}}\ \ket{R_{i,[i,k-1]}}}\nonumber
\\[-1.4cm]&&\ \ \ket{V_E(S,R_1)}\ \ket{V_E(0,D_1,R_2)}\nonumber
\\&&\ \ \ \ \ \ \ \ \ \ \ \ \ \ \ \ \ \ \ \ \ \ \hdots\ket{V_E(0,D_{i-2},R_{i-1})}\nonumber
\\&&\ \ \ \ \ \ \ket{V(0,D_i,R_{i+1})}\hdots\ket{V(0,D_{k-1},R_k)}\nonumber
\\\sum_{\substack{R_{i,[1,i-1]}\\\in\mathbb{F}_q^{(i-1)\times b_i}}}&&\ket{V_E(0,D_{i-1},R_i)}\textcolor{blue}{\ket{V_E(0,D_{i-1},R_i)}}
\end{eqnarray*}
\begin{eqnarray*}
=\textcolor{blue}{\ket{S}}\!\!\!\sum_{\substack{(R_1,R_2,\hdots R_{i-1},\\R_{i,[i,k-1]},\\R_{i+1}\hdots R_k)\nonumber
\\\in\mathbb{F}_q^{m(k-1)-(i-1)b_i}}}\hspace{-0.5cm}
&&\textcolor{blue}{\ket{R_1,R_2,\hdots ,R_{i-1}}\ \ket{R_{i,[i,k-1]}}}\nonumber
\\[-1.4cm]&&\ \ \ket{V_E(S,R_1)}\ \ket{V_E(0,D_1,R_2)}\nonumber
\\&&\ \ \ \ \ \ \ \ \ \ \ \ \ \ \ \ \ \ \ \ \ \ \hdots\ket{V_E(0,D_{i-2},R_{i-1})}\nonumber
\\&&\ \ \ \ \ \ \ket{V(0,D_i,R_{i+1})}\hdots\ket{V(0,D_{k-1},R_k)}\nonumber
\\&&\ \ \ \ \ \ \sum_{T_i\in\mathbb{F}_q^{(i-1)\times b_i}}\ket{T_i}\textcolor{blue}{\ket{T_i}}
\end{eqnarray*}
because the state
\begin{equation*}
\sum_{\substack{R_{i,[1,i-1]}\\\in\mathbb{F}_q^{(i-1)\times b_i}}}\ket{V_E(0,D_{i-1},R_i)}\textcolor{blue}{\ket{V_E(0,D_{i-1},R_i)}}    
\end{equation*}
is a uniform superposition of states $\ket{T_i}\ket{T_i}$ over $T_i\in\mathbb{F}_q^{(i-1)\times b_i}$ independent of the value of $D_{i-1}$ and $R_{i,[i,k-1]}$.

Repeating these operations with all $\ket{R_j}$ for $1\leq j\leq i-1$, we obtain,
\begin{eqnarray*}
\textcolor{blue}{\ket{S}}\sum_{\substack{(R_{i+1}\hdots R_k)\\\in\mathbb{F}_q^{(m-a_i)(k-1)}\\(R_{1,[i,k-1]},\hdots R_{i,[i,k-1]})\\\in\mathbb{F}_q^{(k-i)a_i}}}\hspace{-0.5cm}
&&\textcolor{blue}{\ket{R_{1,[i,k-1]},R_{2,[i,k-1]},\hdots R_{i,[i,k-1]}}}\nonumber
\\[-1.4cm]&&\ \ \ket{V(0,D_i,R_{i+1})}\hdots\ket{V(0,D_{k-1},R_k)}\nonumber
\\\nonumber\\\nonumber
\\&&\hspace{-2.4cm}\sum_{\substack{T_1\in\\\mathbb{F}_q^{(i-1)\times b_1}}}\ket{T_1}\textcolor{blue}{\ket{T_1}}\sum_{\substack{T_2\in\\\mathbb{F}_q^{(i-1)\times b_2}}}\ket{T_2}\textcolor{blue}{\ket{T_2}}\hdots\sum_{\substack{T_i\in\\\mathbb{F}_q^{(i-1)\times b_i}}}\ket{T_i}\textcolor{blue}{\ket{T_i}}
\end{eqnarray*}
At this point the secret is completely disentangled with the rest of the qudits and the recovery is complete.
\end{proof}

\subsection{Secrecy}\label{ss:iv_c}
In the scheme given by \eqref{eq:enc_qudits_univ_d}, the combiner can recover the secret by accessing $k$ parties (from case (ii) when $i=k$ in the proof of Lemma \ref{lm:recovery}). So, by no-cloning theorem, the remaining $k-1$ parties in the scheme should have no information about the secret. Thus, this scheme satisfies the secrecy property.
With these results in place we have our central contribution. 

\begin{theorem}[Staircase construction for universal CE-QTS]
The encoding given in \eqref{eq:enc_qudits_univ_d} gives a $((k,n=2k-1,*))$ universal CE-QTS scheme with the following parameters.
\label{th:staircase-constn-univ-ceqts}
\begin{gather*}
%\label{eq:sc-univ-ceqts}. 
q\geq 2(2k-1)\text{ (prime)}
\\m=\text{lcm}\{1,2,\hdots,k\}
\\w_1=w_2=\hdots=w_n=m
\\\text{CC}_n(d)=\frac{dm}{d-k+1} \text{ for }d\in\{k,k+1,\hdots,2k-1\}
\end{gather*}
\end{theorem}
We can compare this universal CE-QTS scheme with the scheme from Corollary \ref{co:ramp-univ-ceqts-constn}. (Refer Table \ref{tab:contributions}.) Both give the same values for the parameters $w_j/m$ and $CC_n(d)/m$.
However the concatenated construction could give a smaller secret size for $n<2k-1$.

\subsection{Discussion on communication complexity gains}
\label{ss:iv_d}
In the standard $((k,n))$ QTS scheme from \cite{cleve99}, the secret can be recovered when the combiner communicates with  $k$ parties. 
Here, if the secret is of size $m$ qudits, then the number of qudits communicated to  the combiner  is $km$ qudits.
The communication cost per secret qudit is $k$ qudits.

In the $((k,n,d))$ communication efficient QTS schemes from \cite{senthoor19} and concatenated construction in Corollary \ref{co:ramp-fixed-ceqts-constn}, the secret can be recovered when the combiner contacts $k$ parties and receiving $km$ qudits where $m=d-k+1$.
This leads to  a cost of $k$ qudits per secret qudit.
However, when the combiner contacts $d$ parties, where $d$ is a fixed value such that $k<d\leq n$, 
the secret can be recovered with a communication cost of $\frac{dm}{d-k+1}$ qudits.
The cost per qudit is $\frac{d}{d-k+1}$ which is strictly less than $k$.

In the $((k,n,*))$ universal QTS schemes, the secret can be recovered by the combiner by accessing any $d$ parties, where the number of parties accessed given by $k\leq d\leq n$ can also be chosen by the combiner. For the chosen value of $d$, the secret can be recovered by downloading $\frac{d m}{d-k+1}$ qudits.
The communication cost for each qudit of the secret is $\frac{d}{d-k+1}$ which is same as that of \cite{senthoor19}.
The communication cost decreases with the increasing number of parties accessed.
(Refer Fig. \ref{fig:di-vs-cc}.)
However, we are able to achieve this for all possible $d$ using the same scheme and not fixing $d$ a priori.
Refer Table \ref{tab:contributions} for a comparison of the different QTS schemes we discussed so far.
The optimality of our construction with respect to the communication complexity will be discussed in the next section.
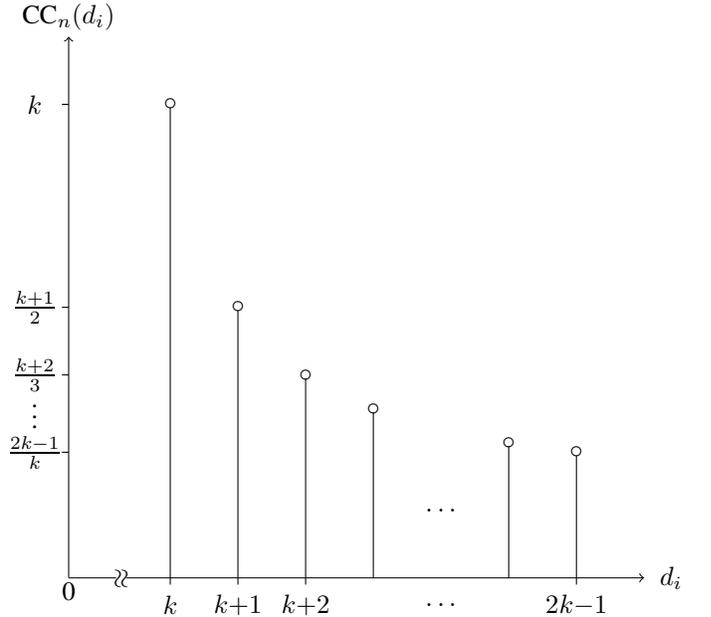
\begin{figure}[ht]
\begin{center}
\hspace{-0.5cm}
\begin{tikzpicture}[scale=0.9, every node/.style={scale=1}]
%\begin{scope}[xshift=0cm,yshift=0cm]
\draw (5.5,0) -- (6.2,0);
\node at (6.2,0) {\rotilde};
\node at (6.3,0) {\rotilde};
\draw[->] (6.35,0) -- (14,0);
\node at (14.4,0) {$d_i$};
\draw[->] (5.5,0) -- (5.5,8);
\node at (5.5,8.3) {$\text{CC}_n(d_i)$};
\node at (5.5,-0.2) {0};
%d=7
\draw (7,0) -- (7,-0.1);
\node at (7,-0.4) {$k$};
\draw (5.5,7) -- (5.4,7);
\node at (5,7) {$k$};
\node at (7,7) {$\circ$};
\draw (7,0) -- (7,6.95);
%d=8
\draw (8,0) -- (8,-0.1);
\node at (8,-0.4) {$k$$+$$1$};
\draw (5.5,4) -- (5.4,4);
\node at (5,4) {$\frac{k+1}{2}$};
\node at (8,4) {$\circ$};
\draw (8,0) -- (8,3.95);
%d=9
\draw (9,0) -- (9,-0.1);
\node at (9,-0.4) {$k$$+$$2$};
\draw (5.5,3) -- (5.4,3);
\node at (5,3) {$\frac{k+2}{3}$};
\node at (9,3) {$\circ$};
\draw (9,0) -- (9,2.95);
\node at (5,2.5) {$\vdots$};
%d=10
\node at (10,2.5) {$\circ$};
\draw (10,0) -- (10,2.45);
%d=11
%\node at (11,2.2) {$\circ$};
%\draw (11,0) -- (11,2.15);
\node at (11,-0.4) {$\hdots$};
\node at (11,1.0) {$\hdots$};
%d=12
\node at (12,2) {$\circ$};
\draw (12,0) -- (12,1.95);
%d=13
\draw (13,0) -- (13,-0.1);
\node at (13,-0.4) {$2k$$-$$1$};
\draw (5.5,1.857) -- (5.4,1.857);
\node at (5,1.857) {$\frac{2k-1}{k}$};
\node at (13,1.857) {$\circ$};
\draw (13,0) -- (13,1.807);
%\end{scope}
\end{tikzpicture}
\captionsetup{justification=justified}
\caption{Communication cost for $d\in\{k,k+1,\hdots,n\}$ in $((k,n=2k-1,*))$ universal CE-QTS schemes from Concatenated construction and Staircase construction}
\label{fig:di-vs-cc}
\end{center}
\end{figure}

\section{Optimality of CE-QTS schemes}
\label{s:v}
In this section, we derive lower bounds on the quantum communication complexity of the quantum threshold schemes. 
Our bounds are applicable for both universal and non-universal communication efficient schemes. 
Specifically, we show  that  secret recovery from a set of $d$ shares in communication efficient QTS schemes (for both fixed $d$ and universal) requires at least $\frac{d}{d-k+1}$ qudits to be transmitted to the combiner for each qudit in the secret. Then we show that our constructions satisfies these bounds on communication complexity. We also discuss the optimality of our constructions with respect to the storage cost. 

\subsection{Lower bound on communication complexity}
Bound on communication complexity for the $((k,n,d))$ CE-QTS was already shown in \cite{senthoor19} for the special case of $n=2k-1$. 
Here we generalize these bounds to both $((k,n,d))$ and $((k, n, *))$ QTS and also lift the restriction that $n=2k-1$.
We first bound the combined size of partial shares from $d-k+1$ parties. The generalization of the result from $n=2k-1$ to $n\leq 2k-1$ is mainly due to a difference in our approach to prove this bound. Then we use this to prove the bound on the communication cost \textit{i.e.} the combined size of partial shares from all $d$ parties in a way similar to \cite{senthoor19}.
Our bounds imply that the proposed CE-QTS and universal CE-QTS constructions for all $n\leq 2k-1$ are optimal with respect to the communication cost. 
First we need the following lemmas.
\begin{lemma}%[Secret replacement with authorized set]
\label{lm:sec-rep}
\cite[Theorem~5]{gottesman00}
A party having access to an authorized set of shares in a quantum secret sharing scheme can replace the secret encoded with any arbitrary state (of the same dimension as the secret) without disturbing the remaining shares. After this replacement, secret recovery from any of the authorized sets will give only the new state.
\end{lemma}
\begin{lemma}%[Minimal dimension of quantum channel for transmission] 
\label{lm:lim-qc}
\cite[Lemma~5]{senthoor19}
Even in the presence of pre-existing entanglement between two parties, transmitting an arbitrary quantum state from a Hilbert space of dimension $M$ requires a channel of dimension $M$.
\end{lemma}

With these two lemmas we can bound the  combined size of partial shares from $d-k+1$ parties in the secret recovery from $d$ parties.
\begin{lemma}%[Lower bound on part of communication cost]
\label{lm:bound-helps}
In any $((k,n))$ QSS scheme, which recovers a secret of dimension $M$ by accessing a set of $d$ parties, the total communication to the combiner from any $d-k+1$ parties among the $d$ parties is of dimension at least $M$.
\end{lemma}

\begin{proof}
Let $S_1, S_2,\hdots, S_n$ be the shares of the $n$ parties in the $((k,n))$ QSS scheme. By Lemma~\ref{lm:mixed-to-pure}, consider an extra share $E$ for the given $((k,n))$ scheme such that the new QSS scheme with $n+1$ parties thus obtained is a pure state QSS scheme. (This pure state QSS scheme need not be a threshold QSS scheme.) Now, we prove the lemma by means of a communication protocol between Alice and Bob based on this pure state QSS scheme. The objective of the protocol is for Alice to send an arbitrary state $\ket{\psi}$ of dimension $M$ to Bob.

First, encode the state $\ket{0}$ using the pure state QSS scheme. Consider the set of $d$ parties $D\subseteq [n]$ where each participant in $D$ can send a part of its share to the combiner to recover the secret.
Consider any subset $L\subseteq D$ with $d-k+1$ parties. Bob is given the $k-1$ shares from the parties in $D\backslash L$, which form an unauthorized set.
Alice is given the $d-k+1$ shares from $L$, the $n-d$ shares from $[n]\backslash D$ and the extra share $E$. By Lemma~\ref{lm:pu-auth}, the set of shares with Alice form an authorized set, as this set is actually a complement of the unauthorized set with Bob.

Now, Alice replaces the secret $\ket{0}$ in the scheme with $\ket{\psi}$ (by Lemma~\ref{lm:sec-rep}).
Clearly, Bob has no prior information on $\ket{\psi}$ even though he may share some entanglement with Alice due to qudits he received so far.

Now, if Alice needs to transmit $\ket{\psi}$ to Bob, she needs to transmit some of the qudits with her to Bob so that Bob can use the secret recovery of the underlying QSS scheme to recover $\ket{\psi}$.
To achieve this, Alice can transmit to Bob the necessary parts from the $d-k+1$ shares from $L$ (which along with necessary parts from the $k-1$ shares from $D\backslash L$ already with him will give complete information about the secret).
Applying Lemma~\ref{lm:lim-qc} here, it is implied that the communication from the shares in $L$ during the secret recovery from the shares in $D$ has to be at least $M$.
\end{proof}

Next we use Lemma~\ref{lm:bound-helps} to obtain a lower bound on communication complexity of $d$ partial shares. We use the same technique as in \cite{senthoor19} to achieve this.
\begin{theorem}[Lower bound on communication cost]\label{th:co-lbounds} % on communication cost
In any $((k,n))$ quantum secret sharing scheme, recovery of a secret of dimension $M$ by accessing $d$ parties requires communication of a state from a Hilbert space of dimension at least $M^{d/(d-k+1)}$ to the combiner.
\end{theorem}
\begin{proof}
Consider a set of $d$ parties given by $D\subseteq[n]$ accessed by the combiner for secret recovery.
For each $i\in D$, let the part of the share transmitted by the $j$th party to the combiner be denoted as $H_{j,D}$. Clearly $H_{j,D}$ is a subsystem of $S_j$.

Without loss of generality, we take the set of parties to be given by $D=\{1,2,\hdots,d\}$ such that 
\begin{eqnarray}
\dim(H_{1,D})\geq\dim(H_{2,D})\geq\hdots\geq\dim(H_{d,D}).
\label{eq:share-sizes}
\end{eqnarray}

Applying Lemma \ref{lm:bound-helps} for the partial shares $H_{k,D}, H_{k+1,D},\hdots H_{d,D}$ sent to the combiner, the overall communication from these $d-k+1$ shares is bounded as 
\begin{eqnarray}
\prod_{j=k}^d \dim(H_{j,D})\geq M.
\label{eq:bound_some_helps}
\end{eqnarray}
Then by \eqref{eq:share-sizes}, we have 
\begin{eqnarray*}
\dim(H_{k,D})^{d-k+1}\geq M\nonumber
\\\dim(H_{k,D})\geq M^{1/(d-k+1)}.
\end{eqnarray*}
This implies
\begin{eqnarray}
\dim(H_{j,D}) &\geq& M^{1/(d-k+1)} \label{eq:bound_one_help}
\end{eqnarray}
for $1\leq j\leq k$. 
From \eqref{eq:bound_some_helps}~and~\eqref{eq:bound_one_help}, the communication to the combiner from the $d$ shares in $D$ can be lower bounded as
\begin{eqnarray*}
\prod_{j=1}^d\dim(H_{j,D})&=&\prod_{j=1}^{k-1}\dim(H_{j,D})\prod_{j=k}^d\dim(H_{j,D})\nonumber
\\&\geq&\bigg(\prod_{j=1}^{k-1}M^{1/(d-k+1)}\bigg)M\nonumber
\\&=&M^{d/d-k+1)}
\end{eqnarray*}
This shows that the set of $d$ parties $D$ must communicate a state that is in a Hilbert space of dimension at least $M^{d/(d-k+1)}$.
\end{proof}
When the combiner accesses any $k$ parties, the above theorem implies that the total quantum communication to the combiner must be of at least dimension $M^k$.
Note that the standard $((k,n))$ QTS schemes provided in \cite{cleve99} achieve this lower bound.

In the next subsection, we use this bound to evaluate the performance of our constructions for CE-QTS schemes.

\subsection{Optimality of the proposed schemes}
The bound on the dimension of the communication cost in Theorem \ref{th:co-lbounds} can be used to obtain a bound on the communication cost in terms of qudits. \begin{corollary}
In a $((k,n,d))$ CE-QTS scheme sharing a secret of $m$ qudits, the communication cost is bounded as
\begin{equation*}
    \text{CC}_n(d)\geq \frac{dm}{d-k+1}.
\end{equation*}
\label{lm:ce-qts-opt}
\end{corollary}
\begin{proof}
Let $q$ be the dimension of each qudit in the scheme. Clearly, the dimension of the secret $M=q^m$. By Theorem \ref{th:co-lbounds}, the communication from any $d$ shares is going to be at least $q^\frac{dm}{d-k+1}$. Thus, the $d$ parties need to send at least $\frac{dm}{d-k+1}$ qudits for recovering the secret.
\end{proof}
Recall from Lemma \ref{lm:qts-opt} that, for any QTS scheme, the share size is lower bounded by the size of the secret \textit{i.e.} for all $1\leq j\leq n$
\begin{eqnarray*}
w_j\geq m.
\end{eqnarray*}
\begin{remark}%[Optimality of Staircase CE-QTS]
$((k,n,d))$ CE-QTS scheme from \cite{senthoor19} based on Staircase codes has optimal storage cost and optimal communication cost.
\end{remark}
\begin{remark}%[Optimality of Concatenation CE-QTS]
$((k,n,d))$ CE-QTS scheme from Corollary \ref{co:ramp-fixed-ceqts-constn} based on ramp schemes from \cite{ogawa05} has optimal storage cost and optimal communication cost.
\end{remark}
Note that these bounds apply for both fixed $d$ and universal CE-QTS schemes.
\begin{corollary}
In a $((k,n,*))$ universal CE-QTS scheme sharing a secret of $m$ qudits, for any $d$ such that $k\leq d\leq n$, the communication cost is bounded as
\begin{equation*}
\text{CC}_n(d)\geq \frac{dm}{d-k+1}.
\end{equation*}
\end{corollary}
\begin{remark}%[Optimality of Staircase universal CE-QTS]
$((k,n,*))$ universal CE-QTS scheme from Theorem \ref{th:staircase-constn-univ-ceqts} based on Staircase codes has optimal storage cost and optimal communication cost.
\end{remark}
\begin{remark}%[Optimality of Concatenation universal CE-QTS]
$((k,n,*))$ universal CE-QTS scheme from Corollary \ref{co:ramp-univ-ceqts-constn} based on ramp schemes from \cite{ogawa05} has optimal storage cost and optimal communication cost.
\end{remark}
\par In the following section, we prove the bound on communication cost of CE-QTS schemes using a quantum information theoretic approach.

\section{Information theoretic model of CE-QTS}
\label{s:vi}

The storage cost and the communication complexity required for secret sharing schemes can also be studied with information theory. For classical threshold schemes, such results have been obtained in \cite{karnin83},\cite{wang08},\cite{huang16}. 
In this section, we will be using quantum information theory to develop a framework to get similar results for communication efficient quantum threshold schemes building upon the work by Imai et al\cite{imai03}. 
We propose a quantum information theoretic framework for CE-QTS schemes and use this to study their communication complexity.
We refer the reader to Section~\ref{sec:bg} for some of the definitions and terms. 

\subsection{Information theoretic model for quantum secret sharing}
Let $\mathcal{S}$ be the quantum secret from the Hilbert space $\mathcal{H}_\mathcal{S}$ of dimension $M$. Then the density matrix corresponding to $\mathcal{S}$ can be defined as
\begin{equation}
\rho_\mathcal{S}=\sum_{i=0}^{M-1}p_i\ketbra{\psi_i}{\psi_i}
\nonumber
\end{equation}
where $\{p_i\}$ gives the probability distribution in a measurement over some basis of orthonormal states $\{\ket{\psi_0},\ket{\psi_1},\hdots,\ket{\psi_{M-1}}\}$. Let $\mathcal{R}$ be the reference system such that the combined system $\mathcal{RS}$ is in the pure state
\begin{equation}
\ket{\Psi_\mathcal{RS}}=\sum_{i=0}^{M-1}\sqrt{p_i}\ket{\psi_i}_\mathcal{R}\ket{\psi_i}_\mathcal{S}.
%\nonumber
\end{equation}
Clearly $\mathsf{S}(\mathcal{RS})=0$. Then, by Araki-Lieb inequality, $\mathsf{S}(\mathcal{R})=\mathsf{S}(\mathcal{S})$.

Let $S_1,S_2,\hdots,S_n$ be the quantum systems corresponding to the $n$ shares defined over the Hilbert spaces $\mathcal{H}_1,\mathcal{H}_2,\hdots,\mathcal{H}_n$ respectively. Then the encoding of the secret is given by the encoding map $\mathcal{E}:\mathcal{H}_\mathcal{S}\rightarrow\mathcal{H}_1\otimes\mathcal{H}_2\otimes\cdots\otimes\mathcal{H}_n$. A subset of $\ell$ parties can be indicated by the set $L\subseteq[n]$ corresponding to their indices. The combined system of these parties are then denoted as
\begin{equation}
S_L=S_{i_1}S_{i_2}\hdots S_{i_\ell}
\nonumber
\end{equation}
where $L=\{i_1,i_2,\hdots,i_\ell\}$ with $i_1<i_2<\hdots<i_\ell$. Then the density matrix of the $i$th party for $i\in[n]$ can be written as
\begin{equation}
\rho_i=\text{tr}_{S_{[n]\backslash i}}\mathcal{E}(\rho_\mathcal{S}).
\nonumber
\end{equation}

With these notations, we can define the necessary and sufficient information theoretic conditions for a quantum secret sharing scheme.
From the quantum data processing inequality in Lemma \ref{lm:qdpi}, for any subset of parties given by $A\subseteq[n]$, we can say
\begin{gather*}
\mathsf{S}(\mathcal{S})\geq\mathsf{S}(S_A)-\mathsf{S}(\mathcal{R}S_A).
\end{gather*}
If $A$ is assumed to be an authorized set, then equality holds in the above equation and using the fact that 
$\mathsf{S}(\mathcal{R}\mathcal{S})=0$, we obtain the following. 
\begin{gather*}
%\mathsf{S}(\mathcal{S})=\mathsf{S}(S_A)-\mathsf{S}(\mathcal{R}S_A)\\
%\mathsf{S}(\mathcal{R})+\mathsf{S}(\mathcal{S})=\mathsf{S}(\mathcal{R})+\mathsf{S}(S_A)-\mathsf{S}(\mathcal{R}S_A)\\
I(\mathcal{R}:S_A)=I(\mathcal{R}:\mathcal{S}).
\end{gather*}

By Lemma \ref{lm:mixed-to-pure}, any QSS scheme with $n$ shares given by $S_1,S_2,\hdots,S_n$ can be used to construct a pure state QSS scheme by adding an extra share, say $S_{n+1}$. Any unauthorized set $B\subseteq[n]$ in the original QSS scheme is also an unauthorized set in the pure state scheme. Therefore, by Lemma \ref{lm:pu-auth}, $[n+1]\backslash B$ gives an authorized set in the pure state scheme. 
Hence, by the quantum data processing inequality in Lemma \ref{lm:qdpi}, the secrecy condition in Definition \ref{de:qss} is equivalent to
\begin{gather}
\mathsf{S}(\mathcal{S})=\mathsf{S}(S_{[n+1]\backslash B})-\mathsf{S}(\mathcal{R}S_{[n+1]\backslash B}).
\label{eq:qit-secrecy}
\end{gather}
It can be seen that after encoding of the secret by the pure state QSS scheme, system $\mathcal{R}S_{[n+1]}$ is in the pure state
\begin{equation*}
\sum_{i=0}^{M-1}\sqrt{p_i}\ket{\psi_i}_\mathcal{R}\ket{\overline{\psi_i}}_{S_{[n+1]}}.
\end{equation*}
Here $\ket{\overline{\psi_i}}$ indicates the encoded pure state of the system $S_{[n+1]}$ when the secret $\mathcal{S}$ is in the pure state $\ket{\psi_i}$. If the secret $\mathcal{S}$ is in mixed state, after encoding, $S_{[n+1]}$ will also be in mixed state but $\mathcal{R}S_{[n+1]}$ will be still in pure state.
Clearly $\mathsf{S}(\mathcal{R}S_{[n+1]})=0$. Therefore, by Araki-Lieb inequality,
\begin{gather*}
\mathsf{S}(S_{[n+1\backslash B})=\mathsf{S}(\mathcal{R}S_B),\\ \mathsf{S}(\mathcal{R}S_{[n+1]\backslash B})=\mathsf{S}(S_B).
\end{gather*}
Also $\mathsf{S}(\mathcal{S})=\mathsf{S}(\mathcal{R})$. Substituting these in \eqref{eq:qit-secrecy}, we obtain
\begin{gather*}
%\mathsf{S}(\mathcal{\mathcal{R}})=\mathsf{S}(\mathcal{R}S_B)-\mathsf{S}(S_B)\\
I(\mathcal{R}:S_B)=0.
\end{gather*}
With these conditions, the quantum information theoretic definition for a QSS scheme is given below.
\begin{definition}
\label{de:info-model-qts}
A quantum secret sharing scheme for an access structure $\Gamma$ is a quantum operation which encodes the quantum secret $\mathcal{S}$ into shares $S_1, S_2, \hdots, S_n$ where
\begin{itemize}
\item \textit{(Recoverability)} For every authorized set $A\in \Gamma$,
\begin{equation}
I(\mathcal{R}:S_A)=I(\mathcal{R}:\mathcal{S}),
\label{eq:rec-constraint}
\end{equation}
\item \textit{(Secrecy)} For every unauthorized set $B\notin\Gamma$,
\begin{equation}
I(\mathcal{R}:S_B)=0.
\label{eq:sec-constraint}
\end{equation}
\end{itemize}
\end{definition}
The same definition expands to QTS schemes where the authorized set is given by any $A\subseteq[n]$ such that $|A|\geq k$ and the unauthorized set is given by any $B\subset[n]$ such that $|B|\leq k-1$. The following result from \cite{imai03} gives a bound on the entropy of each share.

\begin{lemma}\label{th:entropy-unauth}%[Lower bound on share size]
%\cite[Theorem 6]{imai03}
In any quantum secret sharing scheme realizing an access structure $\Gamma$ for any subsets of parties $A$ and $B$ such that $A,B\notin\Gamma$ but $A\cup B\in\Gamma$ it holds that $\mathsf{S}(S_A)\geq\mathsf{S}(\mathcal{S})$ where $\mathcal{S}$ is the secret being shared.
\end{lemma}
\begin{corollary}
In any $((k,n))$ QTS scheme, the entropy of any share $S_j$ is bounded as
\begin{equation*}
\mathsf{S}(S_j)\geq\mathsf{S}(\mathcal{S})
\end{equation*}
where $S$ is the secret being shared.
\end{corollary}
\begin{proof}
Take $A=\{j\}$ and some $B\subseteq[n]\backslash\{j\}$ such that $|B|=k-1$ in Lemma~\ref{th:entropy-unauth}.
\end{proof}
Definition \ref{de:info-model-qts} is same as the definition in \cite{imai03} for quantum threshold secret sharing schemes. However for communication efficient quantum threshold schemes, more conditions have to be defined for when the combiner recovers the secret from partial shares from $d>k$ parties.

\subsection{Extension of information theoretic model to CE-QTS}
Let $D\subseteq[n]$, where $|D|=d$, give the indices of some $d$ parties being accessed by the combiner for communication efficient recovery. For each $j\in D$, consider a superoperator $\pi_{j,D}$ acting on $S_j$ such that the resultant state $H_{j,D}$ is then transmitted to the combiner. The density matrix for $H_{j,D}$ can be written as
\begin{equation*}
\sigma_{j,D}=\pi_{j,D}(\rho_j).
\end{equation*}
Here $\pi_{j,D}$ is the operator acting on $S_j$. Consider $E\subseteq D$ corresponding to some $e$ of these $d$ parties. The combined system of the partial shares sent to the combiner by these $e$ parties is denoted as
\begin{equation*}
H_{E,D}=H_{j_1,D}H_{j_2,D}\hdots H_{j_e,D}
\end{equation*}
where $E=\{j_1,j_2,\hdots,j_e\}$ with $j_1<j_2<\hdots<j_e$.

Clearly, the number of qudits in $H_{j,D}$ is $\log_q\text{dim}(H_{j,D})$.
Now, CC$_n(d)$ can be written as
\begin{eqnarray}
\text{CC}_n(d)&=&\max_{\substack{D\subseteq[n]\\\text{s.t. }|D|=d}}\ \sum_{j\in D}\ \log_q\text{dim}(H_{j,D})\nonumber
\\\text{CC}_n(d)&\geq&\frac{1}{\log q}\ \max_{\substack{D\subseteq[n]\\\text{s.t. }|D|=d}}\ \sum_{j\in D}\mathsf{S}(H_{j,D}).
\label{eq:cc-d-bound}
\end{eqnarray}
The inequality in \eqref{eq:cc-d-bound} is from the bound on entropy given by \eqref{eq:max-entropy}.

Similarly, the communication cost for secret recovery in a standard $((k,n))$ threshold scheme can be bounded as,
\begin{equation}
\text{CC}_n(k)\geq\frac{1}{\log q}\max_{\substack{A\subseteq[n]\\\text{s.t. }|A|=k}}\ \sum_{i\in A}\mathsf{S}(S_i).
\end{equation}
Now, the following set of constraints can be included to define the model for a communication efficient quantum threshold scheme.
\begin{definition}
A $((k,n,d))$ CE-QTS scheme is a quantum operation which encodes the quantum secret $\mathcal{S}$ into shares $S_1, S_2, \hdots, S_n$ where
\begin{itemize}
\item \textit{(Recoverability from $k$ shares)} For every $A\subseteq[n]$ such that $|A|\geq k$,
\begin{equation}
I(\mathcal{R}:S_A)=I(\mathcal{R}:\mathcal{S}).
\label{eq:ce-rec-constraint}
\end{equation}
\item \textit{(Recoverability from $d$ partial shares)} For every $D\subseteq[n]$ such that $|D|=d$,
\begin{equation}
I(\mathcal{R}:H_{D,D})=I(\mathcal{R}:\mathcal{S})
\label{eq:ce-d-rec-constraint}
\end{equation}
\item \textit{(Secrecy)} For every $B\subset [n]$ such that $|B|<k$,
\begin{equation}
I(\mathcal{R}:S_B)=0.
\label{eq:ce-sec-constraint}
\end{equation}
\item \textit{(Communication efficiency)}
\begin{equation}
\text{CC}_n(d)<\text{CC}_n(k).
\label{eq:ce-constraint}
\end{equation}
\end{itemize}
\end{definition}
\begin{definition}
A $((k,n,*))$ universal CE-QTS scheme is a quantum operation which encodes the quantum secret $\mathcal{S}$ into shares $S_1, S_2, \hdots, S_n$ where
\begin{itemize}
\item \textit{(Recoverability)} For every $D\subseteq[n]$ such that $k\leq|D|\leq n$,
\begin{equation}
I(\mathcal{R}:H_{D,D})=I(\mathcal{R}:\mathcal{S})
\label{eq:uce-rec-constraint}
\end{equation}
\item \textit{(Secrecy)} For every $B\subset[n]$ such that $|B|<k$,
\begin{equation}
I(\mathcal{R}:S_B)=0.
\label{eq:uce-sec-constraint}
\end{equation}
\item \textit{(Universal communication efficiency)}
\begin{equation}
\text{CC}_n(n)<\text{CC}_n(n-1)<\hdots<\text{CC}_n(k+1)<\text{CC}_n(k).
\label{eq:uce-constraint}
\end{equation}
\end{itemize}
\end{definition}
In the above definition for universal CE-QTS, a separate condition for the threshold of $k$ shares is not needed as $d$ can be assumed to take any value from $k$ to $n$. With these definitions, we can bound the communication cost of CE-QTS schemes (both fixed $d$ and universal) using quantum information theoretic inequalities.
In the following theorem, a similar bound on the entropy of partial shares sent to the combiner has been derived. This result is then used to obtain a bound on the communication cost.
%\addtocounter{lemma}{-2}
\begin{theorem}%[Lower bound on part of communication cost]
%\label{lm:bound-helps}
In any $((k,n))$ quantum secret sharing scheme, recovery of a secret of dimension $M$ by accessing $d$ parties requires communication of a state from a Hilbert space of dimension at least $M^{d/(d-k+1)}$ to the combiner.
\end{theorem}
\begin{proof}
Let $D$ represent the indices of the $d$ parties from which partial shares are sent to the combiner. Clearly, $D$ is an authorized set. For simplicity, we will drop the second subscript $D$ in $H_{E,D}$ for any $E\subseteq D$ and write it simply as $H_E$.

Choose some $F\subseteq D$ such that   $|F|=d-k+1$.
By considering $H_D$ as the bipartite quantum system $H_{D\backslash F}H_F$, \eqref{eq:ce-rec-constraint} gives
\begin{eqnarray}
I(\mathcal{R}:H_{D\backslash F}H_F)&=&I(\mathcal{R}:S)\nonumber 
\\\mathsf{S}(H_{D\backslash F}H_F)-\mathsf{S}(\mathcal{R}H_{D\backslash F}H_F)&=&\mathsf{S}(\mathcal{S})-\mathsf{S}(\mathcal{R}\mathcal{S})\ \ \ \ \nonumber
\\\mathsf{S}(H_{D\backslash F}H_F)-\mathsf{S}(\mathcal{R}H_{D\backslash F}H_F)&=&\mathsf{S}(\mathcal{S})
\label{eq:al-to-be-applied}
\end{eqnarray}
Applying the Araki-Lieb inequality to $\mathsf{S}(\mathcal{R}H_{D\backslash F}H_F)$ gives
\begin{equation*}
\mathsf{S}(\mathcal{R}H_{D\backslash F}H_F)\geq\mathsf{S}(RH_{D\backslash F})-\mathsf{S}(H_F).
%\label{eq:al-applied}
\end{equation*}
Applying this in \eqref{eq:al-to-be-applied}, we obtain
\begin{equation}
\mathsf{S}(H_{D\backslash F}H_F)-\mathsf{S}(\mathcal{R}H_{D\backslash F})+\mathsf{S}(H_F)\geq\mathsf{S}(\mathcal{S}).
\label{eq:temp1}
\end{equation}
Since any set of $k-1$ or lesser shares have no information about the secret, any set of partial shares from $k-1$ or lesser parties have no information about the secret as well. Since  $|D\backslash F|=k-1$, it follows $I(\mathcal{R}:H_{D\backslash F})=0$. This implies 
\begin{eqnarray*}
%I(\mathcal{R}:H_{D\backslash F})&=&0\\
\mathsf{S}(\mathcal{R}H_{D\backslash F})&=&\mathsf{S}(\mathcal{R})+\mathsf{S}(H_{D\backslash F})
\end{eqnarray*}
Substituting this in \eqref{eq:temp1} and because $\mathsf{S}(\mathcal{R})=\mathsf{S}(\mathcal{S})$, we obtain
\begin{eqnarray}
&&\mathsf{S}(H_{D\backslash F}H_F)+\mathsf{S}(H_F)-\mathsf{S}(H_{D\backslash F})\geq2\ \mathsf{S}(\mathcal{S}).
\end{eqnarray}
By subadditivity property, $\mathsf{S}(H_{D\backslash F}H_F)\leq\mathsf{S}(H_{D\backslash F})+\mathsf{S}(H_F)$. Therefore,
\begin{eqnarray}
2\mathsf{S}(H_F)\geq2\mathsf{S}(\mathcal{S})\nonumber
\\\mathsf{S}(H_F)\geq\mathsf{S}(\mathcal{S}).
\label{eq:temp2}
\end{eqnarray}
By subadditivity property,
\begin{equation*}
\mathsf{S}(H_F)\leq\sum_{j\in F}\mathsf{S}(H_{j,D})
\end{equation*}
Hence, from \eqref{eq:temp2}, we get
\begin{equation}
\sum_{j\in F}\mathsf{S}(H_{j,D})\geq\mathsf{S}(S)
\label{eq:temp3}
\end{equation}
This inequality holds for any of the $\binom{d}{d-k+1}$ possible choices for $F\subset D$. Now, sum the inequality \eqref{eq:temp3} over all these $F$ to get
\begin{eqnarray}
\sum_{\substack{F\subset D\\\text{s.t. }|F|=d-k+1}}\sum_{j\in F}\mathsf{S}(H_{j,D})\ &\geq&\sum_{\substack{F\subset D\\\text{s.t. }|F|=d-k+1}}\mathsf{S}(S)\nonumber
\\\binom{d-1}{d-k}\ \sum_{j\in D}\ \mathsf{S}(H_{j,D})\ &\geq&\binom{d}{d-k+1}\ \mathsf{S}(S)\nonumber
\\\sum_{j\in D}\ \mathsf{S}(H_{j,D})\ &\geq&\ \frac{d}{d-k+1}\ \mathsf{S}(S)
\end{eqnarray}
This inequality gives a bound on sum of entropies of the partial shares from $d$ shares to the combiner in terms of the entropy of the secret. This can be extended to a bound on dimensions of these systems as follows.
We know that the maximum value for entropy of a system is related to its dimension by \eqref{eq:max-entropy}. Thus, we obtain
\begin{eqnarray}
\sum_{j\in D}\log\text{dim}(H_{j,D})&\geq&\ \frac{d}{d-k+1}\ \mathsf{S}(S)\nonumber
\\\log\prod_{j\in D}\text{dim}(H_{j,D})&\geq&\ \frac{d}{d-k+1}\ \mathsf{S}(S).
\label{eq:temp5}
\end{eqnarray}
The state of $H_{j,D}$ lies in the same Hilbert space for any state of the secret $S$. Thus dim$(H_{j,D})$ remains the same for any arbitrary secret state and the bound \eqref{eq:temp5} is valid for all possible states of the secret. Consider the secret state with the density matrix
\begin{equation*}
\rho_S=\sum_{\ell=0}^{M-1}\frac{1}{M}\ketbra{\phi_\ell}{\phi_\ell}.
\end{equation*}
For this state, $\mathsf{S}(S)=\log M$. Hence, \eqref{eq:temp5} gives
\begin{eqnarray}
\sum_{j\in D}\log\text{dim}(H_{j,D})&\geq&\ \frac{d}{d-k+1}\ \log M\nonumber
\\\prod_{j\in D}\text{dim}(H_{j,D})&\geq& M^{d/(d-k+1)}.
\end{eqnarray}
This concludes the proof.
\end{proof}
%\addtocounter{lemma}{1}
The above result derived using the quantum information theoretic framework is same as Theorem \ref{th:co-lbounds}. This framework can be potentially generalized to bound communication costs and share sizes for quantum secret sharing schemes with non-threshold access structures as well.

\section{Conclusion}\label{sec:conc}
In this paper, we proposed new constructions for CE-QTS schemes.
We introduced the universal CE-QTS schemes and provided optimal constructions for CE-QTS and universal CE-QTS schemes using concatenation of ramp QSS schemes.
We also proposed another optimal construction for universal CE-QTS schemes based on Staircase codes.
We proved the bounds on communication cost during secret recovery in CE-QTS schemes.
Finally we developed a quantum information theoretic model to study CE-QTS schemes.
A natural direction for further study would be to extend these ideas to non-threshold access structures.
In the recent years there has been tremendous progress in experimental realization of quantum secret sharing schemes.
Hence, it would be also interesting to see if the dimension of the secret can be reduced while constructing CE-QTS schemes particularly for small number of parties. 

\appendices
\section{\titlemath{((k=3,n=5,d=5))} CE-QTS scheme based on Staircase codes}
\label{ap:ceqts-full-eg}
Consider the $((3,5,5))$ CE-QTS scheme from the construction based on Staircase codes given in \cite{senthoor19}. This scheme has the following parameters.
\begin{subequations}
\begin{gather}
q=7\\
m=3\\
w_1=w_2=\hdots=w_5=3\\
\text{CC}_n(3)=9,\ \text{CC}_n(5)=5.
\end{gather}
\end{subequations}
The encoding for the scheme is given by the mapping
\begin{eqnarray}
\label{eq:enc_qudits_3_5_ap}
\ket{\underline{s}}\mapsto\sum_{\underline{r}\in\F_7^6}\ket{c_{11}c_{12}c_{13}}&&\!\!\ket{c_{21}c_{22}c_{23}}\ket{c_{31}c_{32}c_{33}}
\\[-0.5cm]&&\ \ \ \ \ket{c_{41}c_{42}c_{43}}\ket{c_{51}c_{52}c_{53}}\nonumber
\end{eqnarray}
where $\underline{s}=(s_1,s_2,s_3)$ indicates a basis state of the quantum secret, $\underline{r}=(r_1,r_2,\hdots,r_6)$
and 
$c_{ij} $ is the $(i,j)$th entry of the matrix
\begin{equation}
C=VY.\nonumber
\end{equation}
Here the matrices $V$ and $Y$ are given by
\begin{equation*}
V=
\begin{bmatrix}
1&1&1&1&1\\1&2&4&1&2\\1&3&2&6&4\\1&4&2&1&4\\1&5&4&6&2
\end{bmatrix}
\text{and\ }
Y=
\left[
\begin{tabular}{ccc}
$s_1$&0&0\\$s_2$&0&0\\$s_3$&$r_1$&$r_2$\\$r_1$&$r_3$&$r_5$\\$r_2$&$r_4$&$r_6$
\end{tabular}
\right].
\end{equation*}
The encoded state in \eqref{eq:enc_qudits_3_5_ap} can also be written as,
\begin{align*}
\sum_{\underline{r}\in\mathbb{F}_7^6}
\begin{array}{l}
\ket{v_1(\underline{s},r_1,r_2)}\ket{v_1(0,0,r_1,r_3,r_4)}
\ket{v_1(0,0,r_2,r_5,r_6)}
\\\ket{v_2(\underline{s},r_1,r_2)}\ket{v_2(0,0,r_1,r_3,r_4)}
\ket{v_2(0,0,r_2,r_5,r_6)}
\\\ket{v_3(\underline{s},r_1,r_2)}\ket{v_3(0,0,r_1,r_3,r_4)}
\ket{v_3(0,0,r_2,r_5,r_6)}
\\\ket{v_4(\underline{s},r_1,r_2)}\ket{v_4(0,0,r_1,r_3,r_4)}
\ket{v_4(0,0,r_2,r_5,r_6)}
\\\ket{v_5(\underline{s},r_1,r_2)}\ket{v_5(0,0,r_1,r_3,r_4)}
\ket{v_5(0,0,r_2,r_5,r_6)}.
\end{array}
\end{align*}
$v_i()$ indicates the polynomial evaluation given by
\begin{eqnarray*}
v_i(f_1,f_2,f_3,f_4,f_5)&=&f_1+f_2.x_i+f_3.x_i^2+f_4.x_i^3+f_5.x_i^4
\end{eqnarray*}
and the expression $v_i(\underline{s},r_1,r_2)$ denotes $v_i(s_1,s_2,s_3,r_1,r_2)$. Here we have taken $x_i=i$ for $1\leq i\leq 5$.

When combiner requests $k=3$ parties, each party sends its complete share. 
When $d=5$, the combiner downloads the first qudit of each share from all the five parties.

\subsection{Secret recovery for \titlemath{d=5}}
When the combiner accesses all of the five parties, each party sends its first qudit. Thus CC$_n(5)=5$. The qudits with the combiner are given as
\begin{align*}
\sum_{\underline{r}\in\mathbb{F}_7^6}
\begin{array}{l}
\bl{\ket{v_1(\underline{s},r_1,r_2)}}\ket{v_1(0,r_1,r_2,r_3,r_4)}
\ket{v_1(0,0,r_3,r_5,r_6)}
\\\bl{\ket{v_2(\underline{s},r_1,r_2)}}\ket{v_2(0,r_1,r_2,r_3,r_4)}
\ket{v_2(0,0,r_3,r_5,r_6)}
\\\bl{\ket{v_3(\underline{s},r_1,r_2)}}\ket{v_3(0,r_1,r_2,r_3,r_4)}
\ket{v_3(0,0,r_3,r_5,r_6)}
\\\bl{\ket{v_4(\underline{s},r_1,r_2)}}\ket{v_4(0,r_1,r_2,r_3,r_4)}
\ket{v_4(0,0,r_3,r_5,r_6)}
\\\bl{\ket{v_5(\underline{s},r_1,r_2)}}\ket{v_5(0,r_1,r_2,r_3,r_4)}
\ket{v_5(0,0,r_3,r_5,r_6)}
\end{array}
\end{align*}
Applying the operation $U_{V^{-1}}$ on these five qudits, we obtain
\begin{align*}
\bl{\ket{\underline{s}}}\sum_{\underline{r}\in\mathbb{F}_7^6}
\begin{array}{l}
\ket{v_1(0,0,r_1,r_3,r_4)}
\ket{v_1(0,0,r_2,r_5,r_6)}
\\\ket{v_2(0,0,r_1,r_3,r_4)}
\ket{v_2(0,0,r_2,r_5,r_6)}
\\\ket{v_3(0,0,r_1,r_3,r_4)}
\ket{v_3(0,0,r_2,r_5,r_6)}
\\\bl{\ket{r_1}}\ket{v_4(0,0,r_1,r_3,r_4)}
\ket{v_4(0,0,r_2,r_5,r_6)}
\\\bl{\ket{r_2}}\ket{v_5(0,0,r_1,r_3,r_4)}
\ket{v_5(0,0,r_2,r_5,r_6)}
\end{array}
\end{align*}
Here, the three qudits from the first three parties contain the basis state of the secret. Also, these qudits are not entangled with any of the other qudits. Thus, any arbitrary superposition of the basis states can be recovered with the above step.
\subsection{Secret recovery for \titlemath{k=3}}
When the combiner accesses any three parties, the all three qudits from each of the three parties are transmitted to the combiner. Thus CC$_n(3)=9$. Assume that the combiner accesses the first three parties. Then the qudits with the combiner are given as
\begin{align*}
\sum_{\underline{r}\in\mathbb{F}_7^6}
\begin{array}{l}
\bl{\ket{v_1(\underline{s},r_1,r_2)}\ket{v_1(0,0,r_1,r_3,r_4)}
\ket{v_1(0,0,r_2,r_5,r_6)}}
\\\bl{\ket{v_2(\underline{s},r_1,r_2)}\ket{v_2(0,0,r_1,r_3,r_4)}
\ket{v_2(0,0,r_2,r_5,r_6)}}
\\\bl{\ket{v_3(\underline{s},r_1,r_2)}\ket{v_3(0,0,r_1,r_3,r_4)}
\ket{v_3(0,0,r_2,r_5,r_6)}}
\\\ket{v_4(\underline{s},r_1,r_2)}\ket{v_4(0,0,r_1,r_3,r_4)}
\ket{v_4(0,0,r_2,r_5,r_6)}
\\\ket{v_5(\underline{s},r_1,r_2)}\ket{v_5(0,0,r_1,r_3,r_4)}
\ket{v_5(0,0,r_2,r_5,r_6)}
\end{array}
\end{align*}
\begin{enumerate}
\item Apply the operation $U_{K_5}$ on the set of three second qudits and then applying $U_{K_5}$ on the set of third qudits where $K_5$ is the inverse of $V_{[3]}^{[3,5]}$, to obtain
\begin{eqnarray}
\hspace{-0.5cm}
\sum_{\underline{r}\in\mathbb{F}_7^6}
\begin{array}{l}
\bl{\ket{v_1(\underline{s},r_1,r_2)}\ket{r_1}
\ket{r_2}}
\\\bl{\ket{v_2(\underline{s},r_1,r_2)}\ket{r_3}
\ket{r_5}}
\\\bl{\ket{v_3(\underline{s},r_1,r_2)}\ket{r_4}
\ket{r_6}}
\\\ket{v_4(\underline{s},r_1,r_2)}\ket{v_4(0,0,r_1,r_3,r_4)}
\ket{v_4(0,0,r_2,r_5,r_6)}
\\\ket{v_5(\underline{s},r_1,r_2)}\ket{v_5(0,0,r_1,r_3,r_4)}
\ket{v_5(0,0,r_2,r_5,r_6)}
\end{array}\nonumber
\end{eqnarray}
\item Then, apply the following operators.
\begin{enumerate}
\item $L_6\ket{r_2}\ket{v_1(\underline{s},r_1,r_2)}$ to get $\ket{r_2}\ket{v_1(\underline{s},r_1,0)}$
\item $L_5\ket{r_2}\ket{v_2(\underline{s},r_1,r_2)}$ to get $\ket{r_2}\ket{v_2(\underline{s},r_1,0)}$
\item $L_3\ket{r_2}\ket{v_3(\underline{s},r_1,r_2)}$ to get $\ket{r_2}\ket{v_3(\underline{s},r_1,0)}$
\item $L_6\ket{r_1}\ket{v_1(\underline{s},r_1,0)}$ to get $\ket{r_1}\ket{v_1(\underline{s},0,0)}$
\item $L_6\ket{r_1}\ket{v_2(\underline{s},r_1,0)}$ to get $\ket{r_1}\ket{v_2(\underline{s},0,0)}$
\item $L_1\ket{r_1}\ket{v_3(\underline{s},r_1,0)}$ to get $\ket{r_1}\ket{v_3(\underline{s},0,0)}$
\end{enumerate}

Now, we obtain
\begin{eqnarray}
\hspace{-0.5cm}
\sum_{\underline{r}\in\mathbb{F}_7^6}
\begin{array}{l}
\bl{\ket{v_1(\underline{s},0,0)}\ket{r_1}
\ket{r_2}}
\\\bl{\ket{v_2(\underline{s},0,0)}\ket{r_3}
\ket{r_5}}
\\\bl{\ket{v_3(\underline{s},0,0)}\ket{r_4}
\ket{r_6}}
\\\ket{v_4(\underline{s},r_1,r_2)}\ket{v_4(0,0,r_1,r_3,r_4)}
\ket{v_4(0,0,r_2,r_5,r_6)}
\\\ket{v_5(\underline{s},r_1,r_2)}\ket{v_5(0,0,r_1,r_3,r_4)}
\ket{v_5(0,0,r_2,r_5,r_6)}
\end{array}\nonumber
\end{eqnarray}
\item Apply the operation $U_{K_6}$ on the set of three first qudits, where $K_6$ is the inverse of $V_{[3]}^{[3]}$ to obtain
\begin{eqnarray}
\hspace{-0.6cm}
\bl{\ket{\underline{s}}}
\sum_{\underline{r}\in\mathbb{F}_7^6}
\begin{array}{l}
\bl{\ket{r_1}\ket{r_2}}
\\\bl{\ket{r_3}\ket{r_5}}
\\\bl{\ket{r_4}\ket{r_6}}
\\\ket{v_4(\underline{s},r_1,r_2)}\ket{v_4(0,0,r_1,r_3,r_4)}
\ket{v_4(0,0,r_2,r_5,r_6)}
\\\ket{v_5(\underline{s},r_1,r_2)}\ket{v_5(0,0,r_1,r_3,r_4)}
\ket{v_5(0,0,r_2,r_5,r_6)}
\end{array}\nonumber
\end{eqnarray}
%\begin{eqnarray}
%\label{eq:entangled_secret_ap}
%\end{eqnarray}
Here, the three first qudits from the first three parties contain the basis state of the secret. For an equivalent classical secret sharing scheme, the secret recovery would have been complete at this stage. However these three qudits are still entangled with the first qudits from fourth and fifth parties. Thus, any arbitrary superposition of the basis states cannot be recovered at this stage for a quantum secret.
\item Apply the following operators to disentangle the basis state from the rest of the qudits.
\begin{enumerate}
\item $U_{K_7}$ on $\ket{r_2}\ket{r_5}\ket{r_6}$ to get
\\$\ket{r_2}\ket{v_4(0,0,r_2,r_5,r_6)}\ket{v_5(0,0,r_2,r_5,r_6)}$ where
\begin{equation}
K_7=\left[
\begin{tabular}{c}
1 0 0\\\hline
$V_{[4,5]}^{[3,5]}$
\end{tabular}
\right]\nonumber
\end{equation}
\item $U_{K_8}$ on $\ket{r_1}\ket{r_3}\ket{r_4}$ to get
\\$\ket{r_1}\ket{v_4(0,0,r_1,r_3,r_4)}\ket{v_5(0,0,r_1,r_3,r_4)}$ where
\begin{equation}
K_8=\left[
\begin{tabular}{c}
1 0 0\\\hline
$V_{[4,5]}^{[3,5]}$
\end{tabular}
\right]\nonumber
\end{equation}
\item $U_{K_9}$ on $\ket{s_1}\ket{s_2}\ket{s_3}\ket{r_1}\ket{r_2}$ to get
\\$\ket{s_1}\ket{s_2}\ket{s_3}\ket{v_4(\underline{s},r_1,r_2)}\ket{v_5(\underline{s},r_1,r_2)}$ where
\begin{equation}
K_9=\left[
\begin{tabular}{c}
1 0 0 0 0\\
0 1 0 0 0\\
0 0 1 0 0\\\hline
$V_{[4,5]}$
\end{tabular}
\right]\nonumber
\end{equation}
\end{enumerate}
Now, we obtain
\begin{eqnarray}
&&\hspace{-1cm}
\bl{\ket{\underline{s}}}
\sum_{\underline{r}\in\mathbb{F}_7^6}
\begin{array}{l}
\bl{\ket{v_4(\underline{s},r_1,r_2)}\ket{v_5(\underline{s},r_1,r_2)}}
\\\bl{\ket{v_4(0,0,r_1,r_3,r_4)}\ket{v_4(0,0,r_2,r_5,r_6)}}
\\\bl{\ket{v_5(0,0,r_1,r_3,r_4)}\ket{v_5(0,0,r_2,r_5,r_6)}}
\\\ket{v_4(\underline{s},r_1,r_2)}\ket{v_4(0,0,r_1,r_3,r_4)}\ket{v_4(0,0,r_2,r_5,r_6)}
\\\ket{v_5(\underline{s},r_1,r_2)}\ket{v_5(0,0,r_1,r_3,r_4)}\ket{v_5(0,0,r_2,r_5,r_6)}
\end{array}\nonumber
\\&&\hspace{-1cm}
=\bl{\ket{\underline{s}}}
\sum_{\substack{(r_1,r_2,r_3,r_4,\\r_5',r_6')\in\mathbb{F}_7^6}}
\begin{array}{l}
\bl{\ket{v_4(\underline{s},r_1,r_2)}\ket{v_5(\underline{s},r_1,r_2)}}
\\\bl{\ket{v_4(0,0,r_1,r_3,r_4)}\ket{r_5'}}
\\\bl{\ket{v_5(0,0,r_1,r_3,r_4)}\ket{r_6'}}
\\\ket{v_4(\underline{s},r_1,r_2)}\ket{v_4(0,0,r_1,r_3,r_4)}\ket{r_5'}
\\\ket{v_5(\underline{s},r_1,r_2)}\ket{v_5(0,0,r_1,r_3,r_4)}\ket{r_6'}
\end{array}
\label{eq:first-var-replaced}
\end{eqnarray}
\begin{eqnarray}
&&\hspace{-1cm}
=\bl{\ket{\underline{s}}}
\sum_{\substack{(r_1,r_2,r_3',r_4',\\r_5',r_6')\in\mathbb{F}_7^6}}
\begin{array}{l}
\bl{\ket{v_4(\underline{s},r_1,r_2)}\ket{v_5(\underline{s},r_1,r_2)}}
\\\bl{\ket{r_3'}\ket{r_5'}}
\\\bl{\ket{r_4'}\ket{r_6'}}
\\\ket{v_4(\underline{s},r_1,r_2)}\ket{r_3'}\ket{r_5'}
\\\ket{v_5(\underline{s},r_1,r_2)}\ket{r_4'}\ket{r_6'}
\end{array}\nonumber
\\&&\hspace{-1cm}
=\bl{\ket{\underline{s}}}
\sum_{\substack{(r_1',r_2',r_3',r_4',\\r_5',r_6')\in\mathbb{F}_7^6}}
\begin{array}{l}
\bl{\ket{r_1'}\ket{r_2'}}
\\\bl{\ket{r_3'}\ket{r_5'}}
\\\bl{\ket{r_4'}\ket{r_6'}}
\\\ket{r_1'}\ket{r_3'}\ket{r_5'}
\\\ket{r_2'}\ket{r_4'}\ket{r_6'}
\end{array}
%\label{eq:disentangled_secret_ap}
\nonumber
\end{eqnarray}
The variable change in \eqref{eq:first-var-replaced} is possible because independent of $r_1,r_2,r_3,r_4$, the subsystem
\begin{eqnarray}
\sum_{(r_5,r_6)\in\mathbb{F}_7^2}
\begin{array}{l}
\ket{v_4(0,0,r_2,r_5,r_6)}\ket{v_4(0,0,r_2,r_5,r_6)}
\\\ \ \ket{v_5(0,0,r_2,r_5,r_6)}\ket{v_5(0,0,r_2,r_5,r_6)}
\end{array}\nonumber
\end{eqnarray}
gives the uniform superposition
\begin{eqnarray}
\sum_{(r_5',r_6')\in\mathbb{F}_7^2}
\ket{r_5'}\ket{r_5'}\ket{r_6'}\ket{r_6'}.\nonumber
\end{eqnarray}
The succeeding expressions are derived similarly. Now, the secret is disentangled with the rest of the qudits. Thus, any arbitrary superposition of the basis states can be recovered with above steps for $d=3$.
\end{enumerate}

\section{Secret recovery for \titlemath{d=3} in the ((3,5,*)) universal CE-QTS scheme from Section \ref{s:example}}
\label{ap:univ-ceqts-d-3}
Consider the example of $((k=3,n=5,*))$ universal CE-QTS scheme in Section \ref{s:example} with the following parameters.
\begin{subequations}
\begin{gather}
q=11
\\m=3
\\w_1=w_2=\hdots=w_5=3
\\\text{CC}_5(3)=9,\ \text{CC}_5(4)=8,\ \text{CC}_5(5)=5.
\end{gather}
\end{subequations}
The encoding for the scheme is given by the following mapping
\begin{eqnarray}
\label{eq:enc_qudits_3_5_s_ap}
\ket{\underline{s}}\mapsto\sum_{\underline{r}\in\F_{11}^6}\ket{c_{11}c_{12}c_{13}}&&\!\!\ket{c_{21}c_{22}c_{23}}\ket{c_{31}c_{32}c_{33}}
\\[-0.5cm]&&\ \ \ \ \ket{c_{41}c_{42}c_{43}}\ket{c_{51}c_{52}c_{53}}\nonumber
\end{eqnarray}
where $\underline{s}=(s_1,s_2,s_3)$ indicates a basis state of the quantum secret, $\underline{r}=(r_1,r_2,\hdots,r_6)$
and 
$c_{ij} $ is the $(i,j)$th entry of the matrix
\begin{equation}
C=VY.\nonumber
\end{equation}
Here the matrices $V$ and $Y$ are defined as follows. 
\begin{eqnarray*}
V=
\begin{bmatrix}
9&3&4&6&1\\2&9&3&4&6\\8&2&9&3&4\\7&8&2&9&3\\5&7&8&2&9
\end{bmatrix}
\text{\ \ and\ \ \ }
Y=
\left[
\begin{tabular}{ccc}
$s_1$&0&0\\$s_2$&$r_1$&0\\$s_3$&$r_2$&$r_3$\\$r_1$&$r_3$&$r_5$\\$r_2$&$r_4$&$r_6$
\end{tabular}
\right].
\end{eqnarray*}
The encoded state in \eqref{eq:enc_qudits_3_5_s_ap} can also be written as,
\begin{align*}
\sum_{\underline{r}\in\mathbb{F}_{11}^6}
\begin{array}{l}
\ket{v_1(\underline{s},r_1,r_2)}\ket{v_1(0,r_1,r_2,r_3,r_4)}
\ket{v_1(0,0,r_3,r_5,r_6)}
\\\ket{v_2(\underline{s},r_1,r_2)}\ket{v_2(0,r_1,r_2,r_3,r_4)}
\ket{v_2(0,0,r_3,r_5,r_6)}
\\\ket{v_3(\underline{s},r_1,r_2)}\ket{v_3(0,r_1,r_2,r_3,r_4)}
\ket{v_3(0,0,r_3,r_5,r_6)}
\\\ket{v_4(\underline{s},r_1,r_2)}\ket{v_4(0,r_1,r_2,r_3,r_4)}
\ket{v_4(0,0,r_3,r_5,r_6)}
\\\ket{v_5(\underline{s},r_1,r_2)}\ket{v_5(0,r_1,r_2,r_3,r_4)}
\ket{v_5(0,0,r_3,r_5,r_6)}.
\end{array}
\end{align*}
$v_i()$ indicates the expression
\begin{eqnarray*}
v_i(f_1,f_2,f_3,f_4,f_5)=v_{i1}f_1+v_{i2}f_2+v_{i3}f_3+v_{i4}f_4+v_{i5}f_5
\end{eqnarray*}
where $v_{ij}=[V]_{ij}$ and the expression $v_i(\underline{s},r_1,r_2)$ denotes $v_i(s_1,s_2,s_3,r_1,r_2)$. The matrix $V$ is a Cauchy matrix.

When combiner requests $d=5$ parties, they send the first qudit from each of their shares. 
When $d=4$, the combiner downloads the first two qudits of each share of the four parties contacted. 
When $d=3$, the combiner downloads all three qudits of the share of the three parties contacted.
(For clarity, the qudits accessible to the combiner have been highlighted in blue in the description below.)

In the case when $d=3$, each of the three contacted parties sends all three qudits in its share.
\begin{align*}
\sum_{\underline{r}\in\mathbb{F}_{11}^6}
\begin{array}{l}
\bl{\ket{v_1(\underline{s},r_1,r_2)}\ket{v_1(0,r_1,r_2,r_3,r_4)}
\ket{v_1(0,0,r_3,r_5,r_6)}}
\\\bl{\ket{v_2(\underline{s},r_1,r_2)}\ket{v_2(0,r_1,r_2,r_3,r_4)}
\ket{v_2(0,0,r_3,r_5,r_6)}}
\\\bl{\ket{v_3(\underline{s},r_1,r_2)}\ket{v_3(0,r_1,r_2,r_3,r_4)}
\ket{v_3(0,0,r_3,r_5,r_6)}}
\\\ket{v_4(\underline{s},r_1,r_2)}\ket{v_4(0,r_1,r_2,r_3,r_4)}
\ket{v_4(0,0,r_3,r_5,r_6)}
\\\ket{v_5(\underline{s},r_1,r_2)}\ket{v_5(0,r_1,r_2,r_3,r_4)}
\ket{v_5(0,0,r_3,r_5,r_6)}
\end{array}
\end{align*}
\begin{enumerate}
\item Applying the operation $U_{K_5}$ on the set of three third qudits, where $K_5$ is the inverse of $V_{[3]}^{[3,5]}$, we obtain
\begin{eqnarray*}
\hspace{-0.5cm}\sum_{\underline{r}\in\mathbb{F}_{11}^6}
\begin{array}{l}
\bl{\ket{v_1(\underline{s},r_1,r_2)}\ket{v_1(0,r_1,r_2,r_3,r_4)}
\ket{r_3}}
\\\bl{\ket{v_2(\underline{s},r_1,r_2)}\ket{v_2(0,r_1,r_2,r_3,r_4)}
\ket{r_5}}
\\\bl{\ket{v_3(\underline{s},r_1,r_2)}\ket{v_3(0,r_1,r_2,r_3,r_4)}
\ket{r_6}}
\\\ket{v_4(\underline{s},r_1,r_2)}\ket{v_4(0,r_1,r_2,r_3,r_4)}
\ket{v_4(0,0,r_3,r_5,r_6)}
\\\ket{v_5(\underline{s},r_1,r_2)}\ket{v_5(0,r_1,r_2,r_3,r_4)}
\ket{v_5(0,0,r_3,r_5,r_6)}
\end{array}
\end{eqnarray*}
\item Then, on applying the operators $L_5\ket{r_3}\ket{v_1(0,r_1,r_2,r_3,r_4)}$, $L_7\ket{r_3}\ket{v_2(0,r_1,r_2,r_3,r_4)}$ and $L_8\ket{r_3}\ket{v_3(0,r_1,r_2,r_3,r_4)}$,
we obtain
\begin{eqnarray*}
\hspace{-0.5cm}
\sum_{\underline{r}\in\mathbb{F}_{11}^6}
\begin{array}{l}
\bl{\ket{v_1(\underline{s},r_1,r_2)}\ket{v_1(0,r_1,r_2,0,r_4)}
\ket{r_3}}
\\\bl{\ket{v_2(\underline{s},r_1,r_2)}\ket{v_2(0,r_1,r_2,0,r_4)}
\ket{r_5}}
\\\bl{\ket{v_3(\underline{s},r_1,r_2)}\ket{v_3(0,r_1,r_2,0,r_4)}
\ket{r_6}}
\\\ket{v_4(\underline{s},r_1,r_2)}\ket{v_4(0,r_1,r_2,r_3,r_4)}
\ket{v_4(0,0,r_3,r_5,r_6)}
\\\ket{v_5(\underline{s},r_1,r_2)}\ket{v_5(0,r_1,r_2,r_3,r_4)}
\ket{v_5(0,0,r_3,r_5,r_6)}
\end{array}
\end{eqnarray*}
\item Applying the operation $U_{K_6}$ on the set of three second qudits, where $K_6$ is the inverse of $V_{[3]}^{\{2,3,5\}}$, we obtain
\begin{eqnarray*}
\hspace{-0.5cm}\sum_{\underline{r}\in\mathbb{F}_{11}^6}
\begin{array}{l}
\bl{\ket{v_1(\underline{s},r_1,r_2)}\ket{r_1}
\ket{r_3}}
\\\bl{\ket{v_2(\underline{s},r_1,r_2)}\ket{r_2}
\ket{r_5}}
\\\bl{\ket{v_3(\underline{s},r_1,r_2)}\ket{r_4}
\ket{r_6}}
\\\ket{v_4(\underline{s},r_1,r_2)}\ket{v_4(0,r_1,r_2,r_3,r_4)}
\ket{v_4(0,0,r_3,r_5,r_6)}
\\\ket{v_5(\underline{s},r_1,r_2)}\ket{v_5(0,r_1,r_2,r_3,r_4)}
\ket{v_5(0,0,r_3,r_5,r_6)}
\end{array}
\end{eqnarray*}
\item Applying operation $U_{K_7}$ on the qudits $\ket{v_1(\underline{s},r_1,r_2)}$ $\ket{v_2(\underline{s},r_1,r_2)}\ket{v_3(\underline{s},r_1,r_2)}\ket{r_1}\ket{r_2}$ where
\begin{equation*}
K_7=\left[
\begin{tabular}{c}
$V_{[3]}$\\\hline
0 0 0 1 0\\
0 0 0 0 1
\end{tabular}
\right]^{-1}
\end{equation*}
we obtain
\begin{eqnarray*}
\hspace{-0.6cm}\bl{\ket{\underline{s}}}
\sum_{\underline{r}\in\mathbb{F}_{11}^6}
\begin{array}{l}
\bl{\ket{r_1}
\ket{r_3}}
\\\bl{\ket{r_2}
\ket{r_5}}
\\\bl{\ket{r_4}
\ket{r_6}}
\\\ket{v_4(\underline{s},r_1,r_2)}\ket{v_4(0,r_1,r_2,r_3,r_4)}
\ket{v_4(0,0,r_3,r_5,r_6)}
\\\ket{v_5(\underline{s},r_1,r_2)}\ket{v_5(0,r_1,r_2,r_3,r_4)}
\ket{v_5(0,0,r_3,r_5,r_6)}
\end{array}
\end{eqnarray*}
\item After recovering the basis state of the secret, we disentangle it from the rest of qudits by applying suitable operators as follows.
\begin{enumerate}
\item Apply $U_{K_8}$ on $\ket{r_3}\ket{r_5}\ket{r_6}$ to get
\\$\ket{r_3}$ $\ket{v_4(0,0,r_3,r_5,r_6)}\ket{v_5(0,0,r_3,r_5,r_6)}$ where
\begin{equation*}
K_8=\left[
\begin{tabular}{c}
1 0 0\\\hline
$V_{[4,5]}^{[3,5]}$
\end{tabular}
\right].
\end{equation*}
\item Apply $U_{K_9}$ on $\ket{r_1}\ket{r_2}\ket{r_3}\ket{r_4}$ to get \\$\ket{r_1}\ket{r_2}\ket{v_4(0,r_1,r_2,r_3,r_4)}\ket{v_5(0,r_1,r_2,r_3,r_4)}$ where
\begin{equation*}
K_9=\left[
\begin{tabular}{c}
1 0 0 0\\
0 1 0 0\\\hline
$V_{[4,5]}^{[2,5]}$
\end{tabular}
\right].
\end{equation*}
\item Apply $U_{K_{10}}$ on $\ket{s_1}\ket{s_2}\ket{s_3}\ket{r_1}\ket{r_2}$ to get
\\$\ket{s_1}\ket{s_2}\ket{s_3}\ket{v_4(\underline{s},r_1,r_2)}\ket{v_5(\underline{s},r_1,r_2)}$ where
\begin{equation*}
K_{10}=\left[
\begin{tabular}{c}
1 0 0 0 0\\
0 1 0 0 0\\
0 0 1 0 0\\\hline
$V_{[4,5]}$
\end{tabular}
\right].
\end{equation*}
\end{enumerate}
Now, we obtain
\begin{eqnarray*}
&&\hspace{-1cm}
\bl{\ket{\underline{s}}}
\sum_{\underline{r}\in\mathbb{F}_{11}^6}
\begin{array}{l}
\bl{\ket{v_4(\underline{s},r_1,r_2)}\ket{v_4(0,r_1,r_2,r_3,r_4)}}
\\\bl{\ket{v_5(\underline{s},r_1,r_2)}\ket{v_4(0,0,r_3,r_5,r_6)}}
\\\bl{\ket{v_5(0,r_1,r_2,r_3,r_4)}\ket{v_5(0,0,r_3,r_5,r_6)}}
\\\ket{v_4(\underline{s},r_1,r_2)}\ket{v_4(0,r_1,r_2,r_3,r_4)}
\ket{v_4(0,0,r_3,r_5,r_6)}
\\\ket{v_5(\underline{s},r_1,r_2)}\ket{v_5(0,r_1,r_2,r_3,r_4)}
\ket{v_5(0,0,r_3,r_5,r_6)}
\end{array}\nonumber
\\&&\hspace{-1cm}
=\bl{\ket{\underline{s}}}
\sum_{\underline{r}''\in\mathbb{F}_{11}^6}
\begin{array}{l}
\bl{\ket{r_1''}\ket{r_3''}}
\\\bl{\ket{r_2''}\ket{r_5''}}
\\\bl{\ket{r_4'')}\ket{r_6''}}
\\\ket{r_1''}\ket{r_3''}\ket{r_5''}
\\\ket{r_2''}\ket{r_4''}\ket{r_6''}
\end{array}
\end{eqnarray*}
where $\underline{r}''=(r_1'',r_2'',r_3'',r_4'',r_5'',r_6'')$. Now, the secret is disentangled with the rest of the qudits.
\end{enumerate}
Thus, any arbitrary superposition of the basis states can be recovered with above steps for $d=3$.

\medskip
\noindent 
{\em Acknowledgment.} This research was supported by the Department of Science and Technology, Govt. of India, under grant number DST/ICPS/QuST/Theme-3/2019/Q59.

\bibliographystyle{IEEEtran}
%\bibliography{refs-ceqts-trxns}

\end{document}